\documentclass[article]{IEEEtran}

\usepackage{graphicx,colordvi,psfrag}
\usepackage{amsmath,amssymb}
\usepackage{epstopdf}
\usepackage[caption=false]{subfig}
\usepackage{epsfig,cite}
\usepackage{calc,pstricks, pgf, xcolor}
\usepackage{nicefrac}
\usepackage{enumerate}
\usepackage{bm}
\usepackage{setspace}

\def\argmax{\operatornamewithlimits{arg\,max}}

\newtheorem{theorem}{Theorem}
\newtheorem{corollary}{Corollary}
\newtheorem{remark}{Remark}
\newtheorem{example}{Example}
\newtheorem{definition}{Definition}
\newtheorem{lemma}{Lemma}
\newenvironment{proof}[1][Proof]{\noindent\textbf{#1.} }{\ \rule{0.5em}{0.5em}}

\newcommand{\by}{\mathbf{y}}

\newcommand{\bx}{\mathbf{x}}

\newcommand{\bz}{\mathbf{z}}
\newcommand{\bZ}{\mathbf{Z}}
\newcommand{\bh}{\mathbf{h}}

\newcommand{\bA}{\mathbf{A}}
\newcommand{\ba}{\mathbf{a}}

\newcommand{\bd}{\mathbf{d}}
\newcommand{\bB}{\mathbf{B}}

\newcommand{\bu}{\mathbf{u}}
\newcommand{\bV}{\mathbf{V}}
\newcommand{\bv}{\mathbf{v}}
\newcommand{\CV}{\mathcal{V}}

\newcommand{\bt}{\mathbf{t}}
\newcommand{\bT}{\mathbf{T}}

\newcommand{\bR}{\mathbf{R}}
\newcommand{\bI}{\mathbf{I}}

\newcommand{\bL}{\mathbf{L}}
\newcommand{\bg}{\mathbf{g}}
\newcommand{\bG}{\mathbf{G}}
\newcommand{\bs}{\mathbf{s}}
\newcommand{\bS}{\mathbf{S}}

\newcommand{\bF}{\mathbf{F}}
\newcommand{\Zp}{\mathbb{Z}_p}
\newcommand{\ZZ}{\mathbb{Z}}
\newcommand{\RR}{\mathbb{R}}

\newcommand{\Tsnr}{\mathsf{SNR}}
\newcommand{\Tinr}{\mathsf{INR}}
\newcommand{\Ql}{Q_{\Lambda}}
\newcommand{\Mod}{\bmod\Lambda}

\newcommand{\Var}{\mathrm{Var}}
\newcommand{\csym}{C_{\text{SYM}}}
\newcommand{\rsym}{R_{\text{SYM}}}
\newcommand{\diag}{\mathop{\mathrm{diag}}}
\newcommand{\Span}{\mathop{\mathrm{span}}}
\newcommand{\rank}{\mathop{\mathrm{rank}}}
\newcommand{\sign}{\mathop{\mathrm{sign}}}
\newcommand{\Vol}{\mathrm{Vol}}
\newcommand{\map}{\theta}

\newrgbcolor{BoxRed}{1 .5 .5}
\newrgbcolor{BoxBlue}{.9 .9 1}
\newrgbcolor{LineBlue}{.2 .6 .9}

\begin{document}

\title{The Approximate Sum Capacity of the Symmetric Gaussian $K$-User Interference Channel}

\author{Or Ordentlich, Uri Erez, and Bobak Nazer
\thanks{The work of O. Ordentlich was supported by the Adams Fellowship Program of the Israel Academy of Sciences and Humanities, a fellowship from The Yitzhak and Chaya Weinstein Research Institute for Signal Processing at Tel Aviv University, and the Feder Family Award. The work of U. Erez was supported in part by the Israel Science Foundation under Grant No. 1557/10. The work of B. Nazer was supported by the National Science Foundation under Grant CCF-1253918.}
\thanks{O. Ordentlich and U. Erez are with Tel Aviv University, Tel Aviv, Israel (email: ordent,uri@eng.tau.ac.il). B. Nazer is with the Department of Electrical and Computer Engineering, Boston University, Boston, MA 02215, USA (email: bobak@bu.edu)}
}


\maketitle

\begin{abstract}
Interference alignment has emerged as a powerful tool in the analysis of multi-user networks. Despite considerable recent progress, the capacity region of the Gaussian $K$-user interference channel is still unknown in general, in part due to the challenges associated with alignment on the signal scale using lattice codes. This paper develops a new framework for lattice interference alignment, based on the compute-and-forward approach. Within this framework, each receiver decodes by first recovering two or more linear combinations of the transmitted codewords with integer-valued coefficients and then solving these linear combinations for its desired codeword. For the special case of symmetric channel gains, this framework is used to derive the approximate sum capacity of the Gaussian interference channel, up to an explicitly defined outage set of the channel gains. The key contributions are the capacity lower bounds for the weak through strong interference regimes, where each receiver should jointly decode its own codeword along with part of the interfering codewords. As part of the analysis, it is shown that decoding $K$ linear combinations of the codewords can approach the sum capacity of the $K$-user Gaussian multiple-access channel up to a gap of no more than $\frac{K}{2}\log{K}$ bits.
\end{abstract}

\section{Introduction}
\label{sec:Into}

Handling interference efficiently is a major challenge in multi-user wireless communication. Recently, it has become clear that this challenge can sometimes be overcome via \textit{interference alignment} \cite{mmk08,cj08}. For instance, consider the $K$-user Gaussian interference channel, where $K$ transmitter-receiver pairs wish to communicate simultaneously. Through the use of clever encoding strategies, it is possible to align the transmitted signals so that each receiver only observes its desired signal along with a single effective interferer. As a result, each user can achieve roughly half the rate that would be available were there no interference whatsoever, i.e., $K/2$ degrees-of-freedom (DoF) are available. However, many schemes, such as the Cadambe-Jafar framework \cite{cj08} and ergodic interference alignment \cite{ngjv11IT}, require a large number of independent channel realizations to achieve near-perfect alignment. In certain settings, this level of channel diversity may not be attainable; ideally, we would like to achieve alignment over a single channel realization.

The capacity region of the (static) Gaussian $K$-user interference channel \cite{carleial78} is unknown in general, although significant progress has been made recently, in part due to the discovery of interference alignment and the shift from exact capacity results to capacity approximations \cite{etw08,adt11,bpt10}. It has been shown by Motahari \textit{et al.} that $K/2$ DoF are achievable for almost all channel realizations \cite{mgmk09} but it is an open question as to whether this result translates to real gains outside of the very high signal-to-noise ratio (SNR) regime. One promising direction is the use of lattice codes \cite{zse02,ez04,zamir09}, as they can enable alignment on the signal scale. By taking advantage of the fact that the sum of lattice codewords is itself a lattice codeword, a receiver can treat several users as one effective user, thereby reducing the number of effective interferers. A compelling example of this approach is the derivation of the approximate capacity of the many-to-one interference channel by Bresler, Parekh, and Tse \cite{bpt10}. For fully connected channels, much less is known, owing to the difficulty of choosing lattices that simultaneously align at several receivers.

In some cases, focusing on the special case of symmetric channel gains has yielded important insights. For instance, in the two-user case, Etkin, Tse, and Wang \cite{etw08} used the symmetric interference channel to develop the notion of generalized degrees-of-freedom. This in turn revealed five operating regimes, based on relative interference strength:
\begin{itemize}
\item \textit{Noisy:} Each receiver treats interference as noise, which is optimal for sufficiently weak interference \cite{mk09,skc09,av09IT}.
\item \textit{Weak and Moderately Weak:} Each transmitter sends a public and a private codeword following the scheme of Han and Kobayashi \cite{hk81}. Each receiver jointly decodes both public codewords and its desired private codeword while treating the interfering private codeword as noise.
\item \textit{Strong:} Each receiver jointly decodes both users' codewords. This regime and its capacity was discovered by Sato \cite{sato81} as well as Han and Kobayashi \cite{hk81}.
\item \textit{Very Strong:} Each receiver decodes and subtracts the interference before recovering its desired codeword. This regime and its capacity was discovered by Carleial \cite{carleial75}.
\end{itemize} Using these regimes as a guideline, they were able to approximate the capacity region to within half a bit per real channel dimension.

In this paper, we focus on the special case of the symmetric (real) Gaussian $K$-user interference channel. Each receiver observes
\begin{equation}
\by_k = \bx_k + g \sum_{\ell \neq k} \bx_\ell + \bz_k \label{e:chanmodel}
\end{equation} where $\bx_k$ is the codeword sent by the $k$th transmitter, $g$ is the cross-channel gain, and $\bz_k$ is additive white Gaussian noise. Building on the compute-and-forward strategy~\cite{ng11IT}, we propose a framework for lattice-based interference alignment whose performance can be evaluated both numerically and analytically at any SNR. Within our framework, each receiver first decodes integer linear combinations of the codewords and only afterwards solves these for its desired codeword. As we will argue, this choice of receiver architecture allows us to circumvent some of the difficulties encountered in the analysis of a direct decoding strategy. Below, we summarize the main technical contributions of the paper in the context of prior work.

\subsection{Paper Overview} \label{s:overview}

One of the appealing properties of the symmetric Gaussian interference channel is that, if each transmitter draws its codeword $\bx_k$ from the same lattice codebook, the sum of the $K-1$ interfering codewords at each receiver $\sum_{\ell \neq k} \bx_\ell$ will align into a single effective codeword. This is due to the fact that lattices are closed under addition, i.e., the sum of any lattice codewords is itself a lattice codeword. The difficulty is that, depending on the value of the cross-channel gain $g$, the desired codeword may also align with the interference, since it is drawn from the same lattice codebook. The achievable rate is thus closely linked to the behavior of signal scale alignment, which makes this channel an ideal setting to gain a deeper understanding of this phenomenon at finite SNR. In Section~\ref{s:probstate}, we provide a formal problem statement.

When $|g|$ is sufficiently large, it is easy for the receiver to distinguish its desired codeword from the aligned interfering codewords. Specifically, in the very strong regime ($|g| > \sqrt{\Tsnr}$), the sum of the interfering codewords acts as the cloud center from a classical superposition codebook~\cite{cover72} and the desired codebook acts as the cloud. Thus, as proposed by Sridharan~\textit{et al.}~\cite{sjvj08}, the receiver can employ a successive cancellation strategy: first decode the sum of the interference $\sum_{\ell \neq k} \bx_\ell$, then subtract it from its channel observation $\by_k$, and finally decode $\bx_k$ from the resulting interference-free effective channel. We review this approach within the context of our framework in Section~\ref{s:vsregime}.

\begin{figure}[h]
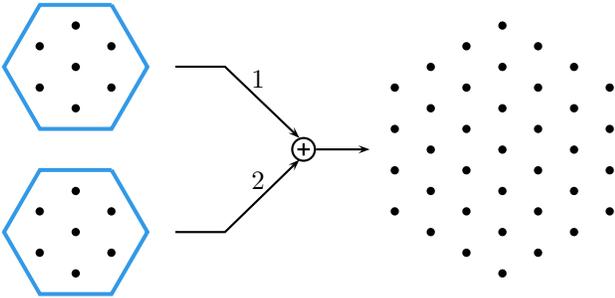

\begin{center}
\psset{unit=0.55mm}
\begin{pspicture}(8,-7)(160,67)

\rput(28,10){
\input{nestedlatticefig}
}
\rput(28,50){
\input{nestedlatticefig}
}

\rput(33,0){
\psline{->}(18,50)(30,50)(48,32.75)
\rput(38,47){$1$}

\psline{->}(18,10)(30,10)(48,27.5)
\rput(38,22.5){$2$}

\pscircle(49,30){3} \psline{-}(47.5,30)(50.5,30)
\psline{-}(49,28.5)(49,31.5)
\psline{->}(52,30)(65,30) 
}

\rput(130,30){
\input{sumlattice2fig}
}

\end{pspicture}
\end{center}
\caption{Two transmitters employ the same $7$-symbol lattice code over the channel $\mathbf{x}_1 + 2 \mathbf{x}_2$. The effective constellation seen by the receiver contains only $37$ points, which means that the receiver cannot always uniquely identify which pair of symbols was transmitted.}
\label{f:latticerational2}
\end{figure}

\begin{figure}[h]
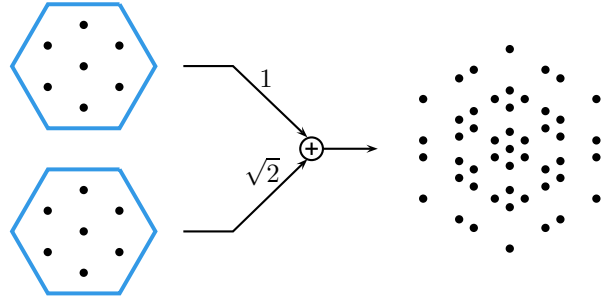

\begin{center}
\psset{unit=0.55mm}
\begin{pspicture}(8,-7)(160,67)

\rput(28,10){
\input{nestedlatticefig}
}
\rput(28,50){
\input{nestedlatticefig}
}

\rput(33,0){
\psline{->}(18,50)(30,50)(48,32.75)
\rput(38,47){$1$}

\psline{->}(18,10)(30,10)(48,27.5)
\rput(37,24.5){$\sqrt{2}$}

\pscircle(49,30){3} \psline{-}(47.5,30)(50.5,30)
\psline{-}(49,28.5)(49,31.5)
\psline{->}(52,30)(65,30) 
}

\rput(131,30){
\input{sqrt2latticefig}
}


\end{pspicture}
\end{center}
\caption{Two transmitters employ the same $7$-symbol lattice code over the channel $\mathbf{x}_1 + \sqrt{2} \mathbf{x}_2$. The effective constellation seen by the receiver consists of $49$ points, which enables the receiver to determine which pair of symbols was transmitted.}
\label{f:latticeirrational}
\end{figure}

As the magnitude of $g$ decreases below $\sqrt{\Tsnr}$, the codebooks corresponding to the desired codeword and the aligned interference will start to overlap from the receiver's perspective. For certain values of $g$, $\bx_k$ and $\sum_{\ell \neq k} \bx_\ell$ will align, which in turn significantly reduces the achievable rates. For example, in Figure~\ref{f:latticerational2}, we illustrate the effective codebook corresponding to the linear combination $\bx_1 + 2 \bx_2$ where $\bx_1$ and $\bx_2$ are drawn from the same lattice codebook. There are only $37$ points in this effective codebook, meaning that it is not always possible to uniquely determine which of the $49$ possible pairs of codewords was transmitted, regardless of the $\Tsnr$. However, for the linear combination $\bx_1 + \sqrt{2} \bx_2$ shown in Figure~\ref{f:latticeirrational}, there are $49$ points in the effective codebook, each corresponding to a unique codeword pair, even though the interference strength has decreased.

Thus, while employing the same lattice codebook at each transmitter aligns the interference at every receiver, it sometimes has the unintended effect of aligning the desired signal as well. When this occurs, the rate must be reduced until the desired codewords can be uniquely identified. We now summarize several recent papers that have aimed to quantify this effect. Etkin and Ordentlich~\cite{eo09} showed that, for the Gaussian $K$-user interference channel, the DoF is strictly less than $K/2$ if all channel gains are rational. They also demonstrated, using a scalar lattice codebook, that if the diagonal elements are irrational algebraic numbers and the off-diagonals are rational, $K/2$ DoF is achievable. Subsequently, Motahari \textit{et al.}~\cite{mgmk09} proposed the ``real interference alignment'' framework. In particular, they argued that scalar lattice codewords can be uniquely identified from a linear combination (in the high SNR limit) provided that the coefficients are rationally independent.\footnote{The coefficients $h_1,\ldots,h_K \in \mathbb{R}$ are said to be \textit{rationally independent} if there is no non-trivial choice of integers $q_1,\ldots,q_K$ such that $q_1 h_1 + \cdots + q_K h_K = 0$.} Using this framework, they demonstrated that, for the Gaussian $K$-user interference channel, $K/2$ DoF is achievable for almost all channel matrices by embedding the asymptotic alignment framework of~\cite{cj08} into a single dimension. This result was generalized by Wu \textit{et al.} using R\'enyi's information dimension~\cite{wsv11}.

For finite SNRs,~\cite{oe13} derived lower bounds on the achievable symmetric rate for a two-user multiple-access channel $\bx_1 + g \bx_2 + \bz$ where each user employs the same linear code over $\mathbb{Z}_p$ for some prime $p$. The sensitivity of the bounds to the rationality of $g$ at different SNRs was investigated, and the bounds were used to obtain achievable rate regions for Gaussian $K$-user interference channels with integer-valued off-diagonal channel gains. For the two-user Gaussian X channel\footnote{In the X channel scenario, each transmitter has an independent message for each receiver.}, Niesen and Maddah-Ali~\cite{nm13} approximated the sum capacity via an ``outage set'' characterization. Their coding scheme is guided by a variation on the deterministic model~\cite{adt11} and consists of a scalar lattice constellation combined with a random i.i.d. outer code. From one perspective, for any $c > 0$, their scheme approximates the sum capacity to within a constant gap of $c + 66$ bits up to an outage set of channel matrices of measure roughly $2^{-c/2}$.

\subsubsection{Novel Coding Strategies} The prior work described above attempts to directly bound the minimum distance in the effective codebook that results from the linear combination of the transmitters' lattice codebooks. This is a challenging task, even for scalar lattices, and limits the analytical and numerical results to relatively high SNRs. In this paper, we take an alternative approach: we lower bound the achievable rate by the rate required to decode enough integer linear combinations to reveal the desired messages. For instance, in the strong regime $1 \leq g \leq \sqrt{\Tsnr}$, each receiver \textit{first decodes two linear combinations} of the form
$$ a_{11} \bx_k + a_{12} \sum_{\ell \neq k} \bx_\ell \qquad \qquad   a_{21} \bx_k + a_{22} \sum_{\ell \neq k} \bx_\ell \ , $$ where $a_{11},a_{12},a_{21},$ and $a_{22}$ are integer-valued coefficients. If the vectors $\mathbf{a}_1 = [a_{11}~a_{12}]^T$ and $\mathbf{a}_2 = [a_{21}~a_{22}]^T$ are linearly independent, then each receiver can solve for its desired codeword $\bx_k$. The rates at which these linear combinations can be decoded can be determined directly via the compute-and-forward framework~\cite{ng11IT}, which we review in Section~\ref{s:cf}. Since this framework employs high-dimensional nested lattice codes that can approach the point-to-point AWGN capacity, we can obtain analytical and numerical results for any finite SNR.

\begin{figure}[htb]
\includegraphics[width=1 \columnwidth]{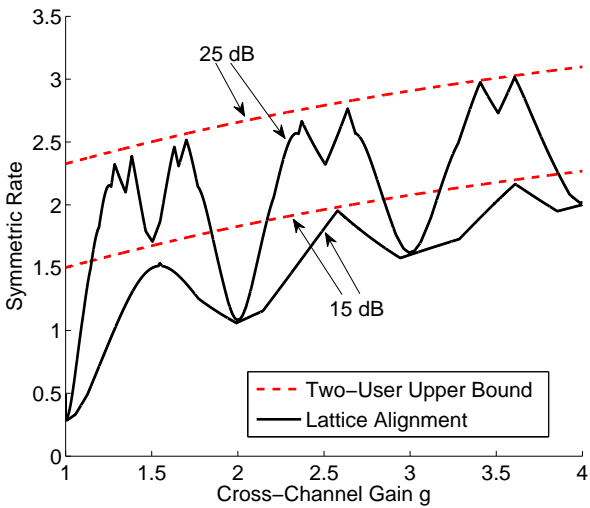}
\caption{Achievable symmetric rate for the symmetric Gaussian $3$-user interference channel from Theorem~\ref{thm:SymICnoLayeres}.}
\label{fig:badrationals}
\end{figure}

In Figure~\ref{fig:badrationals}, we have plotted the symmetric rate of this scheme (Theorem~\ref{thm:SymICnoLayeres}) at SNRs $15$ and $25$dB with respect to the cross-gain $g$ for the symmetric Gaussian $3$-user interference channel. Alongside, we have plotted the symmetric rate for the two-user upper bound described in Section~\ref{s:upper}, i.e., the rate that would be achievable if each receiver only encountered a single interferer. At $15$dB, it is clear that the desired codeword aligns with the interference only at integer-valued cross-gains. At $25$dB, alignment also occurs at $g = 3/2$, $5/2$, and $7/2$, i.e., rationals with denominator $2$. In other words, the number of channel gains where the rate saturates depends on the SNR.

We also propose a lattice version of the Han-Kobayashi scheme~\cite{hk81} for the weak and moderately weak regimes: each transmitter splits its information into a public lattice codeword $\bx_{k1}$ and a private lattice codeword $\bx_{k2}$. Each receiver recovers its desired information by first decoding three linear combinations of the form
$$ a_{m1} \bx_{k1} + a_{m2} \bx_{k2} + a_{m3} \sum_{\ell \neq k} \bx_{\ell 1}~~~~~m=1,2,3$$ for integer-valued coefficients $a_{m1},a_{m2},a_{m3}$ that suffice to solve for the desired public codeword $\bx_{k1}$, the desired private codeword $\bx_{k2}$, and the sum of the public interfering codewords $\sum_{\ell \neq k} \bx_{\ell 1}$. (The private interfering codewords are treated as noise.)

Within the standard compute-and-forward framework, the rate of each codeword should be set according to the lowest computation rate across all desired linear combinations. In Section~\ref{s:mac}, we propose an \textit{algebraic successive cancellation} decoding strategy that can achieve higher rates. Consider a single receiver that decodes $K$ linearly independent combinations of $K$ lattice codewords in a given order. Each linear combination is associated with a certain computation rate, which we set as the rate of one of the codewords. After decoding each linear combination, the receiver can cancel out the effect of one codeword from its channel observation to reduce the effective rate. As we show in Theorem~\ref{thm:SymICHK}, for the lattice Han-Kobayashi scheme, this allows each user to attain the sum of the second and third highest computation rates (as opposed to twice the third highest).

Overall, these two lattice strategies, when combined with successive cancellation for the very strong regime and treating interference as noise for the noisy regime, yield an achievable rate region for the symmetric Gaussian $K$-user interference channel. To evaluate this rate region, we only need to optimize over the integer coefficients of the linear combinations. See Section~\ref{s:numerical} for a discussion on how the space of integer-coefficients can be explored numerically. In Figure~\ref{f:symICrates}, we have plotted the resulting lower bound on the symmetric capacity along with the two-user upper bound from Section~\ref{s:upper}.

\subsubsection{Analytical Bounds} We also develop new tools for deriving closed-form lower bounds for the rate achievable via lattice alignment. These tools and specifically the compute-and-forward transform, derived in Section~\ref{sub:coftransform}, may be of independent interest. Consider again $K$ transmitted codewords and a receiver that decodes $K$ linear combinations according to the $K$ highest computation rates with linearly independent coefficient vectors. While the computation rate for each of these $K$ combinations is very sensitive to the exact values of the channel gains, the \textit{sum of the computation rates is equal to the multiple-access sum capacity up to a constant gap} that is independent of the channel gains and the SNR as we show in Theorem~\ref{thm:SumRate}. See Figure~\ref{fig:SumRateVsh40dB} for a plot of this behavior for $K = 2$. That is, lattice-based multiple-access can operate near the boundary of the capacity region. We also argue in Section~\ref{s:dof} that the degrees-of-freedom associated to each of these $K$ linear combinations is $1/K$ for almost all channel gains.

\begin{figure}[htb]
\includegraphics[width=1 \columnwidth]{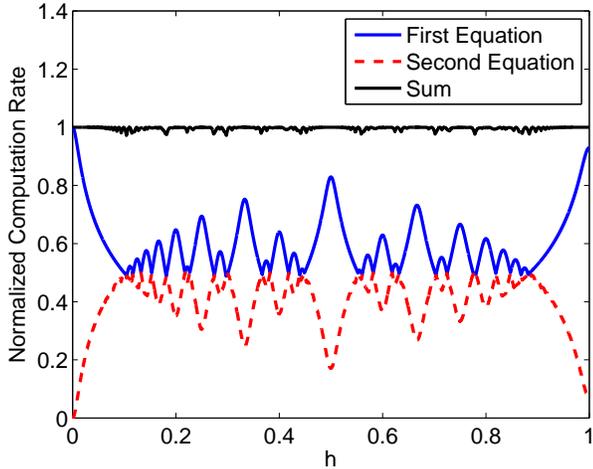}
\caption{Computation rates for the best two linearly independent integer linear combinations vs. $h$ for the channel $\by=\bx_1+h\bx_2+\bz$ at SNR=$40$dB. The sum of these computation rates is nearly equal to the multiple-access sum capacity. All rates are normalized by this sum capacity $1/2\log(1+ (1 + h^2)\Tsnr)$.}
\label{fig:SumRateVsh40dB}
\end{figure}

Interestingly, this sum capacity lower bound is very helpful in deriving closed-form lower bounds. For instance, in the strong regime, each user attains the rate associated with the second best linear combination. Thus, to bound the achievable rate, we should obtain an upper bound on the second best integer approximation of the real-valued channel gain $g$. Instead of attacking this problem directly, we instead develop a lower bound on the best integer approximation of $g$ and combine this with the sum capacity lower bound to obtain our upper bound. More details are given in Section~\ref{s:strongregime} and similar bounds are developed for the moderately weak and weak regime in Sections~\ref{s:moderateregime} and~\ref{s:weakregime}, respectively.

When compared with the two-user upper bound, these lower bounds yield an approximation of the sum capacity in all regimes that we summarize in Theorem~\ref{thm:approxCapacity}. As in~\cite{nm13}, our approximation is stated in terms of outage sets, i.e., for a given constant gap, we exclude a certain measure of channel gains. This outage set can be understood in terms of the quality of the best integer approximation of $g$, and is characterized as part of the analysis in Section~\ref{s:lower}.

\subsection{Related Work}\label{s:relatedwork}

Interference alignment has generated a great deal of excitement, due to the promise of higher throughputs in wireless networks \cite{mmk08,cj08} as well as other applications, including coding for distributed storage \cite{drws11}. See the recent monograph by Jafar for a comprehensive survey \cite{jafar11}. Of particular note is a series of recent papers that delineate the degrees-of-freedom limits of linear beamforming strategies for alignment over a finite number of channel realizations \cite{bt09,wgj11b}. Beamforming strategies can only approach perfect alignment asymptotically, whereas lattice-based schemes can achieve $K/2$ degrees-of-freedom over a single channel realization \cite{mgmk09}. However, lattice-based alignment at finite SNR has to date been limited to special cases, such as symmetric \cite{sjvj08,sjvjs08,oe13}, integer \cite{jv12}, and many-to-one interference channels \cite{bpt10,sb11}. Capacity approximations are also available for one-to-many \cite{bpt10} and cyclic interference channels \cite{zy11}, although these coding schemes do not employ alignment. Bandemer and El Gamal have recently proposed a class of three-user deterministic channels where the interfering signals are passed through a function on their way to the receiver, which, in a certain sense, models interference alignment \cite{be11}. They develop a new rate region based on interference decoding for this model.
%


Nested lattice codes have been thoroughly studied as a framework for efficient source and channel coding with side information \cite{zse02,ez04,esz05}. Recently, it has become clear that the inherent linear structure of lattices can enable many interesting new schemes, including distributed dirty paper coding \cite{pzek11}, distributed source coding of linear functions \cite{kp09,wagner11,mt10ISIT}, distributed antenna systems \cite{nsgs09,hc13IT}, and physical-layer network coding \cite{ng11PIEEE,wnps10,ncl10,ng11IT,fsk11}, to name a few. See~\cite{ramibook} for a comprehensive survey. The origins of these schemes can be traced to the work of K\"orner and Marton \cite{km79}, who showed that linear binning is optimal for the distributed compression of the parity of a doubly symmetric binary source.


\section{Symmetric Gaussian $K$-User Interference Channel}\label{s:kuseric}

\subsection{Problem Statement} \label{s:probstate}

We begin with some notational conventions. We will denote column vectors with boldface lowercase letters and matrices with boldface uppercase letters. For instance, $\ba \in \ZZ^K$ and $\bA \in \ZZ^{K \times K}$. Let $\| \ba \| = \sqrt{\sum_{k=1}^K a_k^2}$ denote the $\ell_2$-norm of the vector $\ba$. Also, let $\mathbf{0}$ denote the zero vector and $\bI_{K \times K}$ denote the identity matrix of size $K$. We use $\lfloor\cdot\rceil$ to denote rounding to the nearest integer, $\lfloor\cdot\rfloor$ to denote the floor operation and $\lceil\cdot\rceil$ for the ceiling operation. In general, the letters $a$ and $b$ are used in this paper whenever the variables they describe are integer valued.
All logarithms are to base $2$. We also occasionally use the notation $\log^+(x)\triangleq\max(0,\log(x))$. All measures in this paper are Lebesgue measures.

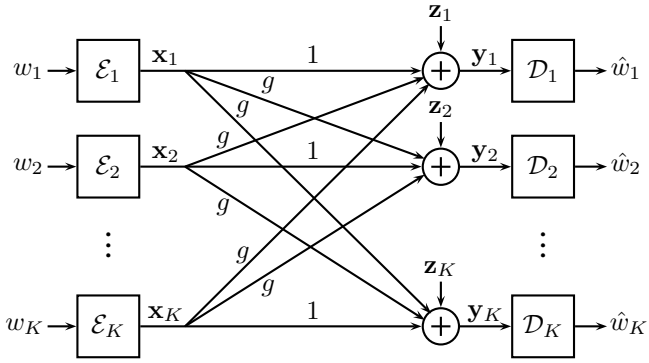
\begin{figure}[h]
\psset{unit=.85mm}
\begin{center}
\begin{pspicture}(0,10)(95,65)

\rput(0.5,55){$w_1$} \psline{->}(3.5,55)(8,55) \psframe(8,50)(18,60)
\rput(13,55){$\mathcal{E}_1$} \rput(22,57){$\bx_1$}
\psline{-}(18,55)(25,55)
\psline{->}(25,55)(62,55)
\rput(45,57.5){$1$}
\psline{->}(25,55)(62.5,41.25)
\rput(38,52.75){$g$}
\psline{->}(25,55)(63.5,17)
\rput(34,49){$g$}

\rput(0.5,40){$w_2$} \psline{->}(3.5,40)(8,40) \psframe(8,35)(18,45)
\rput(13,40){$\mathcal{E}_2$} \rput(22,42){$\bx_2$}
\psline{-}(18,40)(25,40)
\psline{->}(25,40)(62,40)
\rput(45,42.5){$1$}
\psline{->}(25,40)(62.5,54)
\rput(31,44.25){$g$}
\psline{->}(25,40)(62.5,16)
\rput(31,33.75){$g$}

\rput(13,29){\Large{$\vdots$}}
\rput(81,29){\Large{$\vdots$}}

\rput(0,15){$w_K$} \psline{->}(3.5,15)(8,15) \psframe(8,10)(18,20)
\rput(13,15){$\mathcal{E}_K$} \rput(22,17){$\bx_K$}
\psline{-}(18,15)(25,15)
\psline{->}(25,15)(62,15)
\rput(45,17.5){$1$}
\psline{->}(25,15)(62.5,38.75)
\rput(34,26.25){$g$}
\psline{->}(25,15)(63.25,53)
\rput(38,21){$g$}

\rput(10,0){
\pscircle(55,55){3} \psline{-}(55,53.5)(55,56.5) \psline{-}(53.5,55)(56.5,55)
\psline{->}(55,62)(55,58) \rput(55,64){$\bz_1$}
}
\psline{->}(68,55)(76,55) \rput(72,57){$\by_1$}

\rput(0,5){
\rput(10,0){
\pscircle(55,35){3} \psline{-}(55,33.5)(55,36.5) \psline{-}(53.5,35)(56.5,35)
\psline{->}(55,42)(55,38) \rput(55,44){$\bz_2$}
}
\psline{->}(68,35)(76,35) \rput(72,37){$\by_2$}
}
\rput(10,0){
\pscircle(55,15){3} \psline{-}(55,13.5)(55,16.5) \psline{-}(53.5,15)(56.5,15)
\psline{->}(55,22)(55,18) \rput(55,24){$\bz_K$}
}
\psline{->}(68,15)(76,15) \rput(72,17){$\by_K$}

\psframe(76,50)(86,60) \rput(81,55){$\mathcal{D}_1$}
\psline{->}(86,55)(91,55) \rput(94,55.5){${\hat{w}}_1$}

\psframe(76,35)(86,45) \rput(81,40){$\mathcal{D}_2$}
\psline{->}(86,40)(91,40) \rput(94,40.5){${\hat{w}}_2$}

\psframe(76,10)(86,20) \rput(81,15){$\mathcal{D}_K$}
\psline{->}(86,15)(91,15) \rput(94.5,15.5){${\hat{w}}_K$}

\end{pspicture}
\end{center}
\caption{Block diagram of a symmetric Gaussian $K$-user interference channel.}\label{f:symIC}
\end{figure}
 

There are $K$ transmitter-receiver pairs that wish to simultaneously communicate across a shared channel over $n$ time slots, where the channel gains are constant over all $n$ channel uses. We assume a real-valued channel model throughout.

\begin{definition}[Messages] Each transmitter has a {\em message} $w_k$ drawn independently and uniformly over $\{1,2,\ldots,2^{n\rsym}\}$.
\end{definition}

\begin{definition}[Encoders] Each transmitter is equipped with an {\em encoder}, $\mathcal{E}_k: \{1,2,\ldots, 2^{n\rsym}\} \rightarrow \RR^n$, that maps its message into a length-$n$ channel input $\bx_k = \mathcal{E}_k(w_k)$ that satisfies the power constraint,
\begin{align*}
\| \bx_k \|^2 \leq n \Tsnr
\end{align*} where $\Tsnr > 0$ is the signal-to-noise ratio.
\end{definition}

\begin{definition}[Channel Model] The channel output at each receiver is a noisy linear combination of its desired signal and the sum of the interfering terms, of the form
\begin{align}
\by_k = \bx_k + g \sum_{\ell \neq k} \bx_\ell + \bz_k \ ,\label{symICeq}
\end{align} where $g > 0$ parametrizes the interference strength and $\bz_k$ is an i.i.d. Gaussian vector with mean $0$ and variance $1$. We define the {\em interference-to-noise ratio} to be
\begin{align*}
\Tinr \triangleq g^2 \Tsnr
\end{align*} and the {\em interference level} to be
\begin{align*}
\alpha \triangleq \frac{\log (\Tinr)}{\log (\Tsnr)} \ .
\end{align*}
\end{definition}

\begin{remark}
Note that our definition of $\Tinr$ ignores the fact that there are $K-1$ interferers observed at each receiver. This is for two reasons. First, this definition parallels that of the two-user case \cite{etw08}, which will make it easier to compare the two rate regions. Second, the receivers will often be able to treat the interference as stemming from a single effective transmitter, via interference alignment. Of course, this is not the case when the receiver treats the interference as noise, as discussed in Section~\ref{s:noisyinterferenceregime}.
\end{remark}

\begin{definition}[Decoders] Each receiver is equipped with a {\em decoder}, $\mathcal{D}_k:\RR^n \rightarrow \{1,2,\ldots,2^{n\rsym}\}$, that produces an estimate $\hat{w}_k = \mathcal{D}_k(\by_k)$ of its desired message $w_k$.
\end{definition}

\begin{definition}[Symmetric Capacity] A symmetric rate $\rsym$ is {\em achievable} if, for any $\epsilon > 0$ and $n$ large enough, there exist encoders and decoders that can attain probability of error at most $\epsilon$,
\begin{align*}
\Pr\Big(\{ \hat{w}_1 \neq w_1 \} \cup \cdots \cup \{ \hat{w}_K \neq w_K \} \Big) < \epsilon \ .
\end{align*} The {\em symmetric capacity} $\csym$ is the supremum of all achievable symmetric rates.
\end{definition}

\vspace{2mm}

\begin{remark}
Due to the symmetry of the channel, the symmetric capacity is equal to the sum capacity, normalized by the number of users. To see this, assume that the users employ different rates and that a rate tuple $(R_1,R_2,\ldots, R_K)$ is achievable. Since each transmitter-receiver pair sees the same effective channel, we can simply exchange the encoders and decoders to achieve the rate tuple $(R_{\pi(1)}, R_{\pi(2)}, \ldots, R_{\pi(K)})$ for any permutation $\pi$. By time-sharing across all permutations, we find that each user can achieve $\frac{1}{K}\sum_{k=1}^K R_k$, corresponding to a symmetric rate. Thus, the sum of any achievable rate tuple is upper bounded by $K \csym$.
\end{remark}

\begin{definition}[Generalized Degrees-of-Freedom] The {\em generalized degrees-of-freedom} (GDoF) specifies the fraction of the point-to-point Gaussian capacity that can be attained per user for a given interference level $\alpha \geq 0$ as $\Tsnr$ tends to infinity,
\begin{align*}
d(\alpha) = \lim_{\Tsnr \rightarrow \infty} \frac{\csym}{\frac{1}{2}\log(1 + \Tsnr)} \ .
\end{align*}
\end{definition}

\subsection{Approximate Sum Capacity} \label{s:sumcapacity}

As shown by Jafar and Vishwanath \cite[Theorem 3.1]{jv10}, the GDoF of the symmetric $K$-user interference channel is identical to that of the two-user channel, except for a singularity at $\alpha = 1$,
\begin{align}
\label{symIC-GDoF}
d(\alpha) = \begin{cases}
1 - \alpha & 0 \leq \alpha < \frac{1}{2}  \text{~~(noisy)}\\
\alpha & \frac{1}{2} \leq \alpha < \frac{2}{3}  \text{~~(weak)}\\
1 - \frac{\alpha}{2} & \frac{2}{3} \leq \alpha < 1  \text{~~(moderately weak)}\\
\frac{1}{K} & \alpha = 1 \\
\frac{\alpha}{2} & 1 < \alpha < 2  \text{~~(strong)}\\
1 & \alpha \geq 2 \text{~~(very strong).}
\end{cases}
\end{align} See Figure \ref{f:gdof} for a plot. Notice that since $\Tsnr$ is taken to infinity, the GDoF characterization treats all channel gains $g$ that do not scale with $\Tsnr$ as a single point at $\alpha = 1$. A finer view of this regime is possible at high SNR by simply setting $g$ to be some fixed value and then taking $\Tsnr$ to infinity, corresponding to the standard notion of degrees-of-freedom. Surprisingly, this degrees-of-freedom characterization is discontinuous at rational values of $g$ \cite{eo09}. This presents an obstacle towards a clean capacity approximation at finite SNR.

\begin{figure}[!htb]
\psset{unit=.7mm}
\begin{center}
\begin{pspicture}(-5,-5)(120,80)

\psline{->}(0,0)(120,0)
\psline{->}(0,0)(0,80)

\rput(118,-5){$\alpha$}
\psline{-}(100,-2)(100,2)
\rput(100,-5){$2$}
\psline{-}(50,-2)(50,2)
\rput(50,-5){$1$}
\psline{-}(33.3,-2)(33.3,2)
\rput(33.3,-6){$\frac{2}{3}$}
\psline{-}(25,-2)(25,2)
\rput(25,-6){$\frac{1}{2}$}

\rput(7,80){$d(\alpha)$}
\psline{-}(-2,75)(2,75)
\rput(-4,75){$1$}
\psline{-}(-2,50)(2,50)
\rput(-4,50){$\frac{2}{3}$}
\psline{-}(-2,37.5)(2,37.5)
\rput(-4,37.5){$\frac{1}{2}$}
\psline{-}(-2,15)(2,15)
\rput(-4.5,15){$\frac{1}{K}$}

\psline[linewidth=2pt]{-}(0,75)(25,37.5)(33.3,50)(50,37.5)(100,75)(120,75)
\pscircle[linewidth=1pt,fillcolor=white,fillstyle=solid](50,37.5){1.5}
\pscircle[linewidth=1pt,fillcolor=black,fillstyle=solid](50,15){1.5}

\end{pspicture}
\end{center}
\caption{Generalized degrees-of-freedom for the symmetric Gaussian $K$-user interference channel.}\label{f:gdof}
\end{figure}
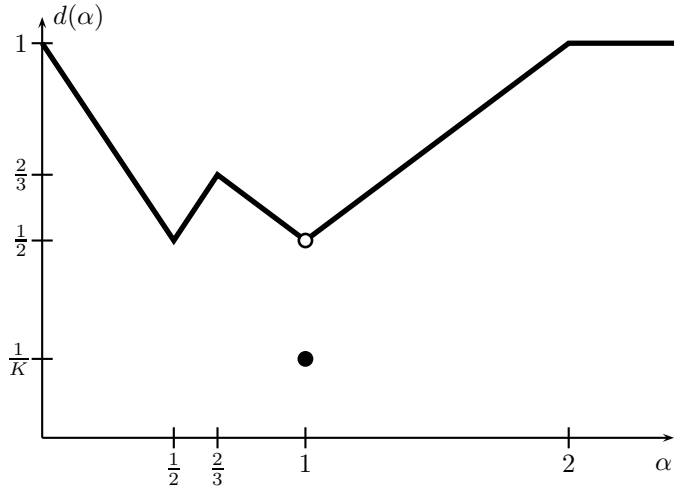
 

\begin{figure*}[!htb]
\begin{center}
\subfloat[$\Tsnr = 20$dB]{
\includegraphics[width=0.45\textwidth]{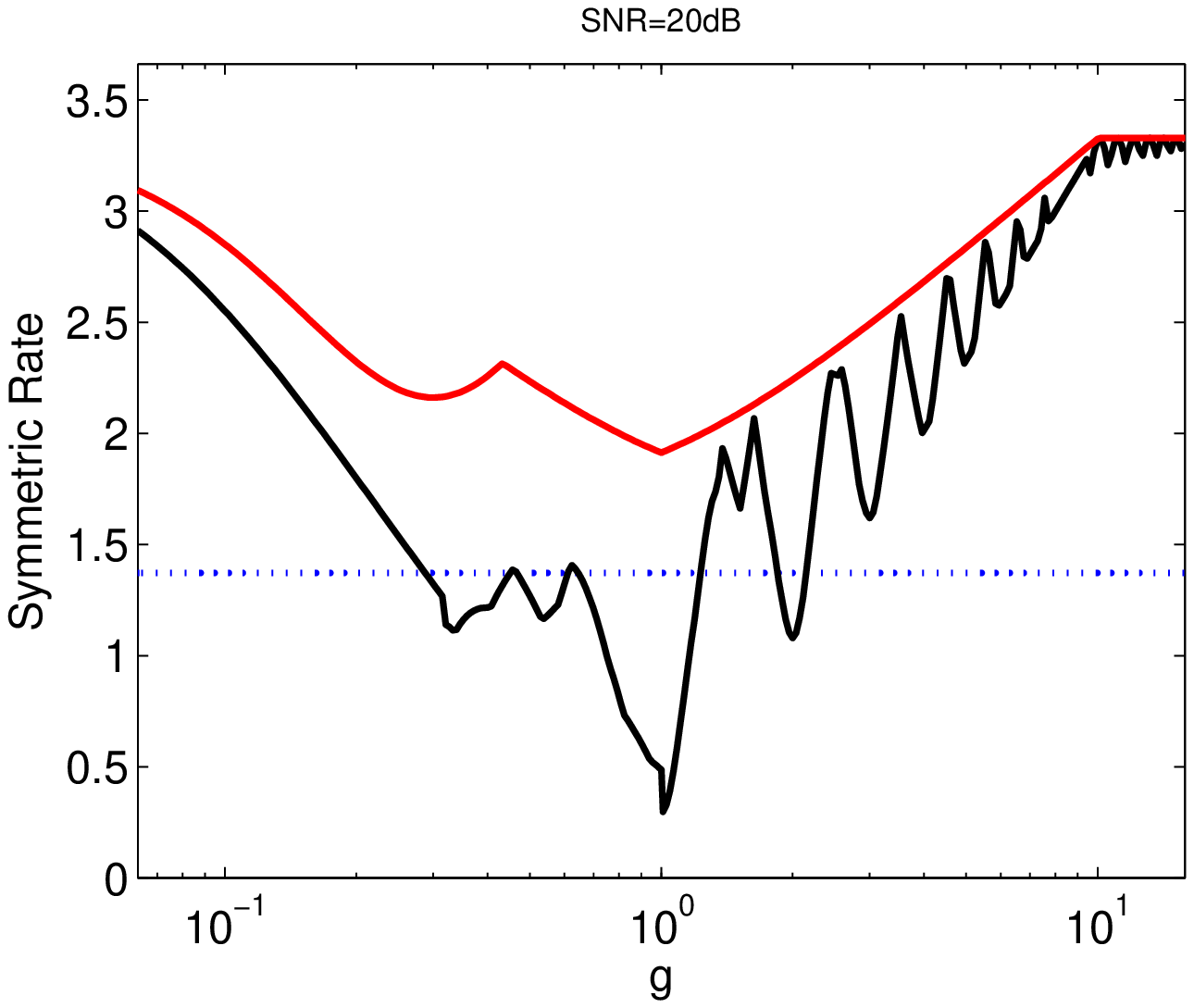}
\label{f:rate20dB}}
\qquad
\subfloat[$\Tsnr = 35$dB]{
\includegraphics[width=0.45\textwidth]{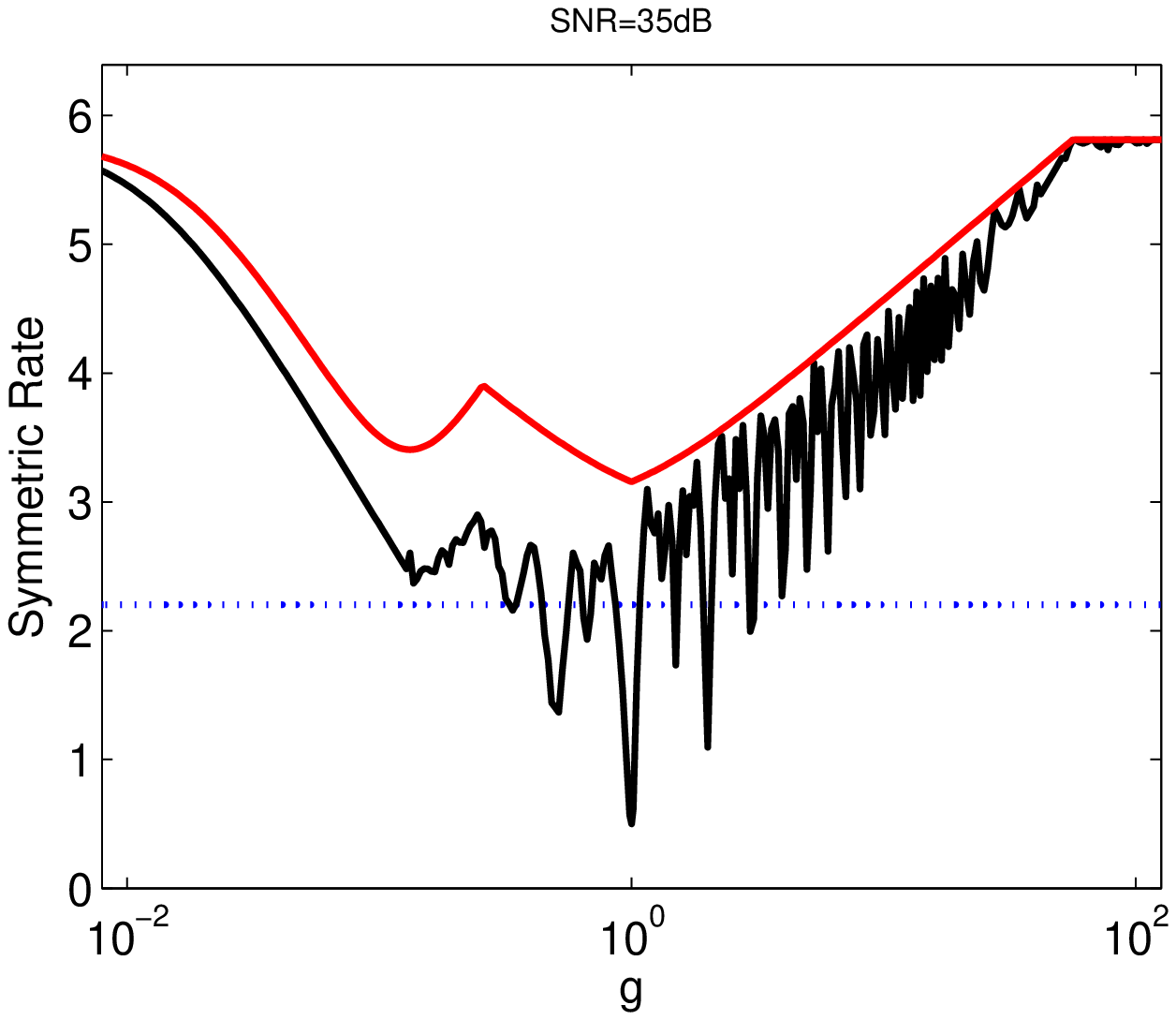}
\label{f:rate35dB}}

\subfloat[$\Tsnr = 50$dB]{
\includegraphics[width=0.45\textwidth]{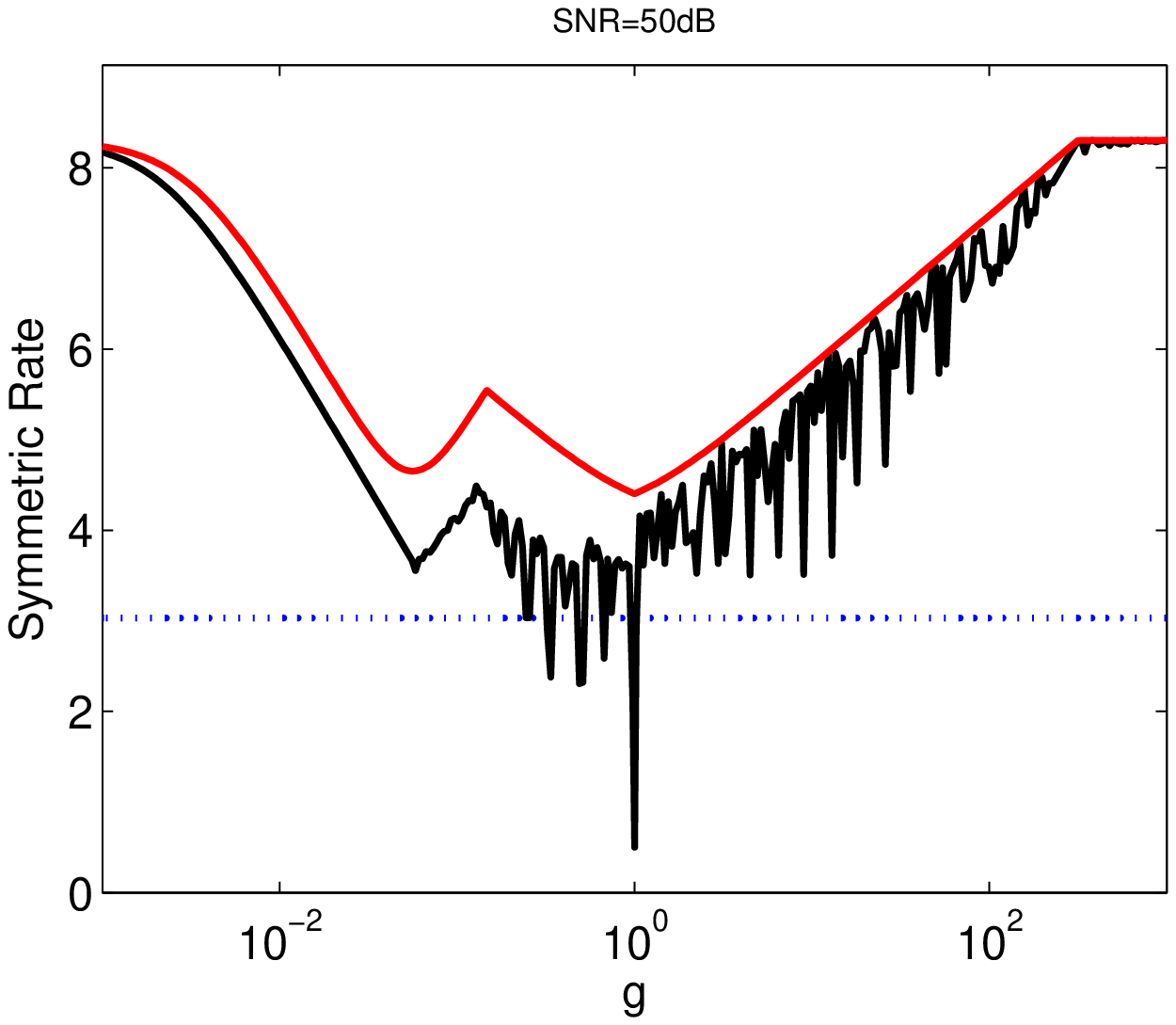}
\label{f:rate50dB}}
\qquad
\subfloat[$\Tsnr = 65$dB]{
\includegraphics[width=0.45\textwidth]{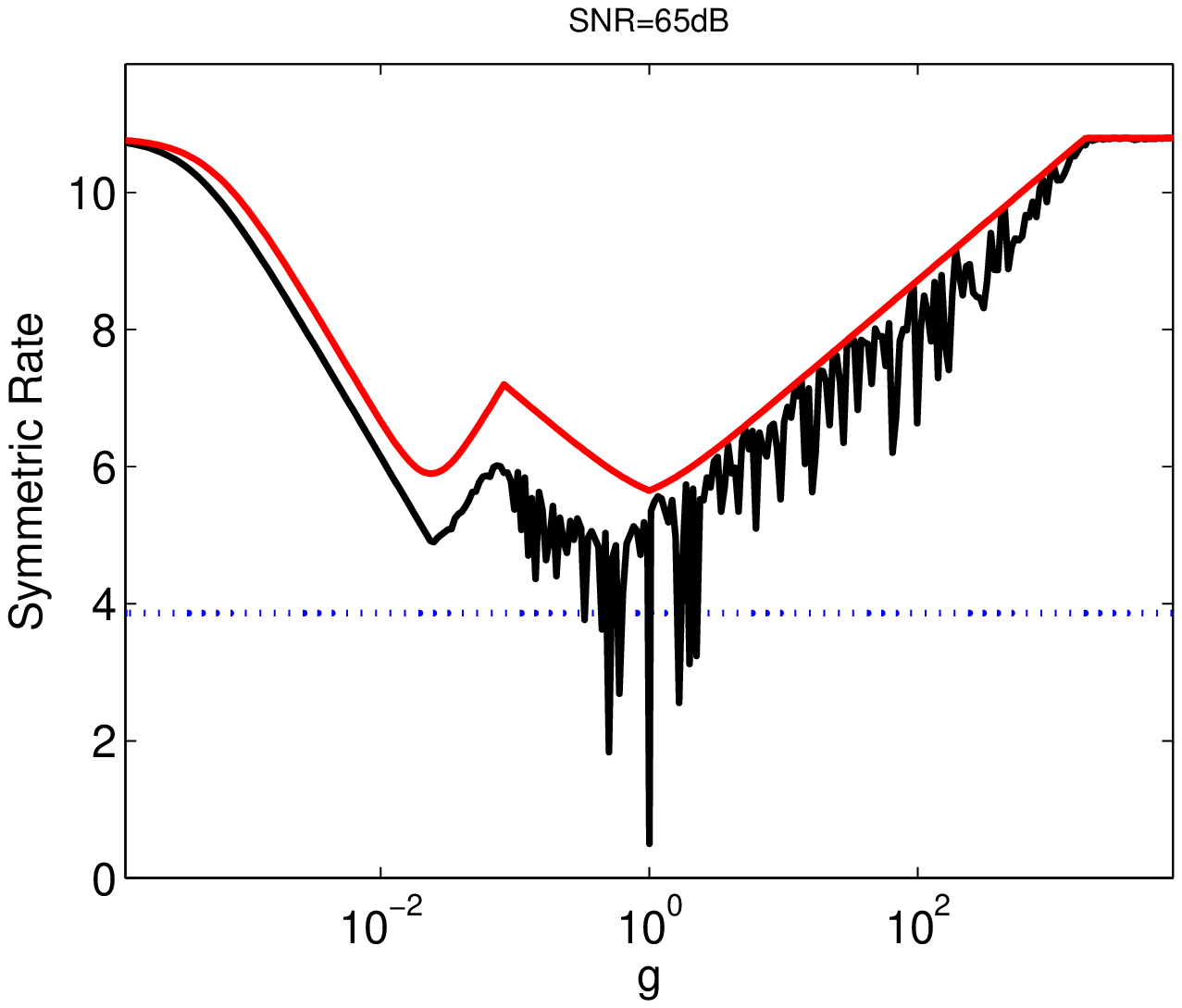}
\label{f:rate65dB}}
\end{center}
\caption{Upper and lower bounds on the symmetric capacity of a $3$-user symmetric Gaussian interference channel with respect to the cross-gain $g$. The upper bound (red line) is given by \eqref{upperBoundsEtkin} and the lower bound (black line) is the maximum of the achievable rates from Theorem \ref{thm:SymICnoLayeres} and Corollary \ref{cor:SymICHK}, which were computed numerically, and Theorem~\ref{thm:noisyInterference}. The lower bounds from Theorem~\ref{thm:approxCapacity} are not plotted in this figure.
For reference, we have also plotted the symmetric rate achievable via time-division (dotted blue line).}
\label{f:symICrates}
\end{figure*}

To overcome this difficulty, our approximations allow for the possibility of an {\em outage set}, which is explicitly characterized. Specifically, in the regime around $\alpha = 1$, our capacity results take the following shape: for any constant $c > 0$, the capacity is approximated within at most $c + 9 + \log{K}$ bits over the entire range of $\Tsnr$, and all channel gains $g$, except for a set of measure $\mu(c)$ which vanishes rapidly with $c$. This type of capacity approximation has also been used by Niesen and Maddah-Ali for the two-user Gaussian X channel \cite{nm13} and seems to arise from the capacity region itself, not just the lower bound. That is, it appears that the capacity may in fact simultaneously vary rapidly with the fine scale of the channel gains (e.g., the distance to an appropriately scaled integer) and slowly on the coarse scale (e.g., relative interference strength). In the high SNR limit, this behavior shows up as a discontinuity on the rationals but, at reasonable SNRs, our achievable scheme shows that this variation is in fact fairly smooth. The theorem below captures our capacity approximations in a simple form. All upper bounds in the theorem are based on~\cite{etw08} and~\cite{jv10}. The lower bound for the noisy interference regime is straightforward and the lower bound for the very strong interference regime is (a slight variation of) that of~\cite{sjvj08}. Our contibution is in the lower bounds for the weak and strong interference regimes.

\begin{theorem}
\label{thm:approxCapacity}
The symmetric capacity of the symmetric Gaussian $K$-user interference channel can be lower and upper bounded as follows:
\begin{itemize}
\item {\em Noisy Interference Regime,} $0 \leq \alpha < \frac{1}{2}$,
\begin{align*}
&\frac{1}{2}\log\left(1+\frac{\Tsnr}{1+\Tinr}\right)-\frac{1}{2}\log(K-1)\\ &~~~ \qquad \qquad \qquad \leq\csym<\frac{1}{2}\log\left(1+\frac{\Tsnr}{1+\Tinr}\right)+1
\end{align*}
\item {\em Weak Interference Regime,} $\frac{1}{2} \leq \alpha < \frac{2}{3}$,
\begin{align*}
\frac{1}{2} \log^+(\Tinr) - \frac{7}{2}- \log(K) \leq \csym \leq \frac{1}{2} \log^+(\Tinr) + 1
\end{align*} for all channel gains.
\item {\em Moderately Weak Interference Regime,} $\frac{2}{3} \leq \alpha < 1$,
\begin{align*}
&\frac{1}{2} \log^+\left(\frac{\Tsnr}{\sqrt{\Tinr}}\right) - c-8-\log(K) \\ & ~\qquad\qquad \qquad \qquad \leq \csym  \leq \frac{1}{2} \log^+\left(\frac{\Tsnr}{\sqrt{\Tinr}}\right)+1
\end{align*}
for all channel gains except for an outage set of measure $\mu < 2^{-c}$ for any $c > 0$.
\item {\em Strong Interference Regime,} $1 \leq \alpha < 2$,
\begin{align*}
\frac{1}{4} \log^+(\Tinr) - \frac{c}{2}-3 \leq \csym \leq \frac{1}{4} \log^+(\Tinr) + 1
\end{align*} for all channel gains except for an outage set whose measure is a fraction of $2^{-c}$ of the interval $1<|g|<\sqrt{\Tsnr}$, for any $c > 0$.

\item {\em Very Strong Interference Regime,} $\alpha \geq 2$,
\begin{align*}
\frac{1}{2} \log(1+\Tsnr)-1\leq\csym \leq \frac{1}{2} \log(1+\Tsnr)
\end{align*}
\end{itemize}
\end{theorem}

\begin{remark}
Our characterization of the outage set in the strong and moderately weak interference regimes is in fact somewhat stronger than the characterization given in Theorem~\ref{thm:approxCapacity}. Specifically, for the strong interference regime we show that, for any integer $b$ in the range $[1,\sqrt{\Tsnr})$ and constant gap $c > 0$, the measure of the set of channel coefficients in the interval $g\in[b,b+1)$ for which our inner bound does not hold is smaller than $2^{-c}$. Similarly, for the moderately weak interference regime we show that, for any integer $b$ in the range $[1,1/6\log(\Tsnr))$ and constant gap $c >0$, the measure of the set of channel coefficients the interval $g\in[2^{-b},2^{-b+1})$ for which our inner bound does not hold is smaller than $2^{-(c+b)}$. Using this refined characterization, our results can be interpreted in the following way: For all values of $\alpha$ except for an outage set with Lebesgue measure smaller than $2^{-c}$, the symmetric capacity of the symmetric Gaussian $K$-user interference channel is
\begin{align}
\csym=\frac{d(\alpha)}{2}\log(\Tsnr)\pm\delta(K,c),\nonumber
\end{align}
where $0\leq \delta(K,c)<c+\log(K)+10$ and $d(\alpha)$ is given in~\eqref{symIC-GDoF}.
\end{remark}

\section{Preliminaries} \label{s:cf}

In this section, we give some basic definitions and results that will be extensively used in the sequel.

\subsection{$K$-user Gaussian MAC}
Consider the $K$-user Gaussian MAC
\begin{align}
\by=\sum_{k=1}^K h_k \bx_k +\bz,\label{MACchannel}
\end{align}
where the vector $\bh=[h_1 \ \cdots \ h_K]^T\in\RR^K$ represents the channel gains, $\bx_k\in\RR^n$, $k=1,\ldots,K$, are the channel inputs, $\bz\in\RR^n$ is additive white Gaussian noise (AWGN) with zero mean and unit variance and $\by\in\RR^n$ is the channel output. Without loss of generality, we assume all $K$ users are subject to the same power constraint\footnote{As otherwise the different powers can be absorbed into the channel gains.}
\begin{align}
\|\bx_k\|^2\leq n\Tsnr, \ k=1,\ldots,K.\label{powerConstraint}
\end{align}
The capacity region of the channel~\eqref{MACchannel} is known (see e.g.,~\cite[Theorem 15.3.6]{coverthomas}) to be the set of all rate tuples $(R_1,\ldots, R_K)$ satisfying
\begin{align}
\sum_{k\in\mathcal{S}} R_k<\frac{1}{2}\log\left(1+\Tsnr\sum_{k\in\mathcal{S}}|h_k|^2\right)
\end{align}
for all subsets $\mathcal{S}\subseteq\{1,\ldots,K\}$. The achievability part of the capacity theorem is established using i.i.d. Gaussian codebooks for all users. Motivated by lattice interference alignment, we are interested in establishing the achievability of certain rate tuples under the constraint that the codebooks employed by the $K$ users form a chain of nested lattice codes.

\begin{remark} Recall that the corner points of the capacity region are achievable via successive interference cancellation, either using i.i.d. Gaussian codebooks \cite[Section 15.3.6]{coverthomas} or nested lattice codebooks \cite[Section VII.A]{ng11IT}. Time-sharing between these corner points suffices to reach any point in the capacity region. However, this time-sharing approach does not suffice for an interference channel, as each receiver will require a different time allocation between users.
\end{remark}

\subsection{Nested Lattice Codes}

We employ the nested lattice framework originally proposed in \cite{ez04}. A lattice $\Lambda$ is a discrete subgroup of $\RR^n$ which is closed under reflection and real addition. Formally, for any $\bt_1,\bt_2\in\Lambda$, we have that $-\bt_1,-\bt_2\in\Lambda$ and $\bt_1+\bt_2\in\Lambda$. Note that by definition the zero vector $\mathbf{0}$ is always a member of the lattice. Any lattice $\Lambda$ in $\RR^n$ is spanned by some $n\times n$ matrix $\bG$ such that \begin{align}
\Lambda=\{\mathbf{t}=\bG\mathbf{q}:\mathbf{q}\in\ZZ^n\}.\nonumber
\end{align}
We say that a lattice is full-rank if its spanning matrix $\bG$ is full-rank.

We denote the nearest neighbor quantizer associated with the lattice $\Lambda$ by
\begin{align}
\Ql(\bx)=\arg\min_{\mathbf{t}\in\Lambda}\|\bx-\mathbf{t}\|.
\label{NNquantizer}
\end{align}
The Voronoi region of $\Lambda$, denoted by $\CV$, is the set of all points in $\RR^n$ which are quantized to the zero vector, where ties in~\eqref{NNquantizer} are broken in a systematic manner.
The modulo operation returns the quantization error w.r.t. the lattice,
\begin{align}
\left[\bx\right]\Mod=\bx-\Ql(\bx),\nonumber
\end{align}
and satisfies the distributive law,
\begin{align}
\big[a[\bx]\Mod+b[\by]\Mod\big]\Mod=\left[a\bx+b\by\right]\Mod,\nonumber
\end{align}
for all $a,b\in\ZZ$ and $\bx,\by\in\RR^n$.

A lattice $\Lambda$ is said to be nested in $\Lambda_1$ if $\Lambda\subseteq\Lambda_1$. The coding schemes presented in this paper utilize a chain of $K+1$ nested lattices satisfying
\begin{align}
\Lambda\subseteq\Lambda_K\subseteq\cdots\subseteq\Lambda_1.\label{nestedChain}
\end{align}
From these lattices, we construct $K$ codebooks, one for each user. Specifically, user $k$ is allocated the codebook $\mathcal{L}_k=\Lambda_{\map(k)}\cap\CV$, where $\CV$ is the Voronoi region of $\Lambda$ and the function $\map(k):\{1,\ldots,K\}\rightarrow\{1,\ldots,K\}$ maps between users and lattices. The rate of each codebook $\mathcal{L}_k$ is
\begin{align}
R_k=\frac{1}{n}\log\big|\Lambda_{\map(k)}\cap\CV\big|.\nonumber
\end{align}
User $k$ encodes its message into a lattice point from its codebook, $\bt_k\in\mathcal{L}_k$. Each user also has a random\footnote{It can be shown that these random dithers can be replaced with deterministic ones, meaning that no common randomness is required.} dither vector $\bd_k$ which is generated independently and uniformly over $\CV$. These dithers are made available to the decoder.
The signal transmitted by user $k$ is
\begin{align}
\bx_k=\left[\bt_k-\bd_k\right]\Mod.\nonumber
\end{align}

\begin{remark} \label{r:intshaping}
The nested lattice construction from \cite{ez04} employs Construction A. To create each fine lattice, this procedure first embeds codewords drawn from a linear code into the unit cube, and then applies the generator matrix for the coarse lattice $\Lambda$. As shown in \cite{ez04}, this ensemble of nested lattice codes can approach the capacity of a point-to-point Gaussian channel. If the integers $\mathbb{Z}^n$ are selected as the coarse lattice, the resulting nested lattice code is equivalent to a linear code coupled with a pulse amplitude modulation (PAM) constellation. Furthermore, the $\hspace{-0.07in}\mod\Lambda$ operation simplifies to the quantization error from rounding to the integers. It can be shown that the cost of this simplification is only the shaping gain, which corresponds to at most $1/2\log(2\pi e/12)\simeq 0.255$ bits per channel use~\cite{oe12}.
\end{remark}

\subsection{Compute-and-Forward}

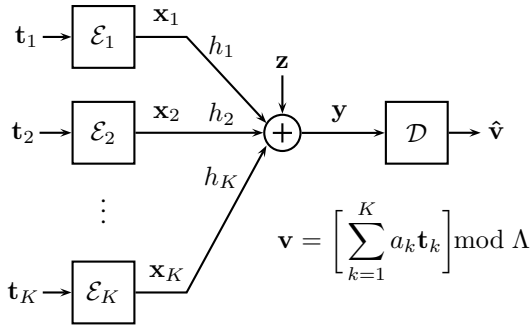
\begin{figure}[h]
\psset{unit=.85mm}
\begin{center}
\begin{pspicture}(19,10)(94,60)

\rput(20,55){$\mathbf{t}_1$} \psline{->}(22.5,55)(27,55) \psframe(27,50)(37,60)
\rput(32,55){$\mathcal{E}_1$} \rput(42,58){$\mathbf{x}_1$}
\psline[linecolor=black]{->}(37,55)(45,55)(57.5,41.5)
\rput(50.5,53.5){\textcolor{black}{$h_1$}}

\rput(19.5,40){$\mathbf{t}_2$} \psline{->}(22,40)(27,40) \psframe(27,35)(37,45)
\rput(32,40){$\mathcal{E}_2$} \rput(42,43){$\mathbf{x}_2$}
\psline{->}(37,40)(57,40)
\rput(50.5,43){\textcolor{black}{$h_2$}}

\rput(19.5,15){$\mathbf{t}_K$} \psline{->}(22.5,15)(27,15) \psframe(27,10)(37,20)
\rput(32,15){$\mathcal{E}_K$} \rput(42,18){$\mathbf{x}_K$}
\psline[linecolor=black]{->}(37,15)(45,15)(57.5,38)
\rput(50.5,33){\textcolor{black}{$h_K$}}

\rput(32,29){{$\vdots$}}


\rput(0,5){
\rput(5,0){
\pscircle(55,35){3} \psline{-}(55,33.5)(55,36.5) \psline{-}(53.5,35)(56.5,35)
\psline{->}(55,44)(55,38) \rput(55,46){$\mathbf{z}$}
}
\psline{->}(63,35)(76,35) \rput(69,38){$\mathbf{y}$}
}

\psframe(76,35)(86,45) \rput(81,40){$\mathcal{D}$}
\psline{->}(86,40)(91,40) \rput(93.5,40.5){$\mathbf{\hat{v}}$}
\rput(79,23){{${\displaystyle \mathbf{v} = \bigg[ \sum_{k=1}^K a_k \mathbf{t}_k} \bigg]\hspace{-0.05in}\Mod$}}

\end{pspicture}
\end{center}
\caption{Compute-and-forward on a Gaussian multiple-access channel. The transmitters send lattice points $\mathbf{t}_k$ and the receiver decodes an integer combination of them, modulo the coarse lattice $\Lambda$. The rate is  determined by how closely the equation coefficients $a_k$ match the channel coefficients $h_k$.} \label{f:computeforward}
\end{figure}
 

Our objective is to communicate over the MAC using the compute-and-forward scheme from~\cite{ng11IT}. To this end, the receiver first decodes a linearly independent set of $K$ integer linear combinations of the lattice codewords. Afterwards, it solves this set of linear combinations for the lattice codewords. Assume that the receiver is interested in decoding the integer linear combination
\begin{align}
\bv=\left[\sum_{k=1}^K a_k\bt_k\right]\Mod\nonumber
\end{align}
with coefficient vector $\ba=[a_1 \  \cdots \ a_K]^T\in\ZZ^K$.
Following the scheme of~\cite{ng11IT},
the receiver scales the observation $\by$ by a factor $\beta$, removes the dithers, and reduces modulo $\Lambda$ to get
\begin{align}
\bs&=\left[\beta\by+\sum_{k=1}^K a_k\bd_k\right]\bmod\Lambda\nonumber\\
&=\left[\sum_{k=1}^K a_k \bx_k+\sum_{k=1}^K a_k\bd_k+\sum_{k=1}^K (\beta h_k-a_k )\bx_k+\beta\bz\right]\bmod\Lambda\nonumber\\
&=\left[\bv+\bz_{\text{eff}}(\bh,\ba,\beta)\right]\Mod,\label{CoFestimate}
\end{align}
where
\begin{align}
\bz_{\text{eff}}(\bh,\ba,\beta) = \sum_{k=1}^K (\beta h_k-a_k)\bx_k+\beta\bz\label{Zeff}
\end{align}
is effective noise. From~\cite{ng11IT}, we have that $\bz_{\text{eff}}(\bh,\ba,\beta)$ is statistically independent of $\bv$ and its effective variance, defined as
\begin{align}
\sigma^2_{\text{eff}}(\bh,\ba,\beta) \triangleq \frac{1}{n}\mathbb{E}\|\bz_{\text{eff}}(\bh,\ba,\beta)\|^2
\end{align}
is
\begin{align}
\sigma^2_{\text{eff}}(\bh,\ba,\beta) = \|\beta \bh -\ba\|^2\cdot\Tsnr+\beta^2.\label{ZeffVar}
\end{align} 
Let $k^*=\min_{k:a_k\neq 0}\map(k)$ be the index of the densest lattice participating in the integer linear combination $\bv$. The receiver produces an estimate for $\bv$ by applying to $\bs$ the lattice quantizer associated with $\Lambda_{k^*}$,
\begin{align}
\hat{\bv}=\left[Q_{\Lambda_{k^*}}(\bs)\right]\Mod.
\end{align}
Let $\CV_{k^*}$ be the Voronoi region of $\Lambda_{k^*}$, and note that the probability of decoding error is upper bounded by the probability that the effective noise lies outside the Voronoi region of $\Lambda_{k^*}$,
\begin{align}
\Pr\left(\hat{\bv}\neq\bv\right)\leq\Pr\left(\bz_{\text{eff}}(\bh,\ba,\beta)\notin\CV_{k^*}\right).
\end{align}
The next theorem summarizes and reformulates relevant results from Sections IV.C, IV.D, and V.A of~\cite{ng11IT}.
\begin{theorem}
\label{thm:CoF}
For any $\epsilon>0$ and $n$ large enough there exists a chain of $n$-dimensional nested lattices $\Lambda\subseteq\Lambda_K\subseteq\cdots\subseteq\Lambda_1$ forming the set of codebooks $\mathcal{L}_1,\ldots,\mathcal{L}_K$ having rates $R_1,\ldots,R_K$ and satisfying the power constraint~\eqref{powerConstraint} such that:
\begin{enumerate}[(a)]
\item For all channel vectors $\bh\in\RR^K$ and coefficient vectors $\ba\in\ZZ^K$, the average error probability in decoding the integer linear combination $\bv=\left[\sum_{k=1}^K a_k\bt_k\right]\Mod$
of transmitted lattice points $\bt_k\in\mathcal{L}_k$ can be made smaller than $\epsilon$ so long as the message rates do not exceed the computation rate,
\begin{align}
R_k&<R_{\text{comp}}(\bh,\ba,\beta) \triangleq \frac{1}{2}\log\left(\frac{\Tsnr}{\sigma^2_{\text{eff}}(\bh,\ba,\beta)}\right) \ ,\label{compRate}
\end{align}
for all $k$ such that $a_k \neq 0$ and some $\beta \in \RR$. \label{thm:CoF1}
\item The codebooks $\mathcal{L}_1,\ldots,\mathcal{L}_K$ are isomorphic to some set of linear codebooks $\mathcal{C}_1,\ldots,\mathcal{C}_K$ over the finite field $\Zp$, where $p$ is a sufficiently large prime number.\label{thm:CoF2}
\item For the same $p$, the equation $[p\cdot\bt]\Mod=\mathbf{0}$ holds $\forall\bt\in\Lambda_k$, $k=1,\ldots,K$.\label{thm:CoF3}
\end{enumerate}
\end{theorem}

\begin{corollary}
\label{cor:invertibility}
Given $K$ integer linear combinations $\bV=[\bv_1 \ \cdots \ \bv_K]$ with coefficient vectors $\bA=[\ba_1 \ \cdots \ \ba_K]^T$, the lattice points $\bt_1, \ldots, \bt_K$ can be recovered if $[\bA]\bmod p$ is full rank over $\Zp$.
\end{corollary}

\vspace{1mm}

\begin{remark} By taking the blocklength $n$ and field size $p$ to be large enough, it can be shown that, for a fixed channel vector $\bh$ and finite $\Tsnr$, it suffices to check whether $\bA$ is full rank over the reals. See~\cite[Section VI]{ng11IT} for an in-depth discussion.
\end{remark}

\vspace{1mm}

\begin{remark}
Note that it is also possible to map both the messages and the integer linear combinations into an appropriately chosen finite field. That is, the messages can be written as vectors with elements that take values in a prime-sized finite field, and the receiver ultimately recovers linear combinations of the messages over the same finite field. See \cite{ng11IT} for more details.
\end{remark}

\vspace{1mm}

It follows from Theorem~\ref{thm:CoF}\eqref{thm:CoF1} that in order to maximize the computation rate $R_{\text{comp}}(\bh,\ba,\beta)$ for a given coefficient vector, one has to minimize $\sigma^2_{\text{eff}}(\bh,\ba,\beta)$ over $\beta$.
It is seen from~\eqref{ZeffVar} that the expression for $\sigma^2_{\text{eff}}(\bh,\ba,\beta)$ is equal to the mean squared error (MSE) for linear estimation of $\tilde{X}=\sum_{k=1}^K a_k X_k$ from $Y=\sum_{k=1}^K h_k X_k+Z$ where $\{X_k\}_{k=1}^K$ are i.i.d. random variables with zero mean and variance $\Tsnr$ and $Z$ is statistically independent of $\{X_k\}_{k=1}^K$ with zero mean and unit variance.
Hence the minimizing value of $\beta$ is the linear minimum mean squared error (MMSE) estimation coefficient of $\tilde{X}$ from $Y$.
This value of $\beta$ was found in~\cite[Theorem 2]{ng11IT} and the resulting MSE is given by
\begin{align}
\sigma^2_{\text{eff}}(\mathbf{h},\mathbf{a})&\triangleq \min_{\beta\in\RR}\sigma^2_{\text{eff}}(\mathbf{h},\mathbf{a},\beta)\nonumber\\
&=\Tsnr\left(\|\mathbf{a}\|^2-\frac{\Tsnr ( \mathbf{h}^T\mathbf{a})^2}{1+\Tsnr\|\mathbf{h}\|^2}\right)\nonumber\\
&=\Tsnr \ \mathbf{a}^T\left(\bI_{K\times K} - \frac{\Tsnr \ \mathbf{h} \mathbf{h}^T}{1+\Tsnr\|\mathbf{h}\|^2}\right) \mathbf{a} \nonumber \\
&=\mathbf{a}^T\left(\Tsnr^{-1}\bI_{K\times K}+\mathbf{h}\mathbf{h}^T\right)^{-1}\mathbf{a}\label{woodbury}\\
&=\left\|\left(\Tsnr^{-1}\bI_{K\times K}+\mathbf{h}\mathbf{h}^T\right)^{-1/2}\mathbf{a}\right\|^2,\label{optVar}
\end{align}
where~\eqref{woodbury} can be verified using Woodbury's matrix identity (i.e., the Matrix Inversion Lemma)~\cite[Thm 18.2.8]{harville}. Accordingly, we define
\begin{align}
R_{\text{comp}}(\bh,\ba)&\triangleq \max_{\beta\in\RR}R_{\text{comp}}(\bh,\ba,\beta)\nonumber\\
&=\frac{1}{2}\log\left(\frac{\Tsnr}{\sigma^2_{\text{eff}}(\mathbf{h},\mathbf{a})}\right).\label{compRate2}
\end{align}
In the sequel, we will require that the receiver decodes $K$ linearly independent integer linear combinations. However, the specific values of the coefficient vectors for these linear combinations are not important as long as they form a full-rank set. Therefore, we are free to choose these coefficients such as to maximize the corresponding computation rate.

Define the matrix \begin{align}
\bF\triangleq\left(\Tsnr^{-1}\bI_{K\times K}+\mathbf{h}\mathbf{h}^T\right)^{-1/2},\label{latticeMat}
\end{align}
and the lattice $\Lambda(\bF)=\{\bm{\nu}=\bF\ba \ : \ \ba\in\ZZ^K\}$. Notice that this $K$-dimensional lattice is induced by the channel matrix, not the $n$-dimensional coding scheme.
The effective variance for the coefficient vector $\ba$ is
\begin{align}
\sigma^2_{\text{eff}}(\mathbf{h},\mathbf{a})=\|\bF\ba\|^2,\label{effvarlattice}
\end{align}
and hence $\sigma^2_{\text{eff}}(\mathbf{h},\mathbf{a})$ is the length of the lattice vector corresponding to the integer-valued vector $\ba$. It follows that the problem of finding the $K$ linearly independent integer-valued vectors that result in the highest computation rates is equivalent to finding a set of shortest independent vectors $\{\bm{\nu}_1,\ldots,\bm{\nu}_K\}$ in the lattice $\Lambda(\bF)$, and then taking the integer coefficient vectors as $\ba_m=\bF^{-1}\cdot\bm{\nu}_m$. The lengths of the shortest linearly independent vectors in a lattice are called \emph{successive minima}, as defined next.

\begin{definition}[Successive minima]
\label{def:sucmin}
Let $\Lambda(\mathbf{F})$ be a full-rank lattice in $\RR^K$ spanned by the matrix $\mathbf{F} \in \mathbb{R}^{K \times K}$. For $m=1,\ldots,K$, we define the $m$th successive minimum as
\begin{align}
\lambda_m(\mathbf{F}) \triangleq \inf\left\{r \ : \ \dim\left(\Span\left(\Lambda(\mathbf{F})\bigcap \mathcal{B}(\mathbf{0},r)\right)\right)\geq m\right\}\nonumber
\end{align}
where $\mathcal{B}(\mathbf{0},r)=\left\{\bx\in\RR^K \ : \ \|\bx\|\leq r\right\}$ is the closed ball of radius $r$ around $\mathbf{0}$. In words, the $m$th successive minimum of a lattice is the minimal radius of a ball centered around $\mathbf{0}$ that contains $m$ linearly independent lattice points.
\end{definition}

The following definition identifies the $K$ linearly independent coefficient vectors which yield the highest computation rates.

\vspace{2mm}

\begin{definition}
\label{def:optEqs}
Let $\bF$ be the matrix defined in~\eqref{latticeMat}. We say that an ordered set of integer coefficient vectors $\{\ba_1,\ldots,\ba_K\}$ with corresponding computation rates $R_{\text{comp},m}\triangleq R_{\text{comp}}(\bh,\ba_m)$ is \emph{optimal} if the $K$ vectors are linearly independent and $\|\bF \ \ba_m\|=\lambda_m(\bF)$ for any $m=1,\ldots,K$. Note, that such a set always exists by definition of successive minima, and that it is not unique. For example, if $\{\ba_1,\ldots,\ba_K\}$ is an optimal set of coefficient vectors, so is the set $\{-\ba_1,\ldots,-\ba_K\}$. Note also that the optimal computation rates satisfy \mbox{$R_{\text{comp},1}\geq\cdots\geq R_{\text{comp},K}$}.
\end{definition}

\vspace{2mm}

\begin{remark} Several recent papers have proposed families of constellations and codes that are well-suited for low-complexity implementations of compute-and-forward \cite{fsk11,hn11,ozeng11,tn11,hc13IT,bl12,hnt14}. These codes could serve as building blocks for a practical implementation of our alignment scheme.
\end{remark}

\subsection{Numerical Evaluations}\label{s:numerical}

The optimal coefficient vectors and computation rates from Definition~\ref{def:optEqs} play an important role in the achievable rate regions derived in this paper. The problem of determining the optimal coefficient vectors is that of finding the set of $K$ linearly independent \emph{integer-valued} vectors that minimizes the effective noise~\eqref{optVar}. As discussed above, this problem is equivalent to finding the shortest $K$ linearly independent lattice vectors in the lattice $\Lambda(\bF)$ spanned by the matrix $\bF$ defined in~\eqref{latticeMat}.

It is shown in~\cite[Lemma 1]{ng11IT} that only integer vectors \mbox{$\ba \in \mathbb{Z}^K$} that satisfy the condition
\begin{align}
\|\ba\|^2<1+\|\bh\|^2\Tsnr\label{integerCondition}
\end{align}
yield positive rates. Therefore, in our considerations it suffices to enumerate all integer vectors (other than the zero vector) that satisfy~\eqref{integerCondition}, and then exhaustively search over these vectors in order to find the optimal set. At moderate values of $\Tsnr$ this task is computationally reasonable. Nevertheless, it is sometimes simpler to find a set of short linearly independent lattice vectors in $\Lambda(\bF)$, which is not necessarily optimal, in order to obtain lower bounds on the set of optimal computation rates. A simple low-complexity algorithm for computing a short lattice basis (which forms a set of $K$ linearly independent lattice vectors) is the LLL algorithm~\cite{lll82}.\footnote{Pseudocode for the LLL algorithm can be found, e.g., in~\cite{wbkk04}.} In producing the figures for this paper we have employed the LLL algorithm, meaning that the plotted achievable rates in Figure~\ref{f:symICrates} are in fact lower bounds on the rates given by Theorems~\ref{thm:SymICnoLayeres} and~\ref{thm:SymICHK}. 

We note that a similar procedure for finding the optimal coefficient vectors was also described in~\cite{fsk11}, where the optimal coefficient vectors are termed dominated solutions.

\section{Multiple-Access via Compute-and-Forward}\label{s:mac}

This section introduces a new coding technique for reliable communication over the $K$-user Gaussian multiple-access channel. The basic idea is to first decode a linearly independent set of $K$ integer linear combinations of the transmitted codewords, and then solve these for the transmitted messages. As we will argue, under certain technical conditions, it is possible to map the users' rates to the computation rates in a one-to-one fashion.
We begin this section with a high-level overview of the scheme, which is illustrated in Figures~\ref{f:cftransform} and~\ref{f:aftertransform}.

\begin{figure*}[!t]
\begin{center}
\psset{unit=0.68mm}
\begin{pspicture}(9,-23)(260,50)

\rput(-0.5,32){
\rput(6,0){$w_1$}\psline{->}(9.5,0)(14,0)
\psframe(14,-5)(24,5) \rput(19,0){$\mathcal{L}_1$}
\psline{->}(24,0)(33,0) \rput(28,3.6){$\mathbf{t}_1$}
\pscircle(35.5,0){2.5} \psline{-}(34.25,0)(36.75,0)\psline{-}(35.5,-1.25)(35.5,1.25)
\psline{->}(35.5,8.5)(35.5,2.5) \rput(34.5,11.5){$-\mathbf{d}_1$}
\psline{->}(38,0)(43,0)
\psframe(43,-5)(58.5,5) \rput(49.5,0){\footnotesize{$\mod\Lambda$}}
}
\rput(-7,0){
\psline{->}(65,32)(75,32)(88,14) \rput(70,35){$\mathbf{x}_1$}
\rput(81,30){$h_1$}
}
\rput(-0.5,12){
\rput(6,0){$w_2$}\psline{->}(9.5,0)(14,0)
\psframe(14,-5)(24,5) \rput(19,0){$\mathcal{L}_2$}
\psline{->}(24,0)(33,0) \rput(28,3.6){$\mathbf{t}_2$}
\pscircle(35.5,0){2.5} \psline{-}(34.25,0)(36.75,0)\psline{-}(35.5,-1.25)(35.5,1.25)
\psline{->}(35.5,8.5)(35.5,2.5) \rput(34.5,11.5){$-\mathbf{d}_2$}
\psline{->}(38,0)(43,0)
\psframe(43,-5)(58.5,5) \rput(49.5,0){\footnotesize{$\mod\Lambda$}}
}
\rput(-7,0){
\psline{->}(65,12)(87,12) \rput(70,15){$\mathbf{x}_2$}
\rput(81,15.5){$h_2$}
}

\rput(19,-1){$\vdots$}
\rput(50,-1){$\vdots$}

\rput(-0.5,-18){
\rput(6,0){$w_K$}\psline{->}(10,0)(14,0)
\psframe(14,-5)(24,5) \rput(19,0){$\mathcal{L}_K$}
\psline{->}(24,0)(33,0) \rput(28,3.6){$\mathbf{t}_K$}
\pscircle(35.5,0){2.5} \psline{-}(34.25,0)(36.75,0)\psline{-}(35.5,-1.25)(35.5,1.25)
\psline{->}(35.5,8.5)(35.5,2.5) \rput(34.8,11.5){$-\mathbf{d}_K$}
\psline{->}(38,0)(43,0)
\psframe(43,-5)(58.5,5) \rput(49.5,0){\footnotesize{$\mod\Lambda$}}
}
\rput(-7,0){
\psline{->}(65,-18)(75,-18)(88,10) \rput(70,-15){$\mathbf{x}_K$}
\rput(80.5,3){$h_K$}
}
\rput(-7,0){
\pscircle(89.5,12){2.5} \psline(88.25,12)(90.75,12) \psline(89.5,10.75)(89.5,13.25)
\psline{<-}(89.5,14.5)(89.5,20.5) \rput(89.5,22.5){$\mathbf{z}$}
\psline(92,12)(102,12) \pscircle[fillstyle=solid,fillcolor=black](102,12){1}
\rput(97,15){$\mathbf{y}$}

\psline{->}(102,12)(115,32)(122,32)
\rput(110,30.5){$\beta_1$}
\psline{->}(102,12)(122,12)
\rput(110,15.5){$\beta_2$}
\psline{->}(102,12)(115,-18)(122,-18)
\rput(111,3){$\beta_K$}
}
\psframe(200,-23)(212,37) \rput(206,8.5){$\mathbf{A}_p^{-1}$}

\rput(100,0){

\rput(-17.5,32){
\pscircle(35,0){2.5}  \rput(35,0){\psline{-}(0,-1.25)(0,1.25) \psline{-}(-1.25,0)(1.25,0)}
\psline{->}(35,7.5)(35,2.5) \rput(37,11){$\sum a_{1k}\mathbf{d}_k$}
\psline{->}(37.5,0)(47,0)
\psframe(47,-5)(60,5) \rput(53.5,0){$Q_{\Lambda_1}$}
\psline{->}(60,0)(66.5,0)
\psframe(66.5,-5)(82,5) \rput(73,0){\footnotesize{$\mod\Lambda$}}
\psline{->}(82,0)(117.5,0) \rput(88,3.7){$\mathbf{\hat{v}}_1$}
\rput(25.5,0){
\psline{->}(104,0)(109.5,0)
\psframe(109.5,-5)(125,5) \rput(116,0){\footnotesize{$\mod\Lambda$}}
\psline{->}(125,0)(135,0) \rput(129.5,3.7){$\mathbf{\hat{t}}_1$}
\rput(35,0){
\psframe(100,-5)(112,5) \rput(106,0){$\mathcal{L}_1^{-1}$}
\psline{->}(112,0)(117,0)\rput(120.5,0.3){${\hat{w}}_1$}
}
}
}

\rput(-17.5,12){
\pscircle(35,0){2.5}  \rput(35,0){\psline{-}(0,-1.25)(0,1.25) \psline{-}(-1.25,0)(1.25,0)}
\psline{->}(35,7.5)(35,2.5) \rput(37,11){$\sum a_{2k}\mathbf{d}_k$}
\psline{->}(37.5,0)(50.5,0) 
\pscircle(53,0){2.5}\psline(53,-1.25)(53,1.25) \psline(51.75,0)(54.25,0)
\psline{<-}(53,2.5)(53,12)(98,12)(98,20)
\pscircle[fillstyle=solid,fillcolor=black](98,20){1}
\rput(58,8.5){$r_{21}$}
\psline{->}(55.5,0)(63,0)
\rput(15,0){
\psframe(48,-5)(61,5) \rput(54.5,0){$Q_{\Lambda_{2}}$}
\psline{->}(61,0)(66.5,0)
\pscircle[fillstyle=solid,fillcolor=black](69,12){1}
\psline{->}(69,12)(69,2.5)
\rput(75,8.5){$-r_{21}$}
\pscircle(69,0){2.5}\psline(69,-1.25)(69,1.25) \psline(67.75,0)(70.25,0)
\psline{->}(71.5,0)(77,0)
\rput(10.5,0){
\psframe(66.5,-5)(82,5) \rput(73,0){\footnotesize{$\mod\Lambda$}}
\psline{->}(82,0)(92,0) \rput(86,3.7){$\mathbf{\hat{v}}_2$}
\psline{->}(104,0)(109.5,0)
\psframe(109.5,-5)(125,5) \rput(116,0){\footnotesize{$\mod\Lambda$}}
\psline{->}(125,0)(135,0) \rput(129.5,3.7){$\mathbf{\hat{t}}_2$}
\rput(35,0){
\psframe(100,-5)(112,5) \rput(106,0){$\mathcal{L}_2^{-1}$}
\psline{->}(112,0)(117,0)\rput(120.5,0.3){${\hat{w}}_2$}
}
}
}
}
\rput(52,2.5){$\vdots$}
\rput(82,2.5){$\vdots$}
\rput(124.5,-1){$\vdots$}
\rput(149,-1){$\vdots$}

\rput(-17.5,-18){
\pscircle(35,0){2.5}  \rput(35,0){\psline{-}(0,-1.25)(0,1.25) \psline{-}(-1.25,0)(1.25,0)}
\psline{->}(35,7.5)(35,2.5) \rput(36,11){$\sum a_{Kk}\mathbf{d}_k$}
\psline{->}(37.5,0)(50.5,0) 
\pscircle(53,0){2.5}\psline(53,-1.25)(53,1.25) \psline(51.75,0)(54.25,0)
\psline{<-}(53,2.5)(53,7.5) \rput(60.25,11){$\sum r_{K\ell}\mathbf{\hat{v}}_\ell$}
\psline{->}(55.5,0)(63,0)
\rput(15,0){
\psframe(48,-5)(61,5) \rput(55,0){$Q_{\Lambda_{K}}$}
\psline{->}(61,0)(66.5,0)
\psline{->}(69,7.5)(69,2.5)
 \rput(73.75,11){$-\sum r_{K\ell}\mathbf{\hat{v}}_\ell$}
\pscircle(69,0){2.5}\psline(69,-1.25)(69,1.25) \psline(67.75,0)(70.25,0)
\psline{->}(71.5,0)(77,0)
\rput(10.5,0){
\psframe(66.5,-5)(82,5) \rput(73,0){\footnotesize{$\mod\Lambda$}}
\psline{->}(82,0)(92,0) \rput(86.5,3.7){$\mathbf{\hat{v}}_K$}
\psline{->}(104,0)(109.5,0)
\psframe(109.5,-5)(125,5) \rput(116,0){\footnotesize{$\mod\Lambda$}}
\psline{->}(125,0)(135,0) \rput(130,3.7){$\mathbf{\hat{t}}_K$}
\rput(35,0){
\psframe(100,-5)(112,5) \rput(105.5,0){$\mathcal{L}_K^{-1}$}
\psline{->}(112,0)(116.5,0)\rput(120.5,0.3){${\hat{w}}_K$}
}

}
}
}

}
\end{pspicture}
\end{center}
\caption{System diagram of the nested lattice encoding and decoding operations employed as part of the compute-and-forward transform. Each message $w_k$ is mapped to a lattice codeword $\mathbf{t}_k$ according to codebook $\mathcal{L}_k$, dithered, and transmitted as $\mathbf{x}_k$. The multiple-access channel scales codeword $k$ by $h_k$ and outputs the sum plus Gaussian noise $\mathbf{z}$. The decoder attempts to recover a linearly independent set of $K$ integer linear combinations with coefficients $\mathbf{A} = \{a_{mk}\}$. For the figure, we have assumed that $R_1 \geq R_2 \geq \cdots \geq R_K$ and that $R_m < R_{\text{comp}}(\mathbf{h},\mathbf{a}_m,\beta_m)$. To decode the first linear combination $\mathbf{v}_1 = [\sum a_{1k} \mathbf{t}_k] \Mod$, the receiver scales $\mathbf{y}$ by $\beta_1$, removes the dithers, quantizes using $Q_{\Lambda_1}$, and takes $\hspace{-0.05in}\mod\Lambda$. For the second linear combination $\mathbf{v}_2 = [\sum a_{2k} \mathbf{t}_k] \Mod$, the decoder scales by $\beta_2$, removes the dithers, and then eliminates the lattice point $\mathbf{t}_1$ using its estimate of the first linear combination $\mathbf{\hat{v}}_1$ so that the rate of the remaining lattice points is at most $R_2$. It then quantizes using $Q_{\Lambda_2}$, adds back in $\mathbf{\hat{v}}_1$, and takes $\hspace{-0.05in}\mod\Lambda$. Decoding proceeds in this fashion, using a form of successive interference cancellation to keep the rates of the lattice points below the computation rates. Afterwards, the receiver solves for the original lattice points by multiplying by $\mathbf{A}_p^{-1}$, which is the inverse of $\mathbf{A}$ over $\mathbb{Z}_p$, and taking $\hspace{-0.05in}\mod\Lambda$. Finally, it maps these estimates $\mathbf{\hat{t}}_k$ of the transmitted lattice points back to the corresponding messages.} \label{f:cftransform}
\end{figure*}
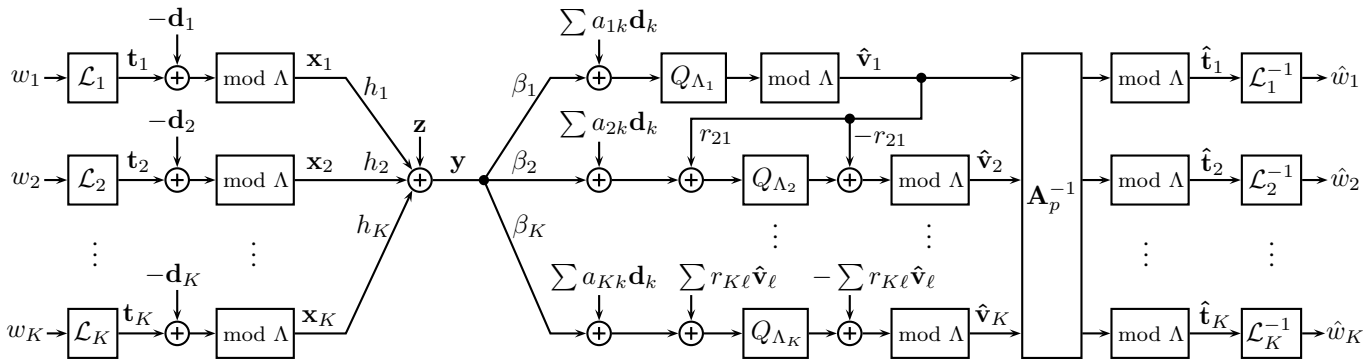
\begin{figure*}[!t]
\begin{center}
\psset{unit=0.68mm}
\begin{pspicture}(9,-23)(240,50)

\rput(-0.5,32){
\rput(6,0){$w_1$}\psline{->}(9.5,0)(14,0)
\psframe(14,-5)(24,5) \rput(19,0){$\mathcal{L}_1$}
\psline{->}(24,0)(36,0) \rput(30,3.6){$\mathbf{t}_1$}
\psline{->}(48,0)(53,0)\psframe(53,-5)(68.5,5) \rput(59,0){\footnotesize{$\mod\Lambda$}}
\psline{->}(68.5,0)(80,0) \rput(74,3.6){$\mathbf{v}_1$}
}
\rput(-0.5,12){
\rput(6,0){$w_2$}\psline{->}(9.5,0)(14,0)
\psframe(14,-5)(24,5) \rput(19,0){$\mathcal{L}_2$}
\psline{->}(24,0)(36,0) \rput(30,3.6){$\mathbf{t}_2$}
\psline{->}(48,0)(53,0)\psframe(53,-5)(68.5,5) \rput(59,0){\footnotesize{$\mod\Lambda$}}
\psline{->}(68.5,0)(80,0) \rput(74,3.6){$\mathbf{v}_2$}
}
\rput(19,-1){$\vdots$}

\rput(60.25,-1){$\vdots$}

\rput(-0.5,-18){
\rput(6,0){$w_K$}\psline{->}(10,0)(14,0)
\psframe(14,-5)(24,5) \rput(19,0){$\mathcal{L}_K$}
\psline{->}(24,0)(36,0) \rput(30,3.6){$\mathbf{t}_K$}
\psline{->}(48,0)(53,0)\psframe(53,-5)(68.5,5) \rput(59,0){\footnotesize{$\mod\Lambda$}}
\psline{->}(68.5,0)(80,0) \rput(74,3.6){$\mathbf{v}_K$}
}

\psframe(35.5,-23)(47.5,37) \rput(41.5,9){$\mathbf{A}$}

\psframe(179.5,-23)(191.5,37) \rput(185.5,8.5){$\mathbf{A}_p^{-1}$}

\rput(79.5,0){

\rput(-17.5,32){
\pscircle(20,0){2.5}  \rput(20,0){\psline{-}(0,-1.25)(0,1.25) \psline{-}(-1.25,0)(1.25,0)}
\psline{->}(20,7.5)(20,2.5) \rput(30,11){$\mathbf{z}_{\text{eff}}(\mathbf{h},\mathbf{a}_1,\beta_1)$}
\psline{->}(22.5,0)(47,0) 
\psframe(47,-5)(60,5) \rput(53.5,0){$Q_{\Lambda_1}$}
\psline{->}(60,0)(66.5,0)
\psframe(66.5,-5)(82,5) \rput(73,0){\footnotesize{$\mod\Lambda$}}
\psline{->}(82,0)(117.5,0) \rput(88,3.7){$\mathbf{\hat{v}}_1$}
\rput(25.5,0){
\psline{->}(104,0)(109.5,0)
\psframe(109.5,-5)(125,5) \rput(116,0){\footnotesize{$\mod\Lambda$}}
\psline{->}(125,0)(135,0) \rput(129.5,3.7){$\mathbf{\hat{t}}_1$}
\rput(35,0){
\psframe(100,-5)(112,5) \rput(106,0){$\mathcal{L}_1^{-1}$}
\psline{->}(112,0)(117,0)\rput(120.5,0.3){${\hat{w}}_1$}
}
}
}

\rput(-17.5,12){
\pscircle(20,0){2.5}  \rput(20,0){\psline{-}(0,-1.25)(0,1.25) \psline{-}(-1.25,0)(1.25,0)}
\psline{->}(20,7.5)(20,2.5) \rput(30,11){$\mathbf{z}_{\text{eff}}(\mathbf{h},\mathbf{a}_2,\beta_2)$}
\psline{->}(22.5,0)(50.5,0) 
\pscircle(53,0){2.5}\psline(53,-1.25)(53,1.25) \psline(51.75,0)(54.25,0)
\psline{<-}(53,2.5)(53,12)(98,12)(98,20)
\pscircle[fillstyle=solid,fillcolor=black](98,20){1}
\rput(58,8.5){$r_{21}$}
\psline{->}(55.5,0)(63,0)
\rput(15,0){
\psframe(48,-5)(61,5) \rput(54.5,0){$Q_{\Lambda_{2}}$}
\psline{->}(61,0)(66.5,0)
\pscircle[fillstyle=solid,fillcolor=black](69,12){1}
\psline{->}(69,12)(69,2.5)
\rput(75,8.5){$-r_{21}$}
\pscircle(69,0){2.5}\psline(69,-1.25)(69,1.25) \psline(67.75,0)(70.25,0)
\psline{->}(71.5,0)(77,0)
\rput(10.5,0){
\psframe(66.5,-5)(82,5) \rput(73,0){\footnotesize{$\mod\Lambda$}}
\psline{->}(82,0)(92,0) \rput(86,3.7){$\mathbf{\hat{v}}_2$}
\psline{->}(104,0)(109.5,0)
\psframe(109.5,-5)(125,5) \rput(116,0){\footnotesize{$\mod\Lambda$}}
\psline{->}(125,0)(135,0) \rput(129.5,3.7){$\mathbf{\hat{t}}_2$}
\rput(35,0){
\psframe(100,-5)(112,5) \rput(106,0){$\mathcal{L}_2^{-1}$}
\psline{->}(112,0)(117,0)\rput(120.5,0.3){${\hat{w}}_2$}
}
}
}
}
\rput(52,2.5){$\vdots$}
\rput(82,2.5){$\vdots$}
\rput(124.5,-1){$\vdots$}
\rput(149,-1){$\vdots$}

\rput(-17.5,-18){
\pscircle(20,0){2.5}  \rput(20,0){\psline{-}(0,-1.25)(0,1.25) \psline{-}(-1.25,0)(1.25,0)}
\psline{->}(20,7.5)(20,2.5) \rput(31,11){$\mathbf{z}_{\text{eff}}(\mathbf{h},\mathbf{a}_K,\beta_K)$}
\psline{->}(22.5,0)(50.5,0) 
\pscircle(53,0){2.5}\psline(53,-1.25)(53,1.25) \psline(51.75,0)(54.25,0)
\psline{<-}(53,2.5)(53,7.5) \rput(60.25,11){$\sum r_{K\ell}\mathbf{\hat{v}}_\ell$}
\psline{->}(55.5,0)(63,0)
\rput(15,0){
\psframe(48,-5)(61,5) \rput(55,0){$Q_{\Lambda_{K}}$}
\psline{->}(61,0)(66.5,0)
\psline{->}(69,7.5)(69,2.5)
 \rput(73.75,11){$-\sum r_{K\ell}\mathbf{\hat{v}}_\ell$}
\pscircle(69,0){2.5}\psline(69,-1.25)(69,1.25) \psline(67.75,0)(70.25,0)
\psline{->}(71.5,0)(77,0)
\rput(10.5,0){
\psframe(66.5,-5)(82,5) \rput(73,0){\footnotesize{$\mod\Lambda$}}
\psline{->}(82,0)(92,0) \rput(86.5,3.7){$\mathbf{\hat{v}}_K$}
\psline{->}(104,0)(109.5,0)
\psframe(109.5,-5)(125,5) \rput(116,0){\footnotesize{$\mod\Lambda$}}
\psline{->}(125,0)(135,0) \rput(130,3.7){$\mathbf{\hat{t}}_K$}
\rput(35,0){
\psframe(100,-5)(112,5) \rput(105.5,0){$\mathcal{L}_K^{-1}$}
\psline{->}(112,0)(116.5,0)\rput(120.5,0.3){${\hat{w}}_K$}
}

}
}
}

}
\end{pspicture}
\end{center}
\caption{Effective MIMO channel induced by the compute-and-forward transform of a Gaussian multiple-access channel. The channel output $\mathbf{y} = \sum h_k \mathbf{x}_k + \mathbf{z}$ is converted into a linearly independent set of $K$ integer linear combinations $\mathbf{v}_m = [\sum a_{mk} \mathbf{t}_k] \Mod$ plus effective noise $\mathbf{z}_{\text{eff}}(\mathbf{h},\mathbf{a}_m,\beta_m) = \beta_m \mathbf{z} + \sum (\beta_m h_k - a_{mk}) \mathbf{x}_k$. As in Figure \ref{f:cftransform}, these linear combinations can be decoded using a version of successive cancellation.} \label{f:aftertransform}
\end{figure*}
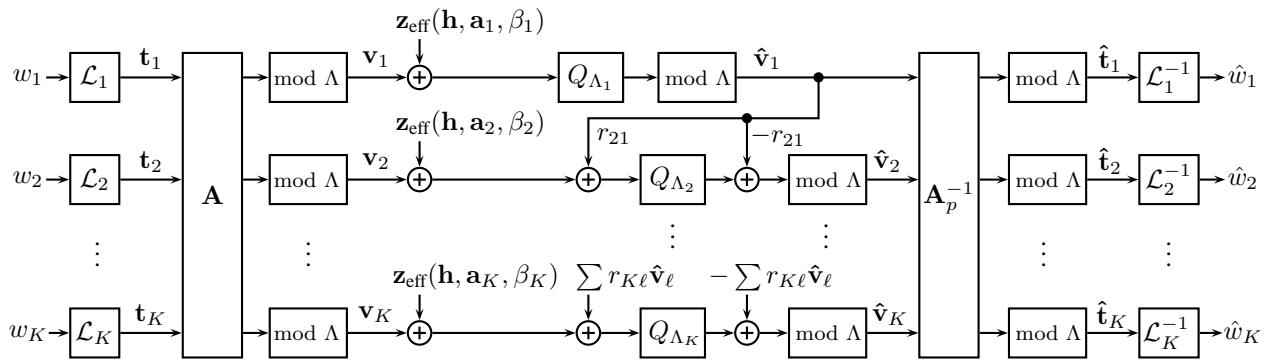

Each user $k$ maps its message to a lattice point $\bt_k$ in its codebook $\mathcal{L}_k$ and transmits a dithered version of it. The $K$ lattice codebooks utilized by the different users form a chain of nested lattices as in~\eqref{nestedChain}. Assume for now that the users are ordered with descending rates \mbox{$R_1\geq R_2\geq\cdots\geq R_K$}, i.e., $\map(k)=k$ for $k=1,\ldots,K$.
The receiver, which sees a noisy real-valued linear combination of the transmitted codewords, begins by decoding the integer linear combination \mbox{$\mathbf{v}_1 = [\sum a_{1k} \mathbf{t}_k] \Mod$} that yields the highest computation rate $R_{\text{comp},1}$. Using the compute-and-forward framework, this is possible if $R_1<R_{\text{comp},1}$.
Then, it proceeds to decode the integer linear combination \mbox{$\mathbf{v}_2 = [\sum a_{2k} \mathbf{t}_k] \Mod$} that yields the second highest computation rate $R_{\text{comp},2}$. In general, $\bt_1$ participates in this linear combination and the condition for correct decoding of $\bv_2$ is therefore $R_1<R_{\text{comp},2}$. Nevertheless, this condition can be relaxed using the linear combination $\bv_1$ that has already been decoded. Specifically, after scaling of the channel output and removing the dithers, the receiver has a noisy observation
\begin{align}
\bs_2=\left[\bv_2+\bz_{\text{eff}}(\bh,\ba_2)\right]\Mod\nonumber
\end{align}
of the desired linear combination $\bv_2$. If $\bt_1$ participates in $\bv_1$, it is possible to cancel out $\bt_1$ from the second linear combination by adding a scaled version of $\bv_1$ to $\bs_2$. Namely, the receiver adds $r_{21}\bv_1$ to $\bs_2$,
where $r_{21}$ is an integer chosen such that $\left[(a_{21}+r_{21}a_{11})\right]\bmod p=0$, which assures that $\left[(a_{21}+r_{21}a_{11})\bt_1\right]\Mod=\mathbf{0}$ for any $\bt_1\in\mathcal{L}_1$. After reducing\!$\mod\Lambda$, this yields
\begin{align}
\bs_2^{\text{SI}}&=\left[\bv_2+r_{21}\bv_1+\bz_{\text{eff}}(\bh,\ba_2)\right]\Mod\nonumber\\
&=\left[\mathbf{\tilde{v}}_2+\bz_{\text{eff}}(\bh,\ba_2)\right]\Mod,\nonumber \\
\mathbf{\tilde{v}}_2 &= \left[ \sum_{k = 2}^K (a_{2k} + r_{21} a_{1k}) \bt_k \right]\Mod \nonumber \ .
\end{align}
Note that $\bt_1$ does not participate in $\mathbf{\tilde{v}}_2$. Since the effective noise $\bz_{\text{eff}}(\bh,\ba_2)$ is unchanged by this process, the receiver can decode $\mathbf{\tilde{v}}_2$ as long as $R_2<R_{\text{comp},2}$. Now, the receiver can obtain $\bv_2$ by subtracting $r_{21}\bv_1$ from $\mathbf{\tilde{v}}_2$ and reducing\!$\mod\Lambda$.\footnote{The operation of extracting $\bv_2$ from $\mathbf{\tilde{v}}_2$ is in fact not necessary as the receiver is only interested in decoding \emph{any} linearly independent set of $K$ integer linear combinations. We describe this step only to simplify the exposition of the scheme.}
The receiver decodes the remaining linear combinations in a similar manner, i.e., before decoding the $m$th linear combination $\bv_m$ with computation rate $R_{\text{comp},m}$ the receiver adds to
\begin{align}
\bs_m=\left[\bv_m+\bz_{\text{eff}}(\bh,\ba_m) \right]\Mod\nonumber
\end{align}
an integer linear combination $\left[\sum_{\ell=1}^{m-1}r_{m\ell}\bv_\ell\right]\Mod$ of its previously decoded linear combinations. The coefficients $r_{m1},\ldots,r_{m,m-1} \in \mathbb{Z}$ are chosen such that the effect of $\bt_1,\ldots,\bt_{m-1}$ is canceled out from $\bv_m$. Assuming that such coefficients exist, the receiver can decode $\mathbf{\tilde{v}}_m=\left[\bv_m+\sum_{\ell=1}^{m-1}r_{m\ell}\bv_\ell\right]\Mod$ as long as $R_m<R_{\text{comp},m}$.

Lemma~\ref{lem:ModuloTriangularization}, stated in Appendix~\ref{app:MACproof}, establishes that for any set of $K$ linearly independent coefficient vectors $\{\ba_1,\ldots,\ba_K\}$ there indeed always exist integer-valued coefficients $\{r_{m\ell}\}$ such that in the $m$th decoding step the receiver can cancel out $m-1$ lattice points from the desired linear combination $\bv_m$, using the previously decoded linear combinations $\left\{\bv_1,\ldots,\bv_{m-1}\right\}$. The procedure for finding these coefficients is reminiscent of the Gaussian elimination procedure of a full-rank matrix.
One of the basic operations in Gaussian elimination is row switching. In our considerations, this would correspond to using a linear combination that has not been decoded yet for eliminating lattice points from another linear combination, which is clearly not possible. Therefore, a major difference between our procedure for finding a good set of coefficients $\{r_{ij}\}$ and Gaussian elimination is that row switching is not permitted. This will sometimes constrain the order in which we can cancel out users from linear combinations. Nevertheless, there always exists at least one valid successive cancellation order. In other words, we can always cancel out the effect of $m-1$ users from $\bv_m$ using the decoded linear combination $\left\{\bv_1,\ldots,\bv_{m-1}\right\}$, but we cannot always control which of the $K$ users to cancel. As a result, there always exists at least one permutation vector $\mathbf{\pi}$ such that all $K$ linear combination can be decoded as long as
\begin{align}
R_{\pi(m)}<R_{\text{comp},m}, \ m=1,\ldots,K.
\end{align}
It follows that a sum-rate of $\sum_{m=1}^K R_{\text{comp},m}$ is achievable over the $K$-user MAC with our scheme, in which all users employ nested lattice codebooks. As we shall see, this sum rate is within a constant gap, smaller than $\nicefrac{K}{2}\log(K)$ bits, from the sum capacity of the MAC, for all channel gains and SNR.

\subsection{The Compute-and-Forward Transform}\label{sub:coftransform}
We first introduce a transformation of a MAC to a multiple-input multiple-output (MIMO) mod-$\Lambda$ channel, where the $K\times K$ channel matrix is integer-valued. This transformation, dubbed the \emph{compute-and-forward transform}, will play an important role in our decoding scheme for the interference channel.

\vspace{2mm}

\begin{definition}
Let $\{\ba_1,\ldots,\ba_K\}$ be a set of optimal integer coefficient vectors (see Definition~\ref{def:optEqs}), $\beta_1,\ldots,\beta_K$ the corresponding optimal scaling factors, and $R_{\text{comp},1}\geq\cdots\geq R_{\text{comp},K}$ the corresponding optimal computation rates. We define the \textit{compute-and-forward transform} of the MAC with nested lattice codes as
\begin{align}
{\bS}=\left(
      \begin{array}{c}
        \bs_1 \\
        \vdots \\
        \bs_K \\
      \end{array}
    \right)
&=\left(\begin{array}{c}
            \left[\beta_1\by+\sum_{k=1}^K a_{1k}\bd_k\right]\bmod\Lambda \\
            \vdots \\
            \left[\beta_K\by+\sum_{k=1}^K a_{Kk}\bd_k\right]\bmod\Lambda
          \end{array}
\right)\nonumber\\
&=\left[\bA\ \left(
                       \begin{array}{c}
                         \bt_1 \\
                         \vdots \\
                         \bt_K \\
                       \end{array}
                     \right)+{\bZ}_{\text{eff}} \right]\bmod\Lambda,
\end{align}
where we have written the channel output $\mathbf{y}$, dithers $\mathbf{d}_k$, and lattice codewords $\mathbf{t}_k$ as length-$n$ row vectors. We also denote $\bA=[\ba_1 \ \cdots\  \ba_K]^T$ and ${\bZ}_{\text{eff}}=[\bz_{\text{eff},1}^T\  \cdots \ \bz_{\text{eff},K}^T]^T$.
\end{definition}

\vspace{2mm}

\begin{remark}
The transform is not unique as the set of optimal integer coefficient vectors is not unique. Nevertheless, the set of optimal computation rates is unique. As we shall see, the set of optimal computation rates dictates the rates attained over the transformed channel. Therefore, we use the term \emph{the} compute-and-forward transform of the channel, with the understanding that although there may be multiple options for the transform, they are all equivalent.
\end{remark}

\vspace{2mm}

The $m$th output $\bs_m$ of the transformed channel corresponds to an integer linear combination plus effective noise. Due to Theorem~\ref{thm:CoF}, each such linear combination can be reliably decoded as long as all lattice points participating in it belong to codes of rates smaller than $R_{\text{comp},m}$. We now lower bound the sum of $K$ optimal computation rates, and in the sequel we show that this sum can be translated to a valid MAC sum rate.

\vspace{2mm}

\begin{theorem}
The sum of optimal computation rates is lower bounded by
\begin{align}
\sum_{m=1}^K R_{\text{comp},m}\geq\frac{1}{2}\log\left(1+\|\mathbf{h}\|^2\Tsnr\right)-\frac{K}{2}\log(K) \ .
\label{sumRate3}
\end{align}
\label{thm:SumRate}
\end{theorem}

\vspace{1mm}

The proof makes use of the following well-known theorem due to Minkowski~\cite[Theorem 1.5]{mg02}, that upper bounds the product of successive minima.

\vspace{2mm}

\begin{theorem}[Minkowski]
For any lattice $\Lambda(\mathbf{F})$ which is spanned by a full-rank $K\times K$ matrix $\bF$
\begin{align}
\prod_{m=1}^K \lambda_m^2(\mathbf{F})\leq K^K \left|\det(\bF)\right|^2.
\end{align}
\label{thm:Minkowski}
\end{theorem}

We are now ready to prove Theorem~\ref{thm:SumRate}.

\vspace{2mm}

\begin{proof}[Proof of Theorem~\ref{thm:SumRate}]
Let $\Lambda(\mathbf{F})$ be a lattice spanned by the matrix $\bF$ from~\eqref{latticeMat}, and let $\lambda_1(\bF),\ldots,\lambda_K(\bF)$ be its $K$ successive minima. Let $\ba_1, \ldots, \ba_K \in \mathbb{Z}^K$ denote the optimal coefficient vectors.
By Definition~\ref{def:optEqs} and~\eqref{effvarlattice} we have $\|\bF\ \ba_m\|=\lambda_m(\bF)$ for $m=1,\ldots,K$.
The sum of optimal computation rates is
\begin{align}
\sum_{m=1}^K R_{\text{comp},m}&=\sum_{m=1}^K R_{\text{comp}}(\mathbf{h},\mathbf{a}_m)\nonumber\\
&=\sum_{m=1}^K \frac{1}{2}\log\left(\frac{\Tsnr}{\sigma^2_{\text{eff}}(\mathbf{h},\mathbf{a}_m)}\right)\nonumber\\
&=\frac{K}{2}\log\left(\Tsnr\right)-\frac{1}{2}\log\left(\prod_{m=1}^K \left\|\bF \ \mathbf{a}_m\right\|^2\right)\nonumber\\
&=\frac{K}{2}\log\left(\Tsnr\right)-\frac{1}{2}\log\left(\prod_{m=1}^K \lambda^2_m(\bF)\right).\nonumber
\end{align}
Applying Theorem~\ref{thm:Minkowski} to the product $\prod_{m=1}^K \lambda^2_m(\bF)$ yields
\begin{align}
\sum_{m=1}^K R_{\text{comp},m}\geq\frac{K}{2}\log\left(\Tsnr\right)-\frac{1}{2}\log\left(K^K \left|\det(\bF)\right|^2\right). \label{sumRate2a}
\end{align}
Using Sylvester's determinant identity (see e.g.,~\cite{harville}) $$\det(\bI_{K \times K}+\Tsnr \ \bh\bh^T)=\det(1+\|\bh\|^2\Tsnr),$$ we have that
\begin{align}
\left|\det(\bF)\right|^2=\frac{\Tsnr^{K}}{1+\|\mathbf{h}\|^2\Tsnr} \ .\label{detG}
\end{align}
Substituting~\eqref{detG} into~\eqref{sumRate2a} proves the theorem.
\end{proof}

\vspace{1mm}

\begin{remark}
It is possible to avoid the loss of the constant factor $\nicefrac{K}{2}\log{K}$ in~\eqref{sumRate3} using successive compute-and-forward, as described in~\cite{nazer12IZS,oen13}. However, in this case the operational interpretation of the sum of computation rates becomes more involved than that described in the sequel. See~\cite{oen13} for more details.
\end{remark}

\vspace{1mm}

Next, we give an operational meaning to the $K$ optimal computation rates.

\subsection{Multiple-Access Sum Capacity to within a Constant Gap}\label{sub:mac2}

We now show that the compute-and-forward transform can be used for achieving several rate tuples within a constant gap from the boundary of the capacity region of the $K$-user MAC. To establish this result, we introduce a decoding technique that we will refer to as \emph{algebraic successive cancellation}. Namely, each decoded linear combination will be used to cancel out the effect of one user from the linear combinations that have yet to be decoded. We first illustrate the coding scheme by an example, and then formalize our result in Theorem~\ref{thm:MAC}.
\begin{example}
\label{ex:successive}
Consider the two-user MAC
\begin{align}
\by=\sqrt{5}\bx_1+\bx_2+\bz,\nonumber
\end{align}
at $\Tsnr=15$dB. It can be shown using~\eqref{optVar} and~\eqref{compRate2} that the compute-and-forward transform of this channel is
\begin{align}
\left(
  \begin{array}{c}
    \bs_1 \\
    \bs_2 \\
  \end{array}
\right)
=\left[\left(
            \begin{array}{cc}
              2 & 1 \\
              3 & 1 \\
            \end{array}
          \right)\left(
                        \begin{array}{c}
                          \bt_1 \\
                          \bt_2 \\
                        \end{array}
                      \right)+\left(
                                \begin{array}{c}
                                  \bz_{\text{eff},1} \\
                                  \bz_{\text{eff},2} \\
                                \end{array}
                              \right)
 \right]\Mod\nonumber
\end{align}
with $R_{\text{comp},1}\simeq2.409$ bits and $R_{\text{comp},2}\simeq1.372$ bits. Note that $(R_{\text{comp},1}+R_{\text{comp},2})/(1/2\log(1+\|\bh\|^2\Tsnr))\simeq0.998$.
We use a chain of three nested lattices $\Lambda\subseteq\Lambda_2\subseteq\Lambda_1$ that satisfy the conditions of Theorem~\ref{thm:CoF} in order to construct the codebooks $\mathcal{L}_1=\Lambda_1\cap\CV$ with rate $R_1$ arbitrarily close to $R_{\text{comp},1}$ for user $1$ and $\mathcal{L}_2=\Lambda_2\cap\CV$ with rate $R_2$ arbitrarily close to $R_{\text{comp},2}$ for user $2$.

From Theorem~\ref{thm:CoF}\eqref{thm:CoF1}, we know that \mbox{$\bv_1=[2\bt_1+\bt_2]\Mod$} can be decoded from $\bs_1$ since $R_1$ and $R_2$ are smaller than $R_{\text{comp},1}$. However, Theorem~\ref{thm:CoF} does not guarantee that $\bv_2=[3\bt_1+\bt_2]\Mod$ can be decoded directly from $\bs_2$ since the first user employs a codebook with a rate $R_1 \approx R_{\text{comp},1}$ which is higher than the second computation rate $R_{\text{comp},2}$. To circumvent this issue, we use the estimate $\mathbf{\hat{v}}_1$ of the linear combination $\mathbf{\hat{v}}_1$ as side information in order to cancel out the lattice point $\bt_1\in\Lambda_1\cap\CV$ from $\bs_2$.
Note that Theorem~\ref{thm:CoF}\eqref{thm:CoF3} guarantees that $[p\cdot\bt_k]\Mod=\mathbf{0}$, $k=1,2$ for some sufficiently large prime number $p$.
Let $2^{-1}\in\mathbb{Z}$ be an integer that satisfies $[2^{-1}\cdot 2]\bmod p=1$. The receiver computes
\begin{align}
\bs_2^{\text{SI}}&=\left[\bs_2-3\cdot 2^{-1}\mathbf{\hat{v}}_1 \right]\Mod\nonumber\\
&\stackrel{(a)}{=}\left[(3-3\cdot2^{-1}\cdot 2)\bt_1+(1-3\cdot 2^{-1})\bt_2+\bz_{\text{eff},2}\right]\Mod\nonumber\\
&\stackrel{(b)}{=}\left[[1-3\cdot2^{-1}]\bmod p\cdot\bt_2+\bz_{\text{eff},2}\right]\Mod\nonumber\\
&=\left[\tilde{a}_{12}\cdot\bt_2+\bz_{\text{eff},2}\right]\Mod,\label{exampleEq}
\end{align}
where $\tilde{a}_{12}=[1\ -\ 3\cdot2^{-1}]\bmod p$. Step $(a)$ in~\eqref{exampleEq} follows from the distributive law. Step $(b)$ follows since $3-3\cdot2^{-1}\cdot 2=M\cdot p$ for some $M\in\ZZ$. Thus,
\begin{align}
\left[(3-3\cdot2^{-1}\cdot 2)\bt_1 \right]\Mod&=\left[M\cdot p \cdot \bt_1 \right]\Mod\nonumber\\
&=\left[M\cdot [p\cdot\bt_1]\Mod \right]\Mod\nonumber\\
&=\mathbf{0},\nonumber
\end{align}
where the last equality is justified by Theorem~\ref{thm:CoF}\eqref{thm:CoF3}.

Now only $\bt_2$ participates in the linear combination $\mathbf{\tilde{v}}_2=[\tilde{a}_{12}\bt_2]\Mod$ and, since $R_2$ is smaller than $R_{\text{comp},2}$, Theorem~\ref{thm:CoF} guarantees that it can be decoded from $\bs_2^{\text{SI}}$. This is accomplished by quantizing onto $\Lambda_2$ and reducing modulo $\Lambda$, $$\mathbf{\hat{\tilde{v}}}_2=\big[Q_{\Lambda_2}(\bs_2^{\text{SI}})\big]\Mod.$$
After decoding both linear combinations $\bv_1$ and $\mathbf{\tilde{v}}_2$ the receiver can solve for the transmitted lattice points $\bt_1$ and $\bt_2$, as the two linear combinations are full-rank over $\Zp$. We have therefore shown that the rate region $R_1<R_{\text{comp},1}$ and $R_2<R_{\text{comp},2}$ is achievable. In a similar manner, we can show that the rate region $R_1<R_{\text{comp},2}$ and $R_2<R_{\text{comp},1}$ is achievable with this scheme.
\end{example}

\vspace{3mm}

In order to formally characterize the achievable rate region, we will need the following definition which identifies the orders for which algebraic successive cancellation can be performed.
\begin{definition}
For a full-rank $K\times K$ matrix $\bA$ with integer-valued entries we define the \emph{pseudo-triangularization} process,
which transforms the matrix $\bA$ to a matrix $\mathbf{\tilde{A}}$ which is upper triangular up to column permutation $\mathbf{\pi}=\left[\pi(1) \ \pi(2) \ \cdots  \ \pi(K)\right]$. This is accomplished by left-multiplying $\bA$ by a lower triangular matrix $\bL$ with unit diagonal, such that $\mathbf{\tilde{A}}=\bL\bA$ is upper triangular up to column permutation $\mathbf{\pi}$. Although the matrix $\bA$ is integer valued, the matrices $\bL$ and $\mathbf{\tilde{A}}$ need not necessarily be integer valued.
Note that the pseudo-triangularization process is reminiscent of Gaussian elimination except that row switching and row multiplication are prohibited. It is also closely connected to the LU decomposition where only column pivoting is permitted.
\end{definition}

\begin{example}
\label{ex:pseudoTriEx}
The $2\times 2$ matrix
\begin{align}
\bA=\left(
      \begin{array}{cc}
        2 & 1  \\
        3 & 1  \\
      \end{array}
    \right)\nonumber
\end{align}
from Example~\ref{ex:successive}
can be pseudo-triangularized with two different permutation vectors
\begin{align}
\mathbf{\tilde{A}}&=\left(
      \begin{array}{cc}
        1 & 0  \\
        -\frac{3}{2} & 1  \\
      \end{array}
    \right)\cdot\bA\nonumber\\
    &=\left(
      \begin{array}{cc}
        2 & 1 \\
        0 & -\frac{1}{2} \\
      \end{array}
       \right), \ \ \ \ \mathbf{\pi}=[1 \ 2],\nonumber
       \end{align}
or
       \begin{align}
\mathbf{\tilde{A}}&=\left(
      \begin{array}{cc}
        1 & 0  \\
        -1 & 1  \\
      \end{array}
    \right)\cdot\bA\nonumber\\
    &=\left(
      \begin{array}{cc}
        2 & 1 \\
        1 & 0 \\
      \end{array}
       \right), \ \ \ \ \mathbf{\pi}=[2 \ 1].\nonumber
\end{align}
\end{example}
\begin{remark}
Any full-rank matrix can be triangularized using the Gaussian elimination process, and therefore any full-rank matrix can be pseudo-triangularized with at least one permutation vector $\mathbf{\pi}$. In particular, since for any MAC the integer-valued matrix $\bA$ from the compute-and-forward transform is full-rank, it can always be pseudo-triangularized with at least one permutation vector $\mathbf{\pi}$. There are full-rank matrices that can be pseudo-triangularized with several different permutation vectors, such as $\bA$ from Example~\ref{ex:pseudoTriEx}. However, there are also full-rank matrices $\bA$ that can be pseudo-triangularized with only one permutation vector $\mathbf{\pi}$. An example of such a matrix is the identity matrix $\bI_{K\times K}$.
\end{remark}
\vspace{1mm}

The next theorem gives an achievable rate region for the MAC under the compute-and-forward transform. The proof is given in Appendix~\ref{app:MACproof} and follows along the same lines as Example~\ref{ex:successive} .
\begin{theorem}
Consider the MAC~\eqref{MACchannel}.
For any $\epsilon>0$ and $n$ large enough, there exists a chain of $n$-dimensional nested lattices $\Lambda\subseteq\Lambda_K\subseteq\cdots\subseteq\Lambda_1$ forming the set of codebooks $\mathcal{L}_1,\ldots,\mathcal{L}_K$ with rates $R_1,\ldots,R_K$ such that for all $\bh\in\RR^K$, if:
\begin{enumerate}
\item each user $k$ encodes its message using the codebook $\mathcal{L}_k$,
\item the integer-valued matrix from the compute-and-forward transform of the MAC~\eqref{MACchannel} can be pseudo-triangularized with the permutation vector $\mathbf{\pi}$, and the optimal computation rates are $R_{\text{comp},1}\geq\cdots\geq R_{\text{comp},K}$,
\item all rates $R_1,\ldots,R_K$ satisfy
\begin{align}
R_k<R_{\text{comp},\pi^{-1}(k)}, \ \text{for } k=1,\ldots,K\label{rateAllocation}
\end{align}
where $\mathbf{\pi}^{-1}$ is the inverse permutation vector of $\mathbf{\pi}$,
\end{enumerate}
then all messages can be decoded with error probability smaller than $\epsilon$.
\label{thm:MAC}
\end{theorem}

\vspace{1mm}

Combining Theorems~\ref{thm:SumRate} and~\ref{thm:MAC} gives the following theorem.

\vspace{1mm}

\begin{theorem}
\label{thm:constantGapMAC}
The sum rate achieved by the compute-and-forward transform has a gap of no more than $\nicefrac{K}{2}\log K$ bits from the sum capacity of the MAC.
\end{theorem}
\begin{proof}
Let $R_{\text{comp},1}\geq\cdots\geq R_{\text{comp},K}$ be the optimal computation rates in the compute-and-forward transform of the MAC~\eqref{MACchannel}.
The integer-valued matrix from the compute-and-forward transform can be pseudo-triangularized with at least one permutation vector $\mathbf{\pi}$. By Theorem~\ref{thm:MAC}, the rate tuple
\begin{align}
R_k=R_{\text{comp},\pi^{-1}(k)}-\delta, \ \text{for } k=1,\ldots,K
\end{align}
is achievable for any $\delta>0$. For this rate tuple we have
\begin{align}
\sum_{k=1}^K R_k &=\sum_{k=1}^K \left(R_{\text{comp},\pi^{-1}(k)}-\delta\right)\nonumber\\
&=\sum_{k=1}^K R_{\text{comp},k}-K\delta\nonumber\\
&\geq \frac{1}{2}\log\left(1+\|\bh\|^2\Tsnr\right)-\frac{K}{2}\log(K)-K\delta,\label{macSumBound}
\end{align}
where~\eqref{macSumBound} follows from Theorem~\ref{thm:SumRate}. Since this is true for any $\delta>0$, the result follows.
\end{proof}

\subsection{Effective Multiple-Access Channel}\label{sub:effmac}

A channel that often arises in the context of lattice interference alignment is a $K$-user Gaussian multiple-access channel (MAC) with integer-valued ratios between some of the channel coefficients. Specifically, the output of such a channel can be written as
\begin{align}
\by=\sum_{\ell=1}^L g_\ell \left(\sum_{i \in \mathcal{K}_\ell}b_{i}\bx_{i}\right)+\bz,\label{integerRatiosMAC}
\end{align}
where $\mathcal{K}_1, \ldots, \mathcal{K}_L$ are disjoint subsets of $\{1,\ldots,K\}$. We assume that the $b_{i}\in\ZZ$ are non-zero integers, which opens up the possibility of lattice alignment.

The channel~\eqref{integerRatiosMAC} may describe the signal seen by a receiver in an interference network, perhaps after appropriate precoding at the transmitters. In such networks, each receiver is only interested in the messages from some of the users while the others act as interferers. Hence, it is beneficial to align several interfering users into one effective interferer, by taking advantage of the fact that the sum of lattice codewords is itself a lattice codeword.

\vspace{1mm}

\begin{definition}[Effective users]
For the MAC specified by~\eqref{integerRatiosMAC}, we define $L$ \textit{effective users}
\begin{align}
\bx_{\text{eff},\ell} \triangleq \sum_{i \in \mathcal{K}_\ell}b_{i}\bx_{i}, \ \ell=1,\ldots,L.\nonumber
\end{align}
\end{definition}

\vspace{1mm}

\begin{definition}[Effective MAC]
The $K$-user MAC~\eqref{integerRatiosMAC} induces the \textit{effective $L$-user MAC}
\begin{align}
\by=\sum_{\ell=1}^L g_\ell \bx_{\text{eff},\ell}+\bz,\label{effectiveMAC}
\end{align}
with the vector of effective channel coefficients $\bg=[g_1 \ \cdots \ g_L]^T\in\RR^L$. The effective channel is further characterized by the effective users' weights  $$b^2_{\text{eff},\ell}\triangleq\sum_{i \in \mathcal{K}_\ell}b^2_{i}$$ for $\ell=1,\ldots,L$, and the effective (diagonal) weight matrix \begin{align}
\bB\triangleq\text{diag}(b^2_{\text{eff},1},\ldots,b^2_{\text{eff},L}).\label{Bdef}
\end{align}
\end{definition}

\vspace{1mm}

\begin{definition}[Effective lattice points]
Let $\bt_{i}$ be the lattice point transmitted by user $i$. We define the \textit{effective lattice point} corresponding to effective user $\ell$ as
\begin{align}
\bt_{\text{eff},\ell}=\left[\sum_{i \in \mathcal{K}_\ell} b_{i}\bt_{i}\right]\Mod.\nonumber
\end{align}
Let $\map_{\text{eff}}(\ell)=\min_{i\in\mathcal{K}_{\ell}}\map(i)$ (where $\map(\cdot)$ is the mapping between users and fine lattices defined in Section~\ref{s:cf})
be the index of the densest lattice contributing to $\bt_{\text{eff},\ell}$. Since all lattices are nested, it follows that $\bt_{\text{eff},\ell}\in\Lambda_{\map_{\text{eff}}(\ell)}$.
\end{definition}

\vspace{2mm}

\begin{example}[Symmetric $K$-user interference channel]
Consider the symmetric $K$-user interference channel~\eqref{symICeq}. The channel seen by the $k$th receiver is of the form of~\eqref{integerRatiosMAC} with $g_1=1$, $g_2=g$, $\mathcal{K}_1=\{k\}$,  $\mathcal{K}_2=\{1,\ldots,K\}\setminus k$, and $b_i=1$ for $i=1,\ldots,K$.
If each of the $K$ users transmits a single codeword drawn from a common nested lattice code, the channel becomes an effective two-user MAC,
\begin{align}
\by_k=\bx_{\text{eff},k1}+g\bx_{\text{eff},k2}+\bz_k,\nonumber
\end{align}
where the effective users are $\bx_{\text{eff},k1}=\bx_k$ and $\bx_{\text{eff},k2}=\sum_{i\neq k}\bx_i$, and the effective users' weights are $b^2_{\text{eff},1}=1$ and $b^2_{\text{eff},2}=K-1$. The effective lattice points are $\bt_{\text{eff},k1}=\bt_k$  and $\bt_{\text{eff},k2}=[\sum_{i\neq k}\bt_i]\Mod$.
\label{ex:symIC}
\end{example}

\vspace{2mm}

Our achievable schemes for the symmetric $K$-user interference channel, developed in Section~\ref{s:lower}, are based on transforming the $K$-user MAC seen by each receiver into an effective MAC with less effective users. We will develop two schemes: One transforms the channel into an effective two-user MAC as in the example above. The other, which mimics the Han-Kobayashi approach, transforms the channel into an effective three-user MAC.

\vspace{2mm}

When lattice interference alignment schemes are designed properly, the message intended for the receiver is mapped into a separate effective user, while multiple interfering users are folded into a smaller number of effective users. In this case, it suffices for the receiver to decode only the $L$ effective lattice points corresponding to the effective users, rather than the $K$ lattice points transmitted by all users.
In our considerations, the effective lattice points are recovered by first decoding $L$ integer linear combinations of the form
\begin{align}
\bv&=\left[\sum_{\ell=1}^L a_{\ell}\sum_{i\in\mathcal{K}_\ell} b_i\bt_i\right]\Mod\nonumber\\
&=\left[\sum_{\ell=1}^L a_{\ell}\bt_{\text{eff},\ell}\right]\Mod
\end{align}
with linearly independent coefficient vectors, and then solving for $\bt_{\text{eff},1},\ldots,\bt_{\text{eff},L}$.

As in Section~\ref{s:cf}, in order to decode an integer linear combination $\bv$, the receiver first scales its observation by a factor $\beta$, removes the dithers, and reduces modulo $\Lambda$, which yields
\begin{align}
\bs&=\left[\beta\by+\sum_{\ell=1}^L a_\ell\sum_{i\in\mathcal{K}_\ell}b_{i}\bd_{i}\right]\Mod\nonumber\\
&=\bigg[\sum_{l=1}^L a_\ell \bx_{\text{eff},\ell}+\sum_{\ell=1}^L a_\ell\sum_{i\in\mathcal{K}_\ell}b_{i}\bd_{i}\nonumber\\
& \ \ \ \ \ \ \ \ +\sum_{\ell=1}^L (\beta g_\ell-a_\ell )\bx_{\text{eff},\ell}+\beta\bz\bigg]\Mod\nonumber\\
&=\left[\bv+\bz_{\text{eff}}(\bg,\ba,\beta,\left\{b_i\right\})\right]\Mod,\label{CoFestimateEff}
\end{align}
where
\begin{align}
\bz_{\text{eff}}(\bg,\ba,\beta,\left\{b_i\right\})&=\beta\by-\sum_{\ell=1}^L a_\ell\bx_{\text{eff},\ell}\nonumber\\
&=\sum_{\ell=1}^L (\beta g_\ell-a_\ell)\sum_{i \in \mathcal{K}_\ell} b_{i}\bx_{i}+\beta\bz\label{ZeffTag}
\end{align}
is effective noise which is statistically independent of $\bv$. Its effective variance is
\begin{align}
\sigma^2_{\text{eff}}(\bg,\ba,\beta,\bB)=\Tsnr\ \sum_{\ell=1}^L(\beta g_\ell-a_\ell)^2b^2_{\text{eff},\ell}+\beta^2,\label{ZeffVarTag}
\end{align}
where $\bB$ is defined in~\eqref{Bdef}.
Let $\ell^*=\min_{\ell:a_\ell\neq 0}\map_{\text{eff}}(\ell)$ be the index of the densest lattice participating in the linear combination $\bv$. Since all lattices are nested, then $\bv\in\Lambda_{\ell^*}$.
The receiver produces an estimate for $\bv$ by applying to $\bs$ the lattice quantizer associated with $\Lambda_{\ell^*}$,
\begin{align}
\hat{\bv}=\left[Q_{\Lambda_{\ell^*}}({\bs})\right]\Mod.
\end{align}
It follows from Theorem~\ref{thm:CoF} that there exists a chain of $K+1$ nested lattices which allows to decode $\bv$ with a vanishing error probability so long as
\begin{align}
R_{i}&<R_{\text{comp}}(\bg,\ba,\beta,\bB)
=\frac{1}{2}\log\left(\frac{\Tsnr}{\sigma^2_{\text{eff}}(\bg,\ba,\beta,\bB)}\right),
\label{CompRatesEff}
\end{align}
for all $i\in \bigcup_{\ell: a_\ell \neq 0} \mathcal{K}_\ell$.

The expression for $\sigma^2_{\text{eff}}(\bg,\ba,\beta,\bB)$ is equal to the MSE for linear estimation of $\tilde{X}_{\text{eff}}=\sum_{\ell=1}^L a_\ell X_{\text{eff},\ell}$ from $Y=\sum_{\ell=1}^L g_\ell X_{\text{eff},\ell}+Z$ where $\{X_{\text{eff},\ell}\}_{\ell=1}^L$ are statistically independent random variables with zero mean and variances $b^2_{\text{eff},\ell} \Tsnr$ respectively and $Z$ is statistically independent of $\{X_{\text{eff},\ell}\}_{\ell=1}^L$ with zero mean and unit variance.
Hence, the minimizing value of $\beta$ is the linear MMSE estimation coefficient of $\tilde{X}$ from $Y$.
A straightforward calculation shows that the minimizing value of $\beta$ is
\begin{align}
\beta=\frac{\mathbb{E}(\tilde{X}_{\text{eff}} Y)}{\Var(Y)}=\frac{\Tsnr \ \bg^T\bB\ba}{1+\Tsnr\  \bg^T\bB\bg}\nonumber
\end{align}
and the MSE it achieves is
\begin{align}
\sigma^2_{\text{eff}}(\bg,\ba,\bB)&\triangleq \min_{\beta\in\RR}\sigma^2_{\text{eff}}(\bg,\ba,\beta,\bB)\nonumber\\
&=\Tsnr\ \ba^T \left(\bB-\frac{\Tsnr\ \bB\bg\bg^T\bB}{1+\Tsnr\cdot\bg^T\bB\bg}\right) \ba\label{preWoodbury2}\\
&=\ba^T\left(\Tsnr^{-1}\bB^{-1}+\bg\bg^T\right)^{-1}\ba\label{woodbury2}\\
&=\left\|\left(\Tsnr^{-1}\bB^{-1}+\bg\bg^T\right)^{-1/2}\ba\right\|^2,\nonumber
\end{align}
where again~\eqref{woodbury2} can be verified using Woodbury's matrix identity~\cite[Thm 18.2.8]{harville}. Accordingly, we define
\begin{align}
R_{\text{comp}}(\bg,\ba,\bB)\triangleq \frac{1}{2}\log\left(\frac{\Tsnr}{\sigma^2_{\text{eff}}(\bg,\ba,\bB)}\right).
\label{effCompRate}
\end{align}

\vspace{1mm}

As in Section~\ref{s:cf}, we define the set of optimal $L$ coefficient vectors for the equivalent channel~\eqref{effectiveMAC} as the $L$ linearly independent vectors $\{\ba_1,\ldots,\ba_L\}$ that yield the highest computation rates $R_{\text{comp},1}=R_{\text{comp}}(\bg,\ba_1,\bB)\geq\cdots\geq R_{\text{comp},L}=R_{\text{comp}}(\bg,\ba_L,\bB)$ (see Definition~\ref{def:optEqs}). The compute-and-forward transform of the effective $L$-user MAC is
\begin{align}
{\bS}
&=\left(\begin{array}{c}
            \left[\beta_1\by+\sum_{\ell=1}^L a_{1\ell}\sum_{i\in\mathcal{K}_\ell}b_{i}\bd_{i}\right]\bmod\Lambda \\
            \vdots \\
            \left[\beta_L\by+\sum_{\ell=1}^L a_{L\ell}\sum_{i\in\mathcal{K}_\ell}b_{i}\bd_{i}\right]\bmod\Lambda
          \end{array}
\right)\nonumber\\
&=\left[\bA \left(
                       \begin{array}{c}
                         \bt_{\text{eff},1} \\
                         \vdots \\
                         \bt_{\text{eff},L} \\
                       \end{array}
                     \right)+{\bZ}_{\text{eff}} \right]\bmod\Lambda,\label{CoFTransformEff}
\end{align}
where $\bA=[\ba_1 \cdots \ba_L]^T$ and ${\bZ}_{\text{eff}}=[{\bz}^{T}_{\text{eff},1} \cdots {\bz}^{T}_{\text{eff},L}]^T$.

\vspace{1mm}

The next two theorems are simple extensions of Theorems~\ref{thm:SumRate} and~\ref{thm:MAC}. Their proofs are given in Appendix~\ref{app:effproofs}.

\vspace{2mm}

\begin{theorem}
The sum of optimal computation rates for the effective $L$-user MAC~\eqref{effectiveMAC} is lower bounded by
\begin{align}
\sum_{\ell=1}^L R_{\text{comp},\ell}\geq\frac{1}{2}\log\left(\frac{1+\Tsnr\sum_{\ell=1}^L g_\ell^2 b^2_{\text{eff},\ell}}{\det(\mathbf{B})}\right)-\frac{L}{2}\log(L).\nonumber
\end{align}
\label{thm:EffSumRate}
\end{theorem}

\vspace{1mm}

\begin{theorem}
Consider the effective $L$-user MAC~\eqref{effectiveMAC}, induced from the $K$-user MAC~\eqref{integerRatiosMAC}, characterized by the effective channel vector $\bg$ and the effective weight matrix $\bB$.
For any $\epsilon>0$ and $n$ large enough there exists a chain of $n$-dimensional nested lattices $\Lambda\subseteq\Lambda_L\subseteq\cdots\subseteq\Lambda_1$ forming the set of codebooks $\mathcal{L}_1,\ldots,\mathcal{L}_L$ with rates $R_1,\ldots,R_L$ such that for all $\bg\in\RR^L$ and $\bB$, if:
\begin{enumerate}
\item each user $i\in\mathcal{K}_\ell$ encodes its message using the codebook $\mathcal{L}_\ell$ or a codebook nested in $\mathcal{L}_\ell$,
\item the integer-valued matrix from the compute-and-forward transform of the effective MAC~\eqref{effectiveMAC} can be pseudo-triangularized with the permutation vector $\mathbf{\pi}$, and the optimal computation rates are $R_{\text{comp},1}\geq\cdots\geq R_{\text{comp},L}$,
\item all rates $R_1,\ldots,R_L$ satisfy
\begin{align}
R_\ell<R_{\text{comp},\pi^{-1}(\ell)}, \ \text{for } \ell=1,\ldots,L\label{rateAllocationEff}
\end{align}
where $\mathbf{\pi}^{-1}$ is the inverse permutation vector of $\mathbf{\pi}$,
\end{enumerate}
then all effective lattice points $\bt_{\text{eff},\ell}$ can be decoded with error probability smaller than $\epsilon$.
\label{thm:MACeff}
\end{theorem}

\vspace{1mm}

\begin{corollary}[Achievable symmetric rate]
\label{cor:minCompRateAchievable}
Consider the effective $L$-user MAC~\eqref{effectiveMAC}, induced from the $K$-user MAC~\eqref{integerRatiosMAC}, characterized by channel coefficients $\bg$ and the effective weight matrix $\bB$.
There exists a pair of $n$-dimensional nested lattices $\Lambda\subseteq\Lambda_1$ forming the codebook $\mathcal{L}$ of rate $R$ such that for all $\bg\in\RR^L$ and $\bB$, if
\begin{enumerate}
\item all users encode their messages using $\mathcal{L}$ (or codebooks nested in $\mathcal{L}$),
\item The $L$th optimal computation rate in the compute-and-forward transform of~\eqref{effectiveMAC} is $R_{\text{comp},L}$,
\item $R<R_{\text{comp},L}$,
\end{enumerate}
then, for $n$ large enough, all effective lattice points $\bt_{\text{eff},\ell}$ can be decoded with an arbitrarily small error probability.
\end{corollary}

\vspace{1mm}

\begin{remark}
Corollary~\ref{cor:minCompRateAchievable} is easily obtained from Theorem~\ref{thm:MACeff}. However, it can also be established without incorporating the compute-and-forward transform machinery. Indeed, if all users transmit from the same lattice codebook with rate smaller than $R_{\text{comp},L}$, by Theorem~\ref{thm:CoF}, each of the $L$ linear combinations with optimal coefficient vectors can be decoded (without using algebraic successive decoding as in the compute-and-forward transform approach). Then, the decoded linear combinations can be solved for the effective lattice points.

In Section~\ref{s:lower}, we introduce two achievable schemes for the $K$-user Gaussian interference channel. One of them is a simple transmission scheme where all users transmit from the same nested lattice code. The result of Corollary~\ref{cor:minCompRateAchievable} suffices to establish the rates achieved by this scheme. In the second achievable scheme, which mimics the Han-Kobabyshi scheme for the two-user interference channel, each user transmits a superposition of codewords taken from two nested lattice codebooks. In this case Corollary~\ref{cor:minCompRateAchievable} does not suffice and Theorem~\ref{thm:MACeff}, which uses the compute-and-forward transform machinery, is needed.
\end{remark}

\vspace{1mm}

\vspace{1mm}

In Section \ref{s:lower}, we leverage these achievability results to lower bound the capacity of the symmetric Gaussian $K$-user interference channel.

\section{Symmetric Capacity Upper Bounds} \label{s:upper}

In this section, we state an upper bound on the symmetric capacity of the symmetric $K$-user Gaussian interference channel. We follow the same arguments given in~\cite{jv10} for showing that the symmetric capacity of the symmetric $K$-user interference channel is upper bounded by that of the symmetric two-user interference channel. Namely, eliminating all but two users, say users $1$ and $2$, the symmetric capacity is upper bounded by the results of~\cite{etw08}. This is simply because removing interferers cannot decrease the symmetric rates for users $1$ and $2$. Thus, the upper bounds from~\cite{etw08} hold for the symmetric rates of user $1$ and $2$ in the $K$-user symmetric interference channel. Repeating the same argument for each pair of users we see that the upper bounds on $\csym$ developed in~\cite{etw08} for $K=2$ continue to hold for all $K>2$ as well. Therefore, the symmetric capacity of the symmetric $K$-user Gaussian interference channel is upper bounded as~\cite{etw08}
\begin{align}
\csym&\leq
\begin{cases}
                  \frac{1}{2}\log\left(1+\Tinr+\frac{\Tsnr}{1+\Tinr}\right) & 0\leq\alpha < \frac{2}{3} \\
                  \frac{1}{4}\log\left(1+\Tsnr\right)+\frac{1}{4}\log\left(1+\frac{\Tsnr}{1+\Tinr}\right) & \frac{2}{3}\leq\alpha < 1 \\
                  \frac{1}{4}\log\left(1+\Tsnr+\Tinr\right) & 1\leq\alpha < 2 \\
                  \frac{1}{2}\log\left(1+\Tsnr\right) & 2\leq \alpha.
\end{cases}
\label{upperBoundsEtkin}
\end{align}
Since we are only after an approximate capacity characterization, we further upper bound $\csym$ as 
\begin{align}
\csym&\leq\begin{cases}
                  \frac{1}{2}\log\left(1+\frac{\Tsnr}{1+\Tinr}\right)+1 & 0\leq\alpha < \frac{1}{2} \\
                  \frac{1}{2}\log^+\left(\Tinr\right)+1 & \frac{1}{2}\leq\alpha < \frac{2}{3} \\
                  \frac{1}{2}\log^+\left(\frac{\Tsnr}{\sqrt{\Tinr}}\right)+1 & \frac{2}{3}\leq\alpha < 1 \\
                  \frac{1}{4}\log^+\left(\Tinr\right)+1 & 1\leq\alpha < 2 \\
                  \frac{1}{2}\log\left(1+\Tsnr\right) & 2\leq \alpha.
                \end{cases}.
\label{CsymUpperBound}
\end{align}
for all values of $\Tsnr$.

\section{Achievable Schemes} \label{s:lower}

This section introduces two simple achievable schemes for reliable communication over the symmetric $K$-user interference channel that are based on nested lattice codes. These schemes are then shown to approximately achieve $\csym$, the symmetric capacity of the channel, for all channel gains $g$, except for an outage set of bounded measure. This outage set is explicitly characterized.

We begin by describing the two schemes and deriving their achievable symmetric rates. These rates are given in terms of the optimal computation rates corresponding to a certain effective multiple access channel, i.e., the rates are given as a solution to an optimization problem. This optimization problem, which amounts to finding the optimal coefficient vectors, can be efficiently solved numerically, as described in Section~\ref{s:numerical}. Figure~\ref{f:symICrates} shows our achievable rates for the three-user symmetric interference channel as a function of the interference level $g$, for several values of SNR. It is evident that the obtained rates significantly improve over time-sharing even for moderate values of SNR.

In order to establish the approximate optimality of these schemes, we derive explicit lower bounds on the rates they achieve which depend only on the $\Tsnr$ and $\Tinr$. As in the two-user case, the symmetric capacity exhibits a different behavior for different regimes of interference strength, characterized by the parameter $\alpha$.

We now present the two achievable schemes. The first achieves the approximate symmetric capacity in the noisy, strong, and very strong interference regimes, while the second achieves the approximate symmetric capacity in the weak and moderately weak interference regimes.

\noindent \underline{\emph{First scheme - A single-layer lattice code}:}
A pair of nested lattices $\Lambda\subseteq\Lambda_1$ is utilized to construct the codebook $\mathcal{L}=\Lambda_1\cap\Lambda$ of rate $\rsym$. All users encode their messages using this codebook.
Since all interferers arrive at the $k$th receiver with the same gain, they will be aligned into one effective lattice point.
Thus, the $K$-user MAC seen by the $k$th receiver becomes an effective two-user MAC of the form defined in Section~\ref{sub:effmac} (see Example~\ref{ex:symIC})
\begin{align}
\by_k=\bx_{\text{eff},k1}+g\bx_{\text{eff},k2}+\bz_k,\label{eff2userMAC}
\end{align}
where $\bx_{\text{eff},k1}=\bx_k$, $\bx_{\text{eff},k2}=\sum_{i\neq k}\bx_i$ are the effective users, $b^2_{\text{eff},1}=1$, $b^2_{\text{eff},2}=K-1$ are the effective users' weights and $\bg=[1 \ g]^T$ is the vector of channel gains.

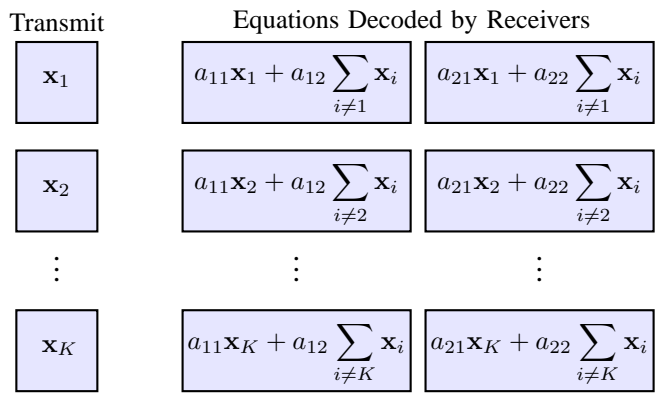
\begin{figure}[h]
\psset{unit=.85mm}
\begin{center}
\begin{pspicture}(0,0)(100,60)
\rput(6.5,58){Transmit}

\rput(62,58){Equations Decoded by Receivers}

\rput(0,42){
\psframe[fillstyle=solid,fillcolor=BoxBlue](0,0)(13,13) \rput(6.5,7){$\mathbf{x}_1$}
\psframe[fillstyle=solid,fillcolor=BoxBlue](26,0)(62,13)  \rput(44,6){${\displaystyle a_{11} \mathbf{x}_{1} + a_{12} \sum_{i \neq 1} \mathbf{x}_{i}}$}
\psframe[fillstyle=solid,fillcolor=BoxBlue](64,0)(100,13)  \rput(82,6){${\displaystyle a_{21} \mathbf{x}_{1} + a_{22} \sum_{i \neq 1} \mathbf{x}_{i}}$}
}

\rput(0,25){
\psframe[fillstyle=solid,fillcolor=BoxBlue](0,0)(13,13) \rput(6.5,7){$\mathbf{x}_2$}
\psframe[fillstyle=solid,fillcolor=BoxBlue](26,0)(62,13)  \rput(44,6){${\displaystyle a_{11} \mathbf{x}_{2} + a_{12} \sum_{i \neq 2} \mathbf{x}_{i}}$}
\psframe[fillstyle=solid,fillcolor=BoxBlue](64,0)(100,13)  \rput(82,6){${\displaystyle a_{21} \mathbf{x}_{2} + a_{22} \sum_{i \neq 2} \mathbf{x}_{i}}$}
}

\rput(6.5,20.5){\large{$\vdots$}}
\rput(44,20.5){\large{$\vdots$}}
\rput(82,20.5){\large{$\vdots$}}

\psframe[fillstyle=solid,fillcolor=BoxBlue](0,0)(13,13) \rput(7,7){$\mathbf{x}_K$}
\psframe[fillstyle=solid,fillcolor=BoxBlue](26,0)(62,13)  \rput(44,6){${\displaystyle a_{11} \mathbf{x}_{K} + a_{12} \sum_{i \neq K} \mathbf{x}_{i}}$}
\psframe[fillstyle=solid,fillcolor=BoxBlue](64,0)(100,13)  \rput(82,6){${\displaystyle a_{21} \mathbf{x}_{K} + a_{22} \sum_{i \neq K} \mathbf{x}_{i}}$}

\end{pspicture}
\end{center}
\caption{Illustration of the single-layer lattice scheme. Each transmitter sends a codeword drawn from a common lattice. Each receiver decodes two equations of the codewords, which it can then solve for its desired message.} \label{f:singlelayer}
\end{figure}

The next theorem gives an achievable rate region for the $K$-user interference channel when each receiver jointly decodes both the effective user $\bx_{\text{eff},k1}$ which carries the desired information, and the effective user $\bx_{\text{eff},k2}$ which carries the sum of interfering codewords. The theorem relies on decoding two independent linear combinations of the effective lattice points. See Figure \ref{f:singlelayer} for an illustration. This is in contrast to the successive decoding technique used in~\cite{sjvj08}, where first the interference is decoded and removed, and only then the desired lattice point is decoded.

\vspace{1mm}

\begin{theorem}
\label{thm:SymICnoLayeres}
Let $R_{\text{comp},1}\geq R_{\text{comp},2}$ be the optimal computation rates for the effective MAC~\eqref{eff2userMAC} induced by the symmetric $K$-user interference channel \eqref{symICeq}.
Any symmetric rate $\rsym<R_{\text{comp},2}$ is achievable for the symmetric $K$-user interference channel~\eqref{symICeq}.
\end{theorem}

\vspace{1mm}

\begin{proof}
Corollary~\ref{cor:minCompRateAchievable} implies that for any symmetric rate $\rsym<R_{\text{comp},2}$ there exists a pair of nested lattices $\Lambda\subseteq\Lambda_1$ such that both effective lattice points can be decoded at each receiver. Since the first effective user $\bx_{\text{eff},k1}$ carries all the desired information for the $k$th receiver, it follows that any $\rsym<R_{\text{comp},2}$ is achievable.
\end{proof}
\vspace{1mm}

The next theorem gives an achievable rate region for the $K$-user interference channel when each receiver decodes only its desired codeword, while treating all other interfering codewords as noise. This theorem can be trivially proved using i.i.d. Gaussian codebooks. Nevertheless, we prove the theorem using nested lattice codebooks for completeness.

\begin{theorem}
\label{thm:noisyInterference}
Any symmetric rate satisfying
\begin{align}
\rsym<\frac{1}{2}\log\left(1+\frac{\Tsnr}{1+(K-1)g^2\Tsnr}\right)\nonumber
\end{align}
is achievable for the symmetric $K$-user interference channel~\eqref{symICeq}.
\end{theorem}

\vspace{1mm}

\begin{proof}
Decoding $\bx_k$ at the $k$th receiver of the symmetric $K$-user interference channel~\eqref{symICeq}, while treating all other users as noise, is equivalent to decoding the linear combination with coefficient vector $\ba=[1 \ 0]^T$ in the effective two-user MAC~\eqref{eff2userMAC}. Therefore, any symmetric rate satisfying $\rsym<R_{\text{comp}}(\bg,[1 \ 0]^T,\bB)$ is achievable. The effective noise variance for decoding this linear combination is found using~\eqref{preWoodbury2} to be
\begin{align}
\sigma^2_{\text{eff}}(\bg,[1 \ 0]^T,\bB)=\Tsnr\left(1+\frac{\Tsnr}{1+(K-1)g^2\Tsnr}\right)^{-1},\nonumber
\end{align}
which, using~\eqref{effCompRate}, implies that
\begin{align}
R_{\text{comp}}(\bg,[1 \ 0]^T,\bB)=\frac{1}{2}\log\left(1+\frac{\Tsnr}{1+(K-1)g^2\Tsnr}\right).\nonumber
\end{align}
\end{proof}

\vspace{1mm}

For the two-user case, it is known that in the weak and moderately weak interference regimes each receiver should decode only part of the message transmitted by the other user \cite{etw08}. A natural extension of this Han-Kobayashi \cite{hk81} approach to the $K$-user case is for each receiver to decode linear combinations that only include parts of the interfering messages. This is enabled by using a superposition of two lattice codewords at each transmitter, as we describe next. See Figure \ref{f:twolayer} for an illustration.

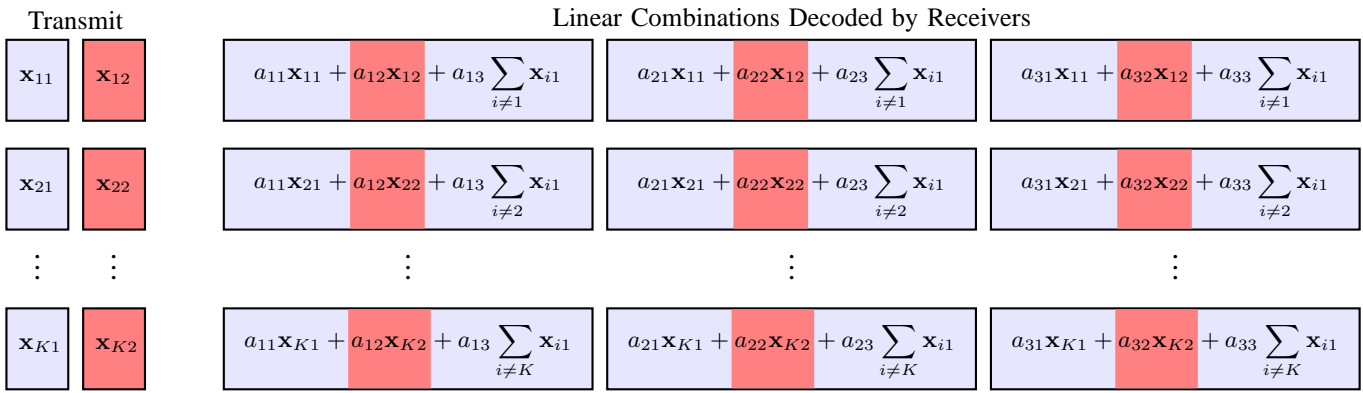
\begin{figure*}[!t]
\psset{unit=.85mm}
\begin{center}
\begin{pspicture}(0,0)(212,60)
\rput(11,58){Transmit}

\rput(123,58){Linear Combinations Decoded by Receivers}

\small

\rput(0,42){
\psframe[fillstyle=solid,fillcolor=BoxBlue](0,0)(10,13) \rput(5,7){$\mathbf{x}_{11}$}
\psframe[fillstyle=solid,fillcolor=BoxRed](12,0)(22,13) \rput(17,7){$\mathbf{x}_{12}$}
\psframe[fillstyle=solid,fillcolor=BoxBlue](34,0)(92,13) \psframe[fillstyle=solid,fillcolor=BoxRed,linestyle=none](53.75,0.15)(65.75,12.85)
\rput(63,6){${\displaystyle a_{11} \mathbf{x}_{11} + a_{12} \mathbf{x}_{12} +a_{13} \sum_{i \neq 1} \mathbf{x}_{i 1}}$}
\rput(60,0){\psframe[fillstyle=solid,fillcolor=BoxBlue](34,0)(92,13) \psframe[fillstyle=solid,fillcolor=BoxRed,linestyle=none](53.75,0.15)(65.75,12.85)
\rput(63,6){${\displaystyle a_{21} \mathbf{x}_{11} + a_{22} \mathbf{x}_{12} +a_{23} \sum_{i \neq 1} \mathbf{x}_{i 1}}$}
}
\rput(120,0){\psframe[fillstyle=solid,fillcolor=BoxBlue](34,0)(92,13) \psframe[fillstyle=solid,fillcolor=BoxRed,linestyle=none](53.75,0.15)(65.75,12.85)
\rput(63,6){${\displaystyle a_{31} \mathbf{x}_{11} + a_{32} \mathbf{x}_{12} +a_{33} \sum_{i \neq 1} \mathbf{x}_{i 1}}$}
}
}

\rput(0,25){
\psframe[fillstyle=solid,fillcolor=BoxBlue](0,0)(10,13) \rput(5,7){$\mathbf{x}_{21}$}
\psframe[fillstyle=solid,fillcolor=BoxRed](12,0)(22,13) \rput(17,7){$\mathbf{x}_{22}$}
\psframe[fillstyle=solid,fillcolor=BoxBlue](34,0)(92,13) \psframe[fillstyle=solid,fillcolor=BoxRed,linestyle=none](53.75,0.15)(65.75,12.85)
\rput(63,6){${\displaystyle a_{11} \mathbf{x}_{21} + a_{12} \mathbf{x}_{22} +a_{13} \sum_{i \neq 2} \mathbf{x}_{i 1}}$}
\rput(60,0){\psframe[fillstyle=solid,fillcolor=BoxBlue](34,0)(92,13) \psframe[fillstyle=solid,fillcolor=BoxRed,linestyle=none](53.75,0.15)(65.75,12.85)
\rput(63,6){${\displaystyle a_{21} \mathbf{x}_{21} + a_{22} \mathbf{x}_{22} +a_{23} \sum_{i \neq 2} \mathbf{x}_{i 1}}$}
}
\rput(120,0){\psframe[fillstyle=solid,fillcolor=BoxBlue](34,0)(92,13) \psframe[fillstyle=solid,fillcolor=BoxRed,linestyle=none](53.75,0.15)(65.75,12.85)
\rput(63,6){${\displaystyle a_{31} \mathbf{x}_{21} + a_{32} \mathbf{x}_{22} +a_{33} \sum_{i \neq 2} \mathbf{x}_{i 1}}$}
}
}

\rput(5,20.5){\large{$\vdots$}}
\rput(17,20.5){\large{$\vdots$}}
\rput(63,20.5){\large{$\vdots$}}
\rput(123,20.5){\large{$\vdots$}}
\rput(183,20.5){\large{$\vdots$}}

\psframe[fillstyle=solid,fillcolor=BoxBlue](0,0)(10,13) \rput(5.5,7){$\mathbf{x}_{K1}$}
\psframe[fillstyle=solid,fillcolor=BoxRed](12,0)(22,13) \rput(17.5,7){$\mathbf{x}_{K2}$}
\psframe[fillstyle=solid,fillcolor=BoxBlue](34,0)(92,13) \psframe[fillstyle=solid,fillcolor=BoxRed,linestyle=none](53.45,0.15)(66.7,12.85)
\rput(63,6){${\displaystyle a_{11} \mathbf{x}_{K1} + a_{12} \mathbf{x}_{K2} +a_{13} \sum_{i \neq K} \mathbf{x}_{i 1}}$}
\rput(60,0){\psframe[fillstyle=solid,fillcolor=BoxBlue](34,0)(92,13) \psframe[fillstyle=solid,fillcolor=BoxRed,linestyle=none](53.45,0.15)(66.7,12.85)
\rput(63,6){${\displaystyle a_{21} \mathbf{x}_{K1} + a_{22} \mathbf{x}_{K2} +a_{23} \sum_{i \neq K} \mathbf{x}_{i 1}}$}
}
\rput(120,0){\psframe[fillstyle=solid,fillcolor=BoxBlue](34,0)(92,13) \psframe[fillstyle=solid,fillcolor=BoxRed,linestyle=none](53.45,0.15)(66.7,12.85)
\rput(63,6){${\displaystyle a_{31} \mathbf{x}_{K1} + a_{32} \mathbf{x}_{K2} +a_{33} \sum_{i \neq K} \mathbf{x}_{i 1}}$}
}

\end{pspicture}
\end{center}
\caption{Illustration of the lattice Han-Kobayashi scheme. Each transmitter sends a public (blue) and a private (red) lattice codeword. Each receiver decodes three linear combinations of the public codewords as well as its desired private codeword while treating the other private codewords as noise. From these linear combinations, the receivers can infer their desired public and private messages.} \label{f:twolayer}
\end{figure*}

\noindent \underline{\emph{Second scheme - Lattice Han-Kobayashi}:}
This scheme employs a chain of nested lattices $\Lambda\subseteq\Lambda_2\subseteq\Lambda_1$ to construct two codebooks $\mathcal{L}_1$ and $\mathcal{L}_2$ with rates $R_1$ and $R_2$, respectively. Each user $k$ splits its message $w_k$ into two messages, a public message $w_{k1}$ that is mapped into a codeword $\bx_{k1}$ from $\mathcal{L}_1$ and a private message $w_{k2}$ that is mapped into a codeword $\bx_{k2}$ from $\mathcal{L}_2$. It is convenient to treat each user $k$ as two virtual users with codewords $\bx_{k1}$ and $\bx_{k2}$ that carry messages $w_{k1}$ and $w_{k2}$, respectively. User $k$ transmits a superposition of its virtual  users' codewords,
\begin{align}
\bx_k=\sqrt{1-\gamma^2}\bx_{k1}+\gamma\bx_{k2},\nonumber
\end{align}
for $\gamma\in[0,1)$. The signal seen by the $k$th receiver is
\begin{align}
\by_k&=\sqrt{1-\gamma^2}\bx_{k1}+\gamma\bx_{k2}\nonumber\\
&~~+g\sqrt{1-\gamma^2}\sum_{i\neq k}\bx_{i1}+g\gamma\sum_{i\neq k}\bx_{i2}+\bz_k,\label{layeredSymIC}
\end{align}
which induces the effective four-user MAC
\begin{align}
\by_k&=\sqrt{1-\gamma^2}\bx_{\text{eff},k1}+\gamma\bx_{\text{eff},k2}\nonumber\\
&~~+g\sqrt{1-\gamma^2}\bx_{\text{eff},k3}+g\gamma\bx_{\text{eff},k4}+\bz_k,\label{layeredEffSymIC}
\end{align}
with effective users $\bx_{\text{eff},k1}=\bx_{k1}$, $\bx_{\text{eff},k2}=\bx_{k2}$, $\bx_{\text{eff},k3}=\sum_{i\neq k}\bx_{i1}$ and $\bx_{\text{eff},k4}=\sum_{i\neq k}\bx_{i2}$. The effective users' weights are $b^2_{\text{eff},1}=1$, $b^2_{\text{eff},2}=1$, $b^2_{\text{eff},3}=K-1$ and $b^2_{\text{eff},4}=K-1$, and
\begin{align*}
\bg=\left[\sqrt{1-\gamma^2} \ \ \gamma \ \ g\sqrt{1-\gamma^2} \ \ g\gamma\right]^T \
\end{align*} is the vector of effective channel gains.

The receiver aims to decode the effective codewords $\bx_{\text{eff},k1}$, $\bx_{\text{eff},k2}$ and $\bx_{\text{eff},k3}$ while treating the fourth effective codeword $\bx_{\text{eff},k4}$ as noise. The next lemma will be useful for the derivation of rates achieved by this scheme. Its proof is given in Appendix~\ref{app:CsymProofs}.

\vspace{1mm}

\begin{lemma}
Consider the effective $L$-user MAC~\eqref{effectiveMAC}, where the decoder is only interested in the first $L-1$ effective lattice points $\bt_{\text{eff},1},\ldots,\bt_{\text{eff},L-1}$ and let $\kappa=1/\sqrt{1+\Tsnr g^2_L b^2_{\text{eff},L}}$. Any rate tuple achievable for decoding $\bt_{\text{eff},1},\ldots,\bt_{\text{eff},L-1}$ over the effective $(L-1)$-user MAC
\begin{align}
\sum_{\ell=1}^{L-1}\kappa g_\ell\bx_{\text{eff},\ell}+\bz\label{zeroEffMAC}
\end{align}
is also achievable for decoding the desired $L-1$ lattice points over~\eqref{effectiveMAC}.
\label{lem:equiEffChannel}
\end{lemma}

\vspace{1mm}

The next theorem gives the achievable rate region for the lattice Han-Kobayashi scheme.

\vspace{1mm}

\begin{theorem}
\label{thm:SymICHK}
Let $\kappa(\gamma)=1/\sqrt{1+\Tsnr g^2\gamma^2(K-1)}$ and consider the effective MAC
\begin{align}
\by_k&=\kappa(\gamma)\sqrt{1-\gamma^2}\bx_{\text{eff},k1}+\kappa(\gamma)\gamma\bx_{\text{eff},k2}\nonumber\\
&~~+\kappa(\gamma)g\sqrt{1-\gamma^2}\bx_{\text{eff},k3}+\bz_k,\label{effSymHK}
\end{align}
with effective channel vector
\begin{align*}
\bg=\left[\kappa(\gamma)\sqrt{1-\gamma^2} \ \  \kappa(\gamma)\gamma \ \ \kappa(\gamma)g\sqrt{1-\gamma^2}\right]^T \ ,
\end{align*} and effective users' weights $b^2_{\text{eff},1}=1$, $b^2_{\text{eff},2}=1$, and $b^2_{\text{eff},3}=K-1$.
Let $\{\ba_1(\gamma),\ba_2(\gamma),\ba_3(\gamma)\}$ and $R_{\text{comp},1}(\gamma)\geq R_{\text{comp},2}(\gamma)\geq R_{\text{comp},3}(\gamma)$ be the optimal coefficient vectors and computation rates, respectively.
Any symmetric rate satisfying
\begin{align*}
\rsym<\max_{\gamma\in[0,1)}R_{\text{comp},2}(\gamma)+R_{\text{comp},3}(\gamma)
\end{align*} is achievable for the symmetric $K$-user interference channel~\eqref{symICeq}.
\end{theorem}

\begin{proof}
The receiver is only interested in the effective lattice points $\bt_{\text{eff},k1}$, $\bt_{\text{eff},k2}$. Nevertheless, we require that it decodes the three effective lattice points $\bt_{\text{eff},k1}$, $\bt_{\text{eff},k2}$ and $\bt_{\text{eff},k3}$.
Due to Lemma~\ref{lem:equiEffChannel}, any rate tuple that is achievable over the effective channel~\eqref{effSymHK} is also achievable for decoding $\bt_{\text{eff},k1}$, $\bt_{\text{eff},k2}$ and $\bt_{\text{eff},k3}$ from the original effective channel~\eqref{layeredEffSymIC} induced by the lattice Han-Kobayashi scheme.

Note that $\bt_{\text{eff},k1}$ and $\bt_{\text{eff},k3}$ are points from the same codebook $\mathcal{L}_1$ with rate $R_1$, and $\bt_{\text{eff},k2}$ is a codeword from $\mathcal{L}_2$ with rate $R_2$.

Consider a compute-and-forward transform coefficient matrix $\bA(\gamma)=[\ba_1(\gamma) \ \ba_2(\gamma) \ \ba_3(\gamma)]^T$ for \eqref{effSymHK}. For any full-rank matrix there exists at least one order of pseudo-triangularization. Therefore, there exists a pseudo-triangularization of $\bA(\gamma)$ with at least one permutation vector $\mathbf{\pi}$.

Consider first the case where $\pi(3)=2$, i.e., the effective lattice point $\bt_{\text{eff},2}$ is the last to be removed in the algebraic successive cancellation decoding procedure of the compute-and-forward transform. According to Theorem~\ref{thm:MACeff}, for any $R_1<R_{\text{comp},2}(\gamma)$ and $R_2<R_{\text{comp},3}(\gamma)$ there exists a chain $\Lambda\subseteq\Lambda_2\subseteq\Lambda_1$ such that $\bt_{\text{eff},k1}$, $\bt_{\text{eff},k2}$ and $\bt_{\text{eff},k3}$ can be decoded from the effective channel~\eqref{effSymHK} via the compute-and-forward transform.

Otherwise, $\pi(1)=2$ or $\pi(2)=2$, which means that the effective lattice point $\bt_{\text{eff},2}$ is either removed first or second from the proceeding linear combinations in the algebraic successive cancellation decoding procedure of the compute-and-forward transform. According to Theorem~\ref{thm:MACeff} for any $R_1<R_{\text{comp},3}(\gamma)$ and $R_2<R_{\text{comp},2}(\gamma)$ there exists a chain $\Lambda\subseteq\Lambda_2\subseteq\Lambda_1$ such that $\bt_{\text{eff},k1}$, $\bt_{\text{eff},k2}$ and $\bt_{\text{eff},k3}$ can be decoded from the effective channel~\eqref{effSymHK} via the compute-and-forward transform.

Since $\rsym=R_1+R_2$, and $\gamma$ can be chosen such as to maximize $\rsym$, the theorem is proved.
\end{proof}

\vspace{1mm}

The problem of optimizing the power allocation $\gamma$ between the private and public codewords, played a major role in the approximation of the two-user interference channel capacity \cite{etw08}. Here, we follow the approach of~\cite{etw08} and choose $\gamma$ such that, at each unintended receiver, the received power of each private codeword is equal to that of the additive noise. Specifically, in the sequel we set $\gamma^2=1/(g^2\Tsnr)$. While this choice of $\gamma$ may be sub-optimal, it suffices to develop our capacity approximations in closed form. The achievable symmetric rate for $\gamma^2=1/(g^2\Tsnr)$ is given in the following corollary to Theorem~\ref{thm:SymICHK}.

\vspace{1mm}

\begin{corollary}
\label{cor:SymICHK}
Assume $g^2\Tsnr>1$ and consider the effective MAC
\begin{align}
\by_k&=\sqrt{\frac{g^2\Tsnr-1}{K\cdot g^2\Tsnr}}\bx_{\text{eff},k1}+\sqrt{\frac{1}{K\cdot g^2\Tsnr}}\bx_{\text{eff},k2}\nonumber\\
&+g\sqrt{\frac{g^2\Tsnr-1}{K\cdot g^2\Tsnr}}\bx_{\text{eff},k3}+\bz_k,\label{effSymHK1}
\end{align}
with effective channel vector
\begin{align}
\bg=\left[\sqrt{\frac{g^2\Tsnr-1}{K\cdot g^2\Tsnr}} \ \ \sqrt{\frac{1}{K\cdot g^2\Tsnr}} \ \ g\sqrt{\frac{g^2\Tsnr-1}{K\cdot g^2\Tsnr}} \ \right]^T,\label{effSymHK1gains}
\end{align}
and effective users' weights $b^2_{\text{eff},1}=1$, $b^2_{\text{eff},2}=1$, and $b^2_{\text{eff},3}=K-1$.
Let $\{\ba^{\text{HK}}_1,\ba^{\text{HK}}_2,\ba^{\text{HK}}_3\}$ and $R^{\text{HK}}_{\text{comp},1}\geq R^{\text{HK}}_{\text{comp},2}\geq R^{\text{HK}}_{\text{comp},3}$ be the optimal coefficient vectors and computation rates for this effective MAC.
Any symmetric rate
\begin{align*}
\rsym<R^{\text{HK}}_{\text{comp},2}+R^{\text{HK}}_{\text{comp},3}
\end{align*} is achievable for the symmetric $K$-user interference channel~\eqref{symICeq}.
\end{corollary}

\vspace{1mm}

Computing the achievable rates given by Theorem~\ref{thm:SymICnoLayeres} and Corollary~\ref{cor:SymICHK} requires finding the optimal computation rates for the effective MACs~\eqref{eff2userMAC} and~\eqref{effSymHK1}, which involves solving an integer least-squares optimization problem (see Section \ref{s:numerical}). In the remainder of this section, we derive lower bounds on these achievable rates that depend only on the values of $\Tsnr$ and $\Tinr$ and can therefore be directly compared to the upper bounds~\eqref{CsymUpperBound}. To simplify the exposition, we assume $g>0$ in the sequel, although all results easily follow for $g<0$ as well.

\subsection{Very Strong Interference Regime} \label{s:vsregime}

The very strong interference regime corresponds to \mbox{$g^2\geq\Tsnr$}. The sum capacity for
\begin{align}
g^2\geq \frac{(\Tsnr+1)^2}{\Tsnr},
\end{align}
which covers almost all of this regime was characterized exactly by Sridharan \textit{et al.}~\cite{sjvj08} using a lattice encoding scheme very similar to the one used in Theorem~\ref{thm:SymICnoLayeres}. The key difference is that in~\cite{sjvj08} each receiver decodes successively: it first decodes the sum of interfering codewords and then subtracts it in order to get a clean view of the desired signal. Recall that in our scheme, each receiver decodes two linear combinations of its signal and the interference.

A slight modification of the scheme given in~\cite{sjvj08} suffices to achieve the interference-free capacity to within a (small) constant gap for all $g^2>\Tsnr$.\footnote{Namely, if $\Tsnr\leq g^2<\Tsnr+2+1/\Tsnr$ all transmitters can reduce their transmission power by a small factor such that the very strong interference condition from~\cite{sjvj08} is satisfied. This power reduction results in a constant rate-loss.} Nevertheless, rather than using the results of~\cite{sjvj08}, we now proceed to lower bound the achievable rate of Theorem~\ref{thm:SymICnoLayeres} for the case $\alpha\geq 2$, i.e., $g^2\geq\Tsnr$. We do this in order to show that our lattice encoding and decoding framework suffices to achieve the approximate capacity in all regimes.

Using the single-layer scheme presented above, the channel seen by each receiver is converted to an effective two-user MAC~\eqref{eff2userMAC}. Let $R_{\text{comp},1}\geq R_{\text{comp},2}$ be the optimal computation rates for this effective channel. Theorem~\ref{thm:SymICnoLayeres} implies that any $\rsym<R_{\text{comp},2}$ is achievable, and hence, it suffices to lower bound $R_{\text{comp},2}$. We have
\begin{align}
R_{\text{comp},2}=R_{\text{comp},1}+R_{\text{comp},2}-R_{\text{comp},1}.\nonumber
\end{align}
Applying Theorem~\ref{thm:EffSumRate} to the effective MAC~\eqref{eff2userMAC}, we find that the sum of the optimal computation rates is lower bounded by
\begin{align}
R_{\text{comp},1}+R_{\text{comp},2}\geq \frac{1}{2}\log\left(\frac{1+\Tsnr(1+g^2(K-1))}{K-1}\right)-1.\nonumber
\end{align}
Therefore
\begin{align}
R_{\text{comp},2}\geq \frac{1}{2}\log\left(\frac{1+\Tsnr(1+g^2(K-1))}{K-1}\right)-1-R_{\text{comp},1},
\label{schem1sumRate}
\end{align}
and it suffices to upper bound $R_{\text{comp},1}$.

Let $R_{\text{comp}}(\bg,[0 \ 1]^T,\bB)$ be the computation rate for decoding the linear combination with coefficient vector $\ba=[0 \ 1]^T$ over the effective MAC~\eqref{eff2userMAC} with $\bg=[1 \ g]^T$ and $\bB=\diag(1,K-1)$. The effective noise variance for the coefficient vector $\ba=[0 \ 1]^T$, which is calculated using~\eqref{preWoodbury2}, is given in~\eqref{effVarInt} at the top of the next page.
\begin{figure*}[!t]
 \begin{align}
\sigma^2_{\text{eff}}(\bg,[0 \ 1]^T,\bB)&=\Tsnr\cdot\left[
                                                   \begin{array}{cc}
                                                     0 & 1 \\
                                                   \end{array}
                                                 \right]\left(\left[
                                                                \begin{array}{cc}
                                                                  1 & 0 \\
                                                                  0 & K-1 \\
                                                                \end{array}
                                                              \right]-\frac{\Tsnr\left[
                                                                                   \begin{array}{cc}
                                                                                     1 & 0 \\
                                                                                     0 & K-1 \\
                                                                                   \end{array}
                                                                                 \right]\left[
                                                                                          \begin{array}{c}
                                                                                            1 \\
                                                                                            g \\
                                                                                          \end{array}
                                                                                        \right]\left[
                                                                                                 \begin{array}{cc}
                                                                                                   1 & g \\
                                                                                                 \end{array}
                                                                                               \right]\left[
                                                                                                        \begin{array}{cc}
                                                                                                          1 & 0 \\
                                                                                                          0 & K-1 \\
                                                                                                        \end{array}
                                                                                                      \right]
                                                              }{1+\Tsnr \left[
                                                                          \begin{array}{cc}
                                                                            1 & g \\
                                                                          \end{array}
                                                                        \right]\left[
                                                                                 \begin{array}{cc}
                                                                                   1 & 0 \\
                                                                                   0 & K-1 \\
                                                                                 \end{array}
                                                                               \right]\left[
                                                                                        \begin{array}{c}
                                                                                          1 \\
                                                                                          g \\
                                                                                        \end{array}
                                                                                      \right]
                                                              }
                                                  \right)\left[
                                                           \begin{array}{c}
                                                             0 \\
                                                             1 \\
                                                           \end{array}
                                                         \right]\nonumber\\
&=\Tsnr\cdot\frac{(K-1)(1+\Tsnr)}{1+\Tsnr+(K-1)g^2\Tsnr}\label{effVarInt}
\end{align}
  \hrulefill
\vspace{-\baselineskip}
\end{figure*}
Substituting $\sigma^2_{\text{eff}}(\bg,[0 \ 1]^T,\bB)$ into~\eqref{CompRatesEff} gives
\begin{align}
R_{\text{comp}}(\bg,[0 \ 1]^T,\bB)&=\frac{1}{2}\log\left(\frac{1+\Tsnr\left(1+g^2(K-1)\right)}{(K-1)(1+\Tsnr)}\right).
\label{RcompInterference}
\end{align}
The coefficient vector $\ba=[0 \ 1]^T$ either gives the highest computation rate or not.
If it does, i.e., if $R_{\text{comp},1}=R_{\text{comp}}(\bg,[0 \ 1]^T,\bB)$, substituting~\eqref{RcompInterference} into~\eqref{schem1sumRate} gives
\begin{align}
R_{\text{comp},2}\geq \frac{1}{2}\log(1+\Tsnr)-1.\label{VSIbound1}
\end{align}

It follows from~\eqref{compRate2} and~\eqref{effvarlattice} that $[0~1]^T$ yields the highest computation rate among all integer coefficient vectors that are linearly dependent with it. Thus, if $R_{\text{comp},1} \neq R_{\text{comp}}(\bg,[0 \ 1]^T,\bB)$, any coefficient vector that attains $R_{\text{comp},1}$ must be linearly independent of $[0~1]^T$. It follows that
\begin{align}
R_{\text{comp},2}&\geq R_{\text{comp}}(\bg,[0 \ 1]^T,\bB)\nonumber\\
&>\frac{1}{2}\log\left(g^2\frac{\Tsnr}{1+\Tsnr}\right).\label{VSIbound2}
\end{align}
Taking the minimum of the two bounds~\eqref{VSIbound1} and~\eqref{VSIbound2}, and using the fact that $g^2\geq\Tsnr$ we obtain
\begin{align}
R_{\text{comp},2}&\geq\min\bigg
(\frac{1}{2}\log(1+\Tsnr)-1,
\frac{1}{2}\log\left(\frac{\Tsnr^2}{1+\Tsnr}\right)\bigg)^+\nonumber\\
&\geq\frac{1}{2}\log(1+\Tsnr)-1.\nonumber
\end{align}
Thus, in the very strong regime, any symmetric rate satisfying
\begin{align}
\rsym<\frac{1}{2}\log(1+\Tsnr)-1
\end{align}
is achievable, which is within $1$ bit of the outer bound~\eqref{CsymUpperBound}.

\subsection{Strong Interference Regime} \label{s:strongregime}

The strong interference regime corresponds to $1\leq\alpha<2$, or equivalently $1\leq g^2<\Tsnr$. As in the previous subsection, we lower bound $R_{\text{comp},2}$ in order to obtain a closed-form expression for the achievable symmetric rate. In contrast to the very strong interference regime, where the lower bound on $R_{\text{comp},2}$ is valid for any $g^2\geq\Tsnr$, here we must exclude certain channel gains in order to get a constant gap from the outer bound~\eqref{CsymUpperBound}. That is, the lower bounds we derive for the strong interference regime are only valid for a predefined subset of the interval $g^2\in[1,\Tsnr)$. As we increase the measure of this subset, our approximation gap worsens. This somewhat strange behavior is to be expected from the existing literature. The results of~\cite{mgmk09} and~\cite{eo09} show that for the $K$-user interference channel the DoF are discontinuous at the rationals. The notion of DoF corresponds to $\alpha\approx 1$. Since the strong interference regime contains values of $\alpha$ near $1$, we cannot expect to achieve rates which are a constant gap from the upper bounds of~\cite{etw08} for all values of $g$. Instead, we show that these upper bounds can be approached up to a constant gap for all $1\leq g^2<\Tsnr$ except for some outage set whose measure can be controlled at the price of increasing the gap. We will see a similar phenomenon when we analyze the moderately weak interference regime.

From \eqref{schem1sumRate}, we have
\begin{align}
R_{\text{comp},2}&\geq\frac{1}{2}\log\left(\frac{1+\Tsnr(1+g^2(K-1))}{K-1}\right)-1-R_{\text{comp},1}\nonumber\\
&>\frac{1}{2}\log\left(g^2\Tsnr\right)-1-R_{\text{comp},1}\label{SIbound1}
\end{align}
The optimal computation rate for the effective MAC~\eqref{eff2userMAC} can be written, by substituting $\bg=[1 \ g]^T$ and $\bB=\diag(1,K-1)$ into~\eqref{ZeffVarTag}, as
\begin{align}
&R_{\text{comp},1}=\frac{1}{2}\log(\Tsnr)-\frac{1}{2}\log(\sigma^2_g)\label{RcompVeryStrong}\\
&\sigma_g^2=\min_{\beta,a_1,a_2}\bigg(\Big((\beta-a_1)^2+(\beta g-a_2)^2(K-1)\Big)\Tsnr+\beta^2 \bigg), \nonumber
\end{align} where $\sigma^2_g$ is the effective noise variance and the minimization is over $\beta\in\mathbb{R}$, and $\ba=[a_1 \ a_2]^T\in\mathbb{Z}^2\setminus\mathbf{0}$.
Substituting~\eqref{RcompVeryStrong} into~\eqref{SIbound1} and applying Theorem~\ref{thm:SymICnoLayeres}, we see that any symmetric rate satisfying
\begin{align}
\rsym<\frac{1}{2}\log\left(g^2\right)+\frac{1}{2}\log\left(\sigma_g^2\right)-1,\label{strongIntGenRateExp}
\end{align}
is achievable over the $K$-user interference channel.
Thus, in order to obtain a lower bound on $\csym$ it suffices to lower bound $\sigma_g^2$.

\begin{remark}
It may at first seem counterintuitive that the symmetric rate expression in~\eqref{strongIntGenRateExp} is an increasing function of the effective noise variance $\sigma_g^2$ for the highest computation rate $R_{\text{comp},1}$. However, as discussed in Section~\ref{s:overview}, when $\sigma_g^2$ is small, the desired signal and the interference are aligned. From another perspective, if the channel vector $\bg$ is very close  to the integer vector $\ba$ (after scaling by $\beta$), then it must be far from the integer coefficient vector that determines $R_{\text{comp},2}$, which in turn determines $R_{\text{SYM}}$. Thus, the best performance is attained when the channel vector is hard to approximate with an integer vector. Building on this idea, the lower bound derived below connects our problem to a Diophantine approximation problem\footnote{Diophantine approximation refers to the branch of number theory that studies how well real numbers can be approximated by rational numbers.} and characterizes the outage set in terms of channel gains that are well-approximated by rationals.
\end{remark}

The effective noise $\sigma_g^2$ can be bounded as
\begin{align}
\sigma_g^2\geq\min_{\beta,a_1,a_2}\bigg(\Big((\beta-a_1)^2+(\beta g-a_2)^2\Big)\Tsnr+\beta^2 \bigg).\label{sigmag}
\end{align}
We first hold $\beta$ constant and minimize over $a_1,a_2$. If $|\beta|\geq 1/(2g)$, the optimal choices for the integers $a_1,a_2$ are
\begin{align}
a_1=\lfloor \beta \rceil, \ a_2=\lfloor \beta g \rceil.\label{integersVeryStrong1}
\end{align}
If $|\beta | < 1/(2g)$, rounding the gains will set both $a_1$ and $a_2$ to zero, which is not allowed. Since $g \geq 1$, the optimal choice is
\begin{align}
a_1=0, \ a_2=\sign(\beta).\label{integersVeryStrong2}
\end{align}
Now, we are left with the problem of minimizing~\eqref{sigmag} over $\beta$. Rather than explicitly solving this minimization problem, we give a lower bound on its solution. We do this by splitting the real line into three intervals, and lower bounding $\sigma^2_g$ for all values of $\beta$ within each one. Then, we take the minimum over these three bounds.

\vspace{1mm}

\underline{Interval $1$ : $0<|{\beta}|\leq 1/(2g)$}

In this interval it is optimal to set $a_2=\sign(\beta)$. Moreover, $|\beta g|\leq 1/2$, and therefore $|\beta g-a_2|>1/2$. Combining this with~\eqref{sigmag} gives
\begin{align}
\sigma_g^2\geq\frac{\Tsnr}{4}.\label{StrongBound1}
\end{align}

\vspace{1mm}

\underline{Interval $2$ : $1/(2g)<|{\beta}|\leq 1/2$}

Here, it is optimal to set $a_1=\lfloor\beta\rceil=0$. Substituting $a_1=0$ in~\eqref{sigmag} gives
\begin{align}
\sigma_g^2\geq \beta^2\Tsnr>\frac{\Tsnr}{4g^2}>\frac{\Tsnr^{1/2}}{4\sqrt{g^2}},\label{StrongBound2}
\end{align}
where the last inequality follows since $g^2<\Tsnr$ in the strong interference regime.

\underline{Interval $3$ : $1/2<|{\beta}|$}

Since $|\beta|>1/2$, we can write $\beta = q + \varphi$ where $q$ is a nonzero integer and $\varphi\in[-1/2,1/2)$. Substituting into~\eqref{sigmag}, we get
\begin{align}
\sigma_g^2&\geq\min_{\varphi,q,a_1,a_2}\bigg((\varphi+q-a_1)^2\Tsnr\nonumber\\
& \ \ \ \ \ \ \ \ \ +(q g-a_2+\varphi g)^2\Tsnr+(\varphi+q)^2 \bigg)\nonumber\\
&\geq\min_{\varphi,q,a_2}\bigg(\Big(\varphi^2+(q g-a_2+\varphi g)^2\Big)\Tsnr+\frac{q^2}{4} \bigg).\label{sigmaBoundphi}
\end{align}
The minimization of~\eqref{sigmaBoundphi} with respect to $\varphi$ (where the constraint $\varphi\in[-1/2,1/2)$ is ignored) can be obtained by differentiation. The minimizing value of $\varphi$ is
\begin{align}
\varphi^*=-\frac{g}{1+g^2}(qg-a_2).\nonumber
\end{align}
Substituting $\varphi^*$ into~\eqref{sigmaBoundphi} gives
\begin{align}
\sigma_g^2\geq\min_{q,a_2}\bigg(\frac{1}{1+g^2}(q g-a_2)^2\Tsnr+\frac{q^2}{4} \bigg),
\end{align}
which, using the fact that $g^2\geq 1$, can be further bounded by
\begin{align}
\sigma_g^2\geq\frac{1}{4}\min_{q,a_2}\max\bigg(\frac{1}{g^2}(q g-a_2)^2\Tsnr,q^2 \bigg).\label{sigmaBoundphiA}
\end{align}
We would like to obtain a lower bound on $\sigma_g^2$ that is valid for all $g\notin \mathcal{S}$, where $\mathcal{S}$ is an outage set with bounded measure. Consider first the interval $[b,b+1)$ for some integer $1\leq b<\sqrt{\Tsnr}$.
Define
\begin{align}
q_{\text{max},b}\triangleq\frac{1}{\sqrt{b+1/2}}\Tsnr^{1/4-\delta/2},\label{qDef}
\end{align}
for some $\delta>0$ to be specified later, and note that $q_{\text{max},b}$ is not necessarily an integer. Also, define
\begin{align}
\Phi_b\triangleq \sqrt{b+1/2} \ \Tsnr^{-1/4-\delta/2}\label{PhiDef}
\end{align} and let $\mathcal{S}_b$ be the set of all values of $g\in[b,b+1)$ such that the inequality
\begin{align}
|q g-a_2|<\Phi_b \label{diophCondition}
\end{align}
has at least one solution with integers $q$ and $a_2$, where $q$ is in the range $0<q\leq q_{\text{max},b}$. Let $\bar{\mathcal{S}}_b=[b,b+1)\setminus \mathcal{S}_b$.
By~\eqref{sigmaBoundphiA},~\eqref{qDef}, and~\eqref{PhiDef}, we have that for all $g\in\bar{\mathcal{S}}_b$
\begin{align}
\sigma_g^2&\geq \frac{1}{4}\min\bigg(\min_{0<q\leq\lfloor q_{\text{max},b}\rfloor,a_2}\max\left(\frac{1}{g^2}(qg-a_2)^2\Tsnr,q^2\right),\nonumber\\
&~~~~~~~~~~~~~~~~\min_{\lceil q_{\text{max},b}\rceil \leq  q,a_2}\max\left(\frac{1}{g^2}(qg-a_2)^2\Tsnr,q^2\right) \bigg)\nonumber\\
&\geq\frac{1}{4}\min\bigg(\frac{1}{g^2}\Phi_b^2\Tsnr,q_{\text{max},b}^2 \bigg)\nonumber\\
&= \frac{1}{4}\min\bigg(\frac{b+1/2}{g^2}\Tsnr^{1/2-\delta},\frac{1}{b+1/2}\Tsnr^{1/2-\delta} \bigg).\label{sigmaStrongIntTmpA}
\end{align}
Since $b\geq 1$, we have that
\begin{align}
\frac{g}{2}< b+\frac{1}{2}<2g.\nonumber
\end{align}
Thus,~\eqref{sigmaStrongIntTmpA} can be further bounded by
\begin{align}
\sigma_g^2\geq \frac{1}{8|g|}\Tsnr^{1/2-\delta}.\label{sigmaStrongInt}
\end{align}
We now turn to upper bound the Lebesgue measure of the set $\mathcal{S}_b$. Our derivation is quite similar to the proof of the convergent part of Khinchine's Theorem~\cite{Diophantine}. Let $\mathcal{I}=[-1,1)$ and define the set
\begin{align}
\mathcal{T}_b(q)=\left[\left\{b,b+\frac{1}{q},\cdots,b+\frac{q-1}{q}\right\}+\frac{\Phi_b}{q}\mathcal{I}\right]\bmod[b,b+1),
\nonumber
\end{align}
where the sum of the two sets is a Minkowski sum. Writing the Diophantine approximation problem~\eqref{diophCondition} as
\begin{align}
\left|g-\frac{a_2}{q}\right|<\frac{\Phi_b}{q},
\end{align}
we see that for a given $q$ and $g\in [b,b+1)$ the inequality admits a solution if and only if $g\in\mathcal{T}_b(q)$. It follows that
\begin{align}
\mathcal{S}_b=\bigcup_{q=1}^{\lfloor q_{\text{max},b} \rfloor}T_b(q).\label{SbDef}
\end{align}
See Figure~\ref{fig:outageSet} for an illustration of the sets $\mathcal{T}_b(q)$ and $\mathcal{S}_b$.
\begin{figure}[]
\begin{center}
\hspace{-0.2in}
\psfrag{b}{ $b$}
\psfrag{b1}{ $b+1$}
\psfrag{b2}{ $b+\frac{1}{2}$}
\psfrag{b3}{ $b+\frac{1}{3}$}
\psfrag{b4}{ $b+\frac{2}{3}$}
\psfrag{a1}{ $\mathcal{T}_b(1)$}
\psfrag{a2}{ $\mathcal{T}_b(2)$}
\psfrag{a3}{ $\mathcal{T}_b(3)$}
\psfrag{a4}{ $\mathcal{S}_b$}
\includegraphics[width=1 \columnwidth]{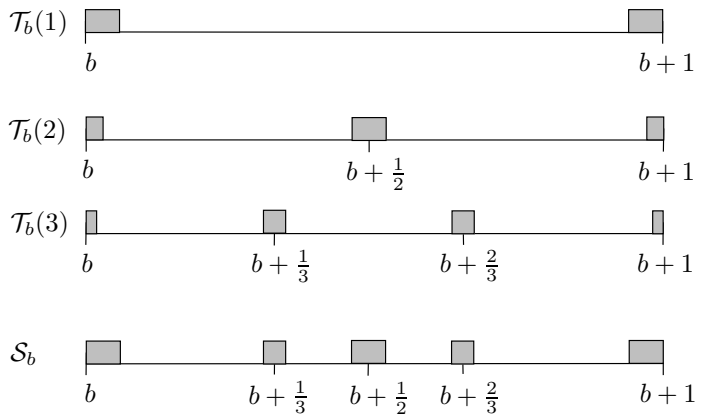}
\caption{An illustration of the sets $\mathcal{T}_b(1)$, $\mathcal{T}_b(2)$, $\mathcal{T}_b(3)$ and their union $\mathcal{S}_b$. In this illustration, $q_{\text{max},b}=3$ and $\Phi_b=\nicefrac{1}{16}$.}
\label{fig:outageSet}
\end{center}
\end{figure}
Thus, the Lebesgue measure of $\mathcal{S}_b$ can be upper bounded by
\begin{align}
\mu(\mathcal{S}_b)&=\Vol\left(\mathcal{S}_b\right)\nonumber\\
&\leq\sum_{q=1}^{\lfloor q_{\text{max},b} \rfloor}\Vol\left(\mathcal{T}_b(q)\right)\nonumber\\
&\leq\sum_{q=1}^{\lfloor q_{\text{max},b} \rfloor}q\cdot2\frac{\Phi_b}{q}\nonumber\\
&\leq 2 q_{\text{max},b}\Phi_b\nonumber\\
&=2\Tsnr^{-\delta}.
\label{outageMeasure}
\end{align}
Setting $\delta=(c+1)/\log(\Tsnr)$ and substituting into~\eqref{sigmaStrongInt} and~\eqref{outageMeasure} gives
\begin{align}
\sigma_g^2\geq \frac{2^{-c}}{16\sqrt{g^2}}\Tsnr^{1/2}\label{sigmaStrongIntTmp}
\end{align}
for all $g\in[b,b+1)$ up to an outage set $\mathcal{S}_b$ of measure not greater than $2^{-c}$.

\vspace{2mm}

Combining the three bounds~\eqref{StrongBound1},~\eqref{StrongBound2}, and~\eqref{sigmaStrongIntTmp} yields
\begin{align}
\sigma_g^2&\geq\min\left(\frac{\Tsnr}{4},\frac{1}{4\sqrt{g^2}}\Tsnr^{1/2},\frac{2^{-c}}{16\sqrt{g^2}}\Tsnr^{1/2} \right)\nonumber\\
&\geq\frac{2^{-c}}{16\sqrt{g^2}}\Tsnr^{1/2}\label{sigmaStrongInt2}
\end{align}
for all $g\in[b,b+1)$ up to an outage set $\mathcal{S}_b$ with measure at most $2^{-c}$.

Combining~\eqref{strongIntGenRateExp} and~\eqref{sigmaStrongInt2} we see that, for all $g\in[b,b+1)$ up to an outage set $\mathcal{S}_b$ of measure not greater than $2^{-c}$, any symmetric rate satisfying
\begin{align}
\rsym&<\frac{1}{4}\log(g^2\Tsnr)-\frac{c}{2}-3\nonumber\\
&=\frac{1}{4}\log(\Tinr)-\frac{c}{2}-3 \label{rateStrongInt}
\end{align}
is achievable. We conclude that the symmetric rate~\eqref{rateStrongInt} is achievable for all channel gains in the strong interference regime except for an outage set whose measure is a fraction of $2^{-c}$ of the interval $1\leq|g|<\sqrt{\Tsnr}$, for any $c > 0$.

\vspace{2mm}

\begin{remark}
In the high $\Tsnr$ limit, the total DoF of the symmetric $K$-user IC drops below $K/2$ when the channel gain $g$ is rational \cite{eo09,mgmk09}. At finite $\Tsnr$, we observe that the channel gains in the interval $g\in[b,b+1)$ that fall within the outage set are the ones close to rational numbers with denominator smaller than $q_{\text{max},b}$. Since $q_{\text{max},b}$ scales as $\Tsnr^{1/4}/\sqrt{|g|}$, only channel gains close to rational numbers with denominator smaller than $\Tsnr^{1/4}$ may result in outage. Moreover, the sensitivity of the achievable rate to the ``rationality'' of $g$ decreases as $g$ increases.
\end{remark}

\begin{remark}
We note that for any $c>0$ the set of channel coefficients that fall in the outage set
\begin{align}
\mathcal{S}=\bigcup_{b=1}^{\lfloor\sqrt{\Tsnr} \rfloor}\mathcal{S}_b,\nonumber
\end{align}
can be easily determined by setting $\delta=(c+1)/\log(\Tsnr)$ in~\eqref{qDef},~\eqref{PhiDef}, and applying~\eqref{SbDef}.
\end{remark}

\subsection{Moderately Weak Interference Regime} \label{s:moderateregime}

The moderately weak interference regime is characterized by $2/3\leq \alpha<1$, or equivalently, $\Tsnr^{-1/3}\leq g^2<1$. As in the strong interference regime, we show the achievability of symmetric rates which are a constant gap from the upper bound for a certain fraction of the channel gains. As opposed to the very strong and strong interference regimes, where a single-layered lattice scheme suffices to achieve the approximate capacity, here we will need the second scheme, which employs two layers of lattice codes at each transmitter.

We will set the power of the private lattice codewords so that they are perceived at noise level at the unintended receivers. Let $g_1$, $g_2$ and $g_3$ be the channel gains in the effective three-user MAC~\eqref{effSymHK1} from Corollary~\ref{cor:SymICHK}, and recall that, for this effective channel, the effective weight matrix is $\bB=\diag(1,1,K-1)$. Let $R^{\text{HK}}_{\text{comp},1}\geq R^{\text{HK}}_{\text{comp},2}\geq R^{\text{HK}}_{\text{comp},3}$ be the three optimal computation rates for the effective channel~\eqref{effSymHK1}. Corollary~\ref{cor:SymICHK} states that any symmetric rate satisfying $\rsym<R^{\text{HK}}_{\text{comp},2}+R^{\text{HK}}_{\text{comp},3}$ is achievable, and we now turn to lower bounding this achievable rate in closed form. First, note that
\begin{align}
R^{\text{HK}}_{\text{comp},2}+R^{\text{HK}}_{\text{comp},3}=\sum_{m=1}^3 R^{\text{HK}}_{\text{comp},m}-R^{\text{HK}}_{\text{comp},1}.\label{MWIrate}
\end{align}
By applying Theorem~\ref{thm:EffSumRate} to the effective channel~\eqref{effSymHK1}, we obtain the following lower bound on the sum of optimal computation rates,
\begin{align}
&\sum_{m=1}^3 R^{\text{HK}}_{\text{comp},m}\nonumber\\
&\geq \frac{1}{2}\log\left(\frac{1+\Tsnr(g_1^2+g_2^2+(K-1)g_3^2)}{K-1}\right)-\frac{3}{2}\log(3)\nonumber\\
&> \frac{1}{2}\log\left(\Tsnr(g_1^2+g_2^2)\right)-\frac{1}{2}\log\left(3^3(K-1)\right)\nonumber\\
&=\frac{1}{2}\log\left(\Tsnr\right)-\frac{1}{2}\log\left(27K(K-1)\right),\label{sumEffHK}
\end{align}
where we have used the fact that $g_1^2+g_2^2=1/K$ in the last equality.
The highest computation rate can be written as
\begin{align}
&R^{\text{HK}}_{\text{comp},1}=\frac{1}{2}\log(\Tsnr)-\frac{1}{2}\log(\sigma^2_{\text{HK}})\label{RHK1}
\end{align}
where $\sigma^2_{\text{HK}}$ is given in~\eqref{sigma2HK} at the top of the next page.
\begin{figure*}[!t]
\begin{align}
\sigma^2_{\text{HK}}=\min_{\beta,\mathbf{a}_1^{\text{HK}}}\left(\beta^2+\left(\Bigg(\beta \sqrt{\frac{g^2\Tsnr-1}{K\cdot g^2\Tsnr}}-a_1\Bigg)^2+\Bigg(\beta \sqrt{\frac{1}{K\cdot g^2\Tsnr}}-a_2\Bigg)^2+(K-1)\Bigg(\beta g\sqrt{\frac{g^2\Tsnr-1}{K\cdot g^2\Tsnr}}-a_3\Bigg)^2\right)\Tsnr \right)\label{sigma2HK}
\end{align}
  \hrulefill
\vspace{-\baselineskip}
\end{figure*}
The minimization in~\eqref{sigma2HK} is performed over all $\beta\in\RR$ and $\ba^{\text{HK}}_1=[a_1 \ a_2 \ a_3]^T\in\mathbb{Z}^3\setminus\mathbf{0}$. Combining~\eqref{MWIrate},~\eqref{RHK1}, and~\eqref{sigma2HK} and applying Corollary~\ref{cor:SymICHK}, we see that any symmetric rate satisfying
\begin{align}
\rsym<\frac{1}{2}\log(\sigma^2_{\text{HK}})-\frac{1}{2}\log\left(27K(K-1)\right)\label{symRateHK}
\end{align}
is achievable for the $K$-user interference channel. Therefore, it suffices to lower bound the effective noise variance $\sigma^2_{\text{HK}}$.
Substituting
$$\beta=\sqrt{\frac{K g^2\Tsnr}{g^2\Tsnr-1}}\tilde{\beta}$$ in~\eqref{sigma2HK}, which is allowed since $\beta$ can take any value in $\RR$, gives
\begin{align}
&\sigma^2_{\text{HK}}=\min_{\tilde{\beta},a_1,a_2,a_3}\Bigg(\tilde{\beta}^2\cdot\frac{Kg^2\Tsnr}{g^2\Tsnr-1}+(\tilde{\beta} -a_1)^2\Tsnr\nonumber\\
&+\left(\frac{\tilde{\beta}}{\sqrt{g^2\Tsnr-1}}-a_2\right)^2\Tsnr+(K-1)(\tilde{\beta} g-a_3)^2\Tsnr \Bigg).\label{sigmaHK}
\end{align}
In the sequel, we assume\footnote{This assumption is valid, since for $\Tsnr\leq 4$ the symmetric capacity is upper bounded by $1/2\log(1+4)=1.161$bits. Our capacity approximations in this subsection, and also in the next subsection, exhibit a constant gap greater than $7/2$ bits, and therefore hold for $\Tsnr<4$.} $\Tsnr>4$. With this assumption, $\sqrt{g^2\Tsnr-1}>1$ for all $g^2\geq \Tsnr^{-1/3}$, i.e., for all values of $g$ in the moderately weak interference regime. We will also use the fact that the inequality $\sqrt{g^2\Tsnr-1}>1$ continues to hold for all $g^2\geq \Tsnr^{-1/2}$, i.e., for all values of $g$ in the weak interference regime.
This implies that \mbox{$g^2\Tsnr/(g^2\Tsnr-1)>1$} and hence~\eqref{sigmaHK} can be lower bounded as
\begin{align}
&\sigma^2_{\text{HK}} \geq \min_{\tilde{\beta},a_1,a_2,a_3}\Bigg(K\tilde{\beta}^2+(\tilde{\beta} -a_1)^2\Tsnr\nonumber\\
&+\left(\frac{\tilde{\beta}}{\sqrt{g^2\Tsnr-1}}-a_2\right)^2\Tsnr+(K-1)(\tilde{\beta} g-a_3)^2\Tsnr \Bigg).\label{sigmaHK2}
\end{align}
We first hold $\tilde{\beta}$ constant, and minimize over $a_1,a_2,a_3$. If $|\tilde{\beta}|\geq1/2$, the optimal choices for the integers $a_1,a_2,a_3$ are
\begin{align}
a_1=\lfloor\tilde{\beta}\rceil, \ a_2=\left\lfloor\tilde{\beta}/\sqrt{g^2\Tsnr-1}\right\rceil, \ a_3=\lfloor\tilde{\beta} g\rceil. \label{integerChoices}
\end{align}
If $|\tilde{\beta}|<1/2$, all three integers $a_1,a_2,a_3$ from~\eqref{integerChoices} are zero, which is not permitted. Therefore, for these values of $\tilde{\beta}$ one of the integers must take the value $1$ or $-1$. Since for $\Tsnr>4$ and $\Tsnr^{-1/2}\leq g^2< 1$ we have
\begin{align}
\max\bigg(|\tilde{\beta}|,\Big|\tilde{\beta}/\sqrt{g^2\Tsnr-1}\Big|,|\tilde{\beta} g| \bigg)=|\tilde{\beta}|,\nonumber
\end{align}
the optimal choices of $a_1,a_2,a_3$ for values of $|\tilde{\beta}|<1/2$ are
\begin{align}
a_1=\sign(\tilde{\beta}), \ a_2=0, \ a_3=0 \ . \label{integerChoices2}
\end{align}
Now, the problem of lower bounding $\sigma^2_{\text{HK}}$ reduces to minimizing~\eqref{sigmaHK2} over $\tilde{\beta}$. Rather than solving this cumbersome minimization problem, we split the real line into four intervals, and lower bound $\sigma^2_{\text{HK}}$ for all values of $\tilde{\beta}$ within each one. Then, we take the minimum over these four lower bounds. In a similar manner to the previous subsection, we define $\delta=(2c+8)/\log(\Tsnr)$, where $c>0$ is some constant. The lower bounds below are derived in Appendix~\ref{app:boundsModeratelyWeak}.

\vspace{3mm}

\underline{Interval $1$ : $0<|\tilde{\beta}|\leq 1/2$}

\begin{align}
\sigma^2_{\text{HK}}\geq\frac{\Tsnr}{4}.\label{sigmaBoundInterval1}
\end{align}

\vspace{1mm}

\underline{Interval 2 : $1/2<|\tilde{\beta}|\leq \sqrt{|g|}\Tsnr^{1/4-\delta/2}/2$}  \footnote{If $\sqrt{|g|}\Tsnr^{1/4-\delta/2}/2<1/2$ this interval is empty, and we skip to interval $3$.}

For all values of $\Tsnr^{-1/3}<|g|\leq 1$ except for an outage set with measure not greater than $2^{-c}$ we have
\begin{align}
\sigma^2_{\text{HK}}>\frac{2^{-2c}}{4\cdot2^{8}}\frac{\Tsnr^{1/2}}{\sqrt{g^2}}.\label{sigmaBoundInterval2}
\end{align}

\vspace{1mm}

\underline{Interval $3$ : $\sqrt{|g|}\Tsnr^{1/4-\delta/2}/2<|\tilde{\beta}|\leq \Tsnr^{1/4}/\sqrt{8|g|}$}
\begin{align}
\sigma^2_{\text{HK}}\geq \frac{2^{-2c}}{4\cdot 2^8}\frac{\Tsnr^{1/2}}{\sqrt{g^2}}.\label{sigmaBoundInterval3}
\end{align}

\underline{Interval $4$ : $\Tsnr^{1/4}/\sqrt{8|g|}<|\tilde{\beta}|$}
\begin{align}
\sigma^2_{\text{HK}}\geq \frac{1}{8}\frac{\Tsnr^{1/2}}{\sqrt{g^2}}.\label{sigmaBoundInterval4}
\end{align}

Combining the four lower bounds~\eqref{sigmaBoundInterval1},~\eqref{sigmaBoundInterval2},~\eqref{sigmaBoundInterval3}, and~\eqref{sigmaBoundInterval4}, we have
\begin{align}
\sigma^2_{\text{HK}}&\geq\min\Bigg(\frac{1}{4}\Tsnr,\frac{2^{-2c}}{2^{10}}\frac{\Tsnr^{1/2}}{\sqrt{g^2}}
,\frac{1}{8}\frac{\Tsnr^{1/2}}{\sqrt{g^2}} \ \Bigg)\nonumber\\
&=\frac{2^{-2c}}{2^{10}}\frac{\Tsnr^{1/2}}{\sqrt{g^2}}\nonumber
\end{align}
for all $\Tsnr^{-1/3}\leq g^2< 1$ up to an outage set of measure not greater than $2^{-c}$. Thus, substituting our lower bound for $\sigma_{\text{HK}}^2$ into \eqref{symRateHK}, we find that any symmetric rate satisfying
\begin{align}
\rsym&<\frac{1}{2}\log\left(\frac{\Tsnr^{1/2}}{\sqrt{g^2}}\right)-c-5-\frac{1}{2}\log(27)-\frac{1}{2}\log(K^2)\nonumber
\end{align}
is achievable over the symmetric $K$-user interference channel for all $\Tsnr^{-1/3}\leq g^2< 1$ up to an outage set of measure not greater than $2^{-c}$. Since
\begin{align}
\frac{1}{2}\log&\left(\frac{\Tsnr^{1/2}}{\sqrt{g^2}}\right)-c-5-\frac{1}{2}\log(27)-\frac{1}{2}\log(K^2)\nonumber\\
&>\frac{1}{2}\log\left(\frac{\Tsnr^{1/2}}{\sqrt{g^2}}\right)-c-8-\log(K)\nonumber\\
&=\frac{1}{2}\log\left(\frac{\Tsnr}{\sqrt{\Tinr}}\right)-c-8-\log(K), \nonumber
\end{align}
any symmetric rate satisfying
\begin{align}
\rsym&<\frac{1}{2}\log\left(\frac{\Tsnr}{\sqrt{\Tinr}}\right)-c-8-\log(K)
\label{RweakInter}
\end{align}
is achievable.

\vspace{1mm}

\begin{remark}
It follows from the derivation in Appendix~\ref{app:boundsModeratelyWeak} that, as in the strong interference regime, the channel gains that fall within the outage set are the ones close to rational numbers with denominator smaller than $\Tsnr^{1/4}$. Here, the sensitivity of the achievable rate to the ``rationality'' of $g$ increases as $g$ approaches $1$.
\end{remark}

\subsection{Weak Interference Regime} \label{s:weakregime}

This regime is characterized by $1/2\leq \alpha<2/3$, or equivalently, $\Tsnr^{-1/2}\leq g^2<\Tsnr^{-1/3}$. As in the moderately weak interference regime, we develop a closed-form lower bound on the achievable symmetric rate of Corollary~\ref{cor:SymICHK}. A key difference is that the bound derived here is valid for all channel gains, rather than up to an outage set.

We first note that equations~\eqref{symRateHK} and~\eqref{sigmaHK2} continue to hold in this regime as in the moderately weak interference regime, and the optimal choices of $a_1,a_2,a_3$ are also as in~\eqref{integerChoices} and~\eqref{integerChoices2}. As before, we divide the real line into four intervals, give lower bounds on $\sigma_{\text{HK}}^2$ which hold for all values of $\tilde{\beta}$ in each one, and conclude that $\sigma_{\text{HK}}^2$ is lower bounded by the minimum of these four bounds. The lower bounds below are derived in Appendix~\ref{app:boundsWeak}

\vspace{1mm}

\underline{Interval $1$ : $0<|\tilde{\beta}|\leq 1/2$}
\begin{align}
\sigma^2_{\text{HK}}\geq\frac{\Tsnr}{4}.\label{sigmaBoundInterval1b}
\end{align}

\vspace{1mm}

\underline{Interval $2$ : $1/2<|\tilde{\beta}|\leq 1/(2|g|)$}
\begin{align}
\sigma^2_{\text{HK}}\geq \frac{g^2\Tsnr}{4}.\label{sigmaBoundInterval2b}
\end{align}

\vspace{1mm}

\underline{Interval $3$ : $1/(2|g|)<|\tilde{\beta}|\leq \sqrt{g^2\Tsnr/8}$}
\begin{align}
\sigma^2_{\text{HK}}\geq \frac{1}{4 g^4}.\label{sigmaBoundInterval3b}
\end{align}

\vspace{1mm}

\underline{Interval $4$ : $\sqrt{g^2\Tsnr/8}<|\tilde{\beta}|$}

\begin{align}
\sigma^2_{\text{HK}}>\frac{g^2\Tsnr}{4}.\label{sigmaBoundInterval4b}
\end{align}

\vspace{2mm}

Combining the four lower bounds~\eqref{sigmaBoundInterval1b},~\eqref{sigmaBoundInterval2b},~\eqref{sigmaBoundInterval3b}, and~\eqref{sigmaBoundInterval4b}, we have
\begin{align}
\sigma^2_{\text{HK}}&\geq\frac{1}{4}\min\bigg(\Tsnr,g^2\Tsnr,g^{-4}\bigg)\nonumber\\
&=\frac{g^2\Tsnr}{4},\label{sigmaWeakInt}
\end{align}
where~\eqref{sigmaWeakInt} is true since $\Tsnr^{-1/2}\leq g^2<\Tsnr^{-1/3}$. It follows by substituting~\eqref{sigmaWeakInt} into~\eqref{symRateHK} that any symmetric rate satisfying
\begin{align}
\rsym&<\frac{1}{2}\log\left(\frac{g^2\Tsnr}{4}\right)-\frac{1}{2}\log(27)-\frac{1}{2}\log(K^2)\nonumber
\end{align}
is achievable for the symmetric $K$-user interference channel with $\Tsnr^{-1/2}\leq g^2<\Tsnr^{-1/3}$.
Since
\begin{align}
\frac{1}{2}\log&\left(\frac{g^2\Tsnr}{4}\right)-\frac{1}{2}\log(27)-\frac{1}{2}\log(K^2)\nonumber\\
&>\frac{1}{2}\log\left(g^2\Tsnr\right)-\frac{7}{2}-\log(K)\nonumber\\
&=\frac{1}{2}\log\left(\Tinr\right)-\frac{7}{2}-\log(K)\nonumber
\end{align}
any symmetric rate satisfying
\begin{align}
\rsym<\frac{1}{2}\log\left(\Tinr\right)-\frac{7}{2}-\log(K)\label{RweakInter}
\end{align}
is achievable.

\subsection{Noisy Interference Regime} \label{s:noisyinterferenceregime}

The noisy interference regime is characterized by $\alpha<1/2$, or equivalently $g^2<\Tsnr^{-1/2}$. In this regime, each receiver decodes its desired codeword while treating all interfering codewords as noise. Lattice codes are not necessary in this regime in order to approximate the symmetric capacity: random i.i.d. Gaussian codebooks suffice. Nevertheless, the same performance can be achieved with lattice codes as shown in Theorem~\ref{thm:noisyInterference} which states that any symmetric rate
\begin{align}
\rsym<\frac{1}{2}\log\left(1+\frac{\Tsnr}{1+(K-1)g^2\Tsnr}\right)\nonumber
\end{align}
is achievable. It follows that any symmetric rate satisfying
\begin{align}
\rsym&<\frac{1}{2}\log\left(1+\frac{\Tsnr}{1+g^2\Tsnr}\right)-\frac{1}{2}\log(K-1)\nonumber\\
&=\frac{1}{2}\log\left(1+\frac{\Tsnr}{1+\Tinr}\right)-\frac{1}{2}\log(K-1)\label{Rnoisy}
\end{align}
is achievable.

\section{Degrees-of-Freedom} \label{s:dof}
In the previous section, we have shown that the compute-and-forward transform can approximate the capacity of the symmetric $K$-user interference channel up to a constant gap for all channel gains outside a small outage set. Ideally, we would like to use a similar approach to approximate the capacity of the general (non-symmetric) interference channel. In contrast to the symmetric case, where all interferers are automatically aligned (if they all use the same lattice codebook), in a general interference channel the interferers will be observed through different channel gains. A linear combination of lattice codewords is always a codeword only if all of the coefficients are integers. Thus, in order to induce alignment, all of the interfering gains should be steered towards integers, which is an overconstrained problem.

The compute-and-forward transform proposed in this paper is quite general, in that its performance can be evaluated for any Gaussian interference network, and it can be combined with precoding schemes that induce alignment. For instance, consider the class of \emph{real interference alignment} precoding schemes that transform the channel seen by each receiver in a non-symmetric interference channel to an effective MAC where some of the interfering users are aligned. A remarkable example of such a scheme is that of~\cite{mgmk09}, which is used to prove that the DoF offered by almost every Gaussian $K$-user interference channel is $K/2$. To date, essentially all real interference alignment schemes utilized a scalar lattice constellation (e.g., $p$-ary pulse amplitude modulation), concatenated with a random i.i.d.~outer code. Potentially, replacing this construction with AWGN capacity achieving $n$-dimensional lattice codes can improve the performance of such schemes and may eventually lead to achievable rate regions that outperform TDMA at reasonable values of SNR. Here, we take a first step and verify that the compute-and-forward transform can attain the same high SNR asymptotics.

Specifically, we show that for almost every $K$-user MAC, each user can achieve $1/K$ DoF using the compute-and-forward transform. In~\cite{mgmk09}, it is shown that the same is true using a scalar lattice concatenated with a random i.i.d.~outer code and maximum likelihood decoding. This result is then used as a building block for the interference alignment scheme. Since real interference alignment schemes often induce effective multiple-access channels whose coefficients are dependent \cite{mgmk09}, our analysis assumes that the channel coefficients belong to a \emph{manifold}, and our results apply for a set of full Lebesgue measure with respect to the considered manifold.

Theorem~\ref{thm:SumRate} in Section~\ref{s:mac} guarantees that the sum of the optimal computation rates is close to the sum capacity of the MAC. However, the theorem does not tell us how the sum rate is divided between the $K$ rates. We now show that,
in a DoF sense, the sum is equally split between all $K$ rates for almost every channel realization. Recall the definition for DoF:
\begin{align}
d_{\text{comp},k}=\lim_{\Tsnr\rightarrow\infty}\frac{R_{\text{comp},k}(\Tsnr)}{\frac{1}{2}\log(1+\Tsnr)} \ .
\end{align}
First, we upper bound $d_{\text{comp},1}$, the DoF provided by the highest computation rate.

\vspace{1mm}

\begin{theorem}
\label{thm:DoF}
Let $f_1,\ldots,f_K$ be functions from $\RR^m$ to $\RR$ satisfying
\begin{enumerate}
\item $f_k$ for $k=1,\ldots,K$ is analytic in $\RR^m$,
\item $1,f_1,\ldots,f_K$ are linearly independent over $\RR$,
\end{enumerate}
and define the manifold
\begin{align}
\mathcal{M}=\left\{\big[f_1(\mathbf{\tilde{h}}) \ \cdots \ f_K(\mathbf{\tilde{h}})\big] \ : \ \mathbf{\tilde{h}}\in\RR^m\right\}.\label{manifoldDef}
\end{align}
For almost every $\bh\in\mathcal{M}$, the DoF offered by the highest computation rate is upper bounded by
\begin{align}
d_{\text{comp},1}\leq\frac{1}{K} \ .
\end{align}
\end{theorem}

\vspace{1mm}

The proof is given in Appendix~\ref{app:DoFproof}, and is based on showing that restricting the scaling coefficient $\beta$ from~\eqref{CoFestimate} to the form $\beta=q/h_1$ for $q\in\ZZ$ (almost surely) incurs no loss from a DoF point of view. This way, the first coefficient of $\beta\bh$ is an integer. Then, a result from the field of Diophantine approximation which is due to Kleinbock and Margulis~\cite{km98} is used in order to lower bound the error in approximating the remaining $K-1$ channel gains with integers.

\vspace{1mm}

As a special case of Theorem~\ref{thm:DoF} we may choose the manifold $\mathcal{M}$ as $\RR^K$ which implies the following corollary.

\vspace{1mm}

\begin{corollary}
For almost every $\bh\in\RR^K$ the DoF offered by the highest computation rate is upper bounded by
\begin{align}
d_{\text{comp},1}\leq \frac{1}{K}.\nonumber
\end{align}
\label{cor:DoFoverR}
\end{corollary}

\vspace{1mm}

\begin{remark}
Niesen and Whiting~\cite{nw11} studied the DoF offered by the highest computation rate and showed that
\begin{align}
d_{\text{comp},1}\leq\bigg\{\begin{array}{cc}
                              1/2 & K=2 \\
                              2/(K+1) & K>2
                            \end{array}
\end{align}
for almost every $\bh\in\RR^K$. Our bound therefore agrees with that of~\cite{nw11} for $K=2$ and improves it for $K>2$.
\end{remark}

\vspace{1mm}

The next corollary shows that all $K$ optimal computation rates offer $1/K$ DoF for almost every $\bh$ satisfying mild conditions.
\begin{corollary}
Let $\mathcal{M}$ be a manifold satisfying the conditions of Theorem~\ref{thm:DoF}.
For almost every $\bh\in\mathcal{M}$ the DoF provided by each of the $K$ optimal computation rates is
$d_{\text{comp},k}=1/K$.
\label{thm:symmetricDoF}
\end{corollary}

\vspace{2mm}

\begin{proof}
Theorem~\ref{thm:SumRate} implies that $\sum_{k=1}^K d_{\text{comp},k}\geq 1$. Using the fact that $d_{\text{comp},k}$ is monotonically decreasing in $k$ and that $d_{\text{comp},1}\leq 1/K$ for almost every $\bh\in\mathcal{M}$, the corollary follows.
\end{proof}

The corollary above implies that, in the limit of very high SNR, not only is the sum of computation rates close to the sum capacity of the MAC, but each computation rate scales like the symmetric capacity of the MAC for almost all channel gains.
Note that our analysis (as well as that of \cite{nw11}) is within the context of the achievable computation rates stemming from Theorem \ref{thm:CoF}.

\vspace{1mm}

The next corollary follows from Corollary~\ref{thm:symmetricDoF} and Theorem~\ref{thm:MAC}.

\vspace{1mm}

\begin{corollary}
\label{cor:DoFMAC}
Let $\mathcal{M}$ be a manifold satisfying the conditions of Theorem~\ref{thm:DoF}.
The DoF attained by each user in the $K$-user MAC under the compute-and-forward transform is $1/K$ for almost every $\bh\in\mathcal{M}$. In particular, the DoF attained by each user in the $K$-user MAC under the compute-and-forward transform is $1/K$ for almost every $\bh\in\RR^{K}$
\end{corollary}

\vspace{1mm}

The next theorem shows that for almost every effective $L$-user multiple access channel of the form introduced in Section~\ref{sub:effmac} each of the effective users achieves $1/L$ degree of freedom. The proof is given in Appendix~\ref{app:DoFeffProof}.

\vspace{1mm}

\begin{theorem}
\label{thm:DoFeff}
Let $f_1,\ldots,f_L$ be functions from $\RR^m$ to $\RR$ satisfying
\begin{enumerate}
\item $f_\ell$ for $\ell=1,\ldots,L$ is analytic in $\RR^m$,
\item $1,f_1,\ldots,f_L$ are linearly independent over $\RR$,
\end{enumerate}
and define the manifold
\begin{align}
\mathcal{M}=\Big\{\big[f_1(\tilde{\bg}) \ \cdots \ f_L(\tilde{\bg})\big] \ : \ \tilde{\bg}\in\RR^m\Big\}.\nonumber
\end{align}
For almost every $\bg\in\mathcal{M}$ the DoF offered by each of the $L$ optimal computation rates for the effective MAC~\eqref{effectiveMAC} is
\begin{align}
d_{\text{comp},\ell}=\lim_{\Tsnr\rightarrow\infty}\frac{R_{\text{comp},\ell}(\Tsnr)}{\frac{1}{2}\log(1+\Tsnr)}=\frac{1}{L}.
\end{align}
\end{theorem}

\vspace{1mm}

\begin{remark}
The manifold $\mathcal{M}$ is the same manifold used in~\cite{mgmk09}. This manifold was general enough to allow the derivation of the DoF characterization of the $K$-user interference channel in~\cite{mgmk09}. Thus, the DoF results from~\cite{mgmk09} can be re-derived using the same real interference alignment scheme from~\cite{mgmk09} with $n$-dimensional lattice codes instead of 1-D integer constellations concatenated with outer codes.
\end{remark}

\section{Discussion}\label{s:dicussion}
In this paper, we have developed a new decoding framework for lattice-based interference alignment.  We used this framework as a building block for two lattice-based interference alignment schemes for the symmetric real Gaussian $K$-user interference channel. These schemes perform well starting from the moderate SNR regime, and are within a constant gap from the upper bounds on the capacity for all channel gains outside of some outage set whose measure can be controlled.

A natural question for future research is how to extend the results above to the general Gaussian $K$-user interference channel. The main problem is that, in the general case, the interfering lattice codewords are not naturally aligned, as their gains are not integer-valued. Therefore, in order to successfully apply lattice interference alignment, some form of precoding, aimed towards forcing the cross channel gains to be integers, is required. Unfortunately, simple power-backoff strategies do not suffice, even in the three-user case.

One option for overcoming this problem is to use many layers at each transmitter, as in~\cite{mgmk09}, and create partial alignment between interfering layers. While this achieves the optimal DoF, it performs poorly at reasonable values of SNR, as there will be a rate loss for each additional layer. As a result the rate region obtained by combining the compute-and-forward transform with the precoding scheme of~\cite{mgmk09} is inferior to that obtained by time-sharing, for values of SNR of practical interest. Another option is to precode not only using power-backoff, but also over time, which may partially compensate for the lack of sufficient free parameters. An example for such a precoding scheme is the power-time code introduced in~\cite{oe13}.

A positive feature of the compute-and-forward framework is that it does not require perfect alignment of the lattice points participating in the integer linear combinations. Namely, the effect of not perfectly equalizing the channel gains to integers is an enhanced effective noise. For the general interference channel, this suggests that it may suffice to find precoding schemes that only approximately force the cross-channel gains to integers.
%

\begin{appendices}

\section{Proof of Theorem~\ref{thm:MAC}}
\label{app:MACproof}
We begin with two lemmas which will be useful for the proof of Theorem~\ref{thm:MAC}.

\vspace{2mm}

\begin{lemma}
\label{lem:ModuloTriangularization}
Let $\bA$ be a $K\times K$ matrix with integer entries of magnitudes bounded from above by some constant $a_{\text{max}}$. If there exists a real-valued $K \times K$ lower triangular matrix $\bL$ with unit diagonal such that $\mathbf{\tilde{A}}=\bL\bA$ is upper triangular up to column permutation $\mathbf{\pi}$, then for any prime $p>K (K!)^2(K a_{\text{max}})^{2K}a_{\text{max}}$ there also exists a lower triangular matrix $\bL^{(p)}$ with elements from $\{0,1,\ldots,p-1\}$ and unit diagonal such that $\mathbf{\tilde{A}}^{(p)}=\left[\bL^{(p)}\bA\right]\bmod{p}$ is upper triangular up to column permutation $\mathbf{\pi}$.
\end{lemma}

\vspace{2mm}

\begin{proof}
Assume that there exists a lower triangular matrix $\bL$ with unit diagonal such that $\mathbf{\tilde{A}}=\bL\bA$ is upper triangular up to column permutation $\mathbf{\pi}$.
We begin by showing that all elements in the $i$th ($i>1$) row of $\bL$ can be written as rational numbers with the same denominator $1\leq q_i\leq K!(K a^2_{\text{max}})^{K}$. To see this note that if $\mathbf{\tilde{A}}$ is triangular up to column permutation vector $\mathbf{\pi}$, then its $i$th row contains at least $i-1$ zeros, namely $\tilde{a}_{ij}=0$ for $j=\pi(1),\ldots,\pi(i-1)$. Since $\bL$ is lower triangular, the following equations must hold
\begin{align}
\tilde{a}_{ij}=\sum_{m=1}^{i}\ell_{im}a_{mj}=0, \ \text{for } j=\pi(1),\ldots,\pi(i-1).\label{triangEq1}
\end{align}
By definition $\ell_{ii}=1$, therefore~\eqref{triangEq1} can be written as
\begin{align}
\sum_{m=1}^{i-1}\ell_{im}a_{mj}=-a_{ij}, \ \text{for } j=\pi(1),\ldots,\pi(i-1).\label{triangEq2}
\end{align}
Define the vectors \mbox{$\bm{\ell}^{(i)}=[\ell_{i1} \ \cdots \ \ell_{i,i-1}]^T$}, \mbox{$\ba^{(i\mathbf{\pi})}=-[a_{i,\pi(1)} \ \cdots \ a_{i\pi(i-1)}]^T$} and the matrix
\begin{align}
\bA^{(i,\mathbf{\pi})}=\left(
  \begin{array}{ccc}
    a_{1\pi(1)} & \hdots & a_{i-1\pi(1)} \\
    \vdots & \ddots & \vdots \\
    a_{1\pi(i-1)} & \hdots & a_{i-1\pi(i-1)} \\
  \end{array}
\right).\nonumber
\end{align}
We have,
\begin{align}
\bA^{(i,\mathbf{\pi})}\bm{\ell}^{(i)}=\ba^{(i,\mathbf{\pi})}.\label{triangEq3}
\end{align}
From the fact that $\bA$ can be pseudo-triangularized with permutation vector $\mathbf{\pi}$, we know that the system of equations~\eqref{triangEq3} has a solution.
Assume that
\begin{align}
\rank\left(\bA^{(i,\mathbf{\pi})}\right)=u\leq i-1.\nonumber
\end{align}
It follows that there are $u$ linearly independent columns in $\bA^{(i,\mathbf{\pi})}$. Let $\mathcal{U}\subseteq\{1,\ldots,i-1\}$ be a set of indices corresponding to $u$ such linearly independent columns, and $\bar{\mathcal{U}}$ be its complement. Let $\mathbf{{A}}_{\mathcal{U}}^{(i,\mathbf{\pi})}\in\ZZ^{i-1\times u}$ be the matrix obtained by taking the columns of $\bA^{(i,\mathbf{\pi})}$ with indices in $\mathcal{U}$. Since~\eqref{triangEq3} has a solution, we have $\ba^{(i,\mathbf{\pi})}\in\Span\left(\mathbf{{A}}_{\mathcal{U}}^{(i,\mathbf{\pi})}\right)$. Thus, we can set $\bm{\ell}^{(i)}(k)=0$ for all $k\in\bar{\mathcal{U}}$, and~\eqref{triangEq3} will still have a solution. Letting $\bm{{\ell}}_{\mathcal{U}}^{(i)}\in\RR^{u\times 1}$ be the vector obtained by taking from $\bm{{\ell}}^{(i)}$ only the entries with indices in $\mathcal{U}$, it follows that
\begin{align}
\mathbf{{A}}_{\mathcal{U}}^{(i,\mathbf{\pi})}\mathbf{{\ell}}_{\mathcal{U}}^{(i)}=\ba^{(i,\mathbf{\pi})}\label{triangEq4}
\end{align}
has a solution. Now, multiplying both sides of~\eqref{triangEq4} by $\left(\mathbf{{A}}_{\mathcal{U}}^{(i,\mathbf{\pi})}\right)^T$ gives
\begin{align}
\mathbf{{A'}}^{(i,\mathbf{\pi})}\bm{{\ell}}_{\mathcal{U}}^{(i)}=\ba'^{(i,\mathbf{\pi})},\label{triangEq5}
\end{align}
where $\mathbf{{A'}}^{(i,\mathbf{\pi})}=\left(\mathbf{{A}}_{\mathcal{U}}^{(i,\mathbf{\pi})}\right)^T\mathbf{{A}}_{\mathcal{U}}^{(i,\mathbf{\pi})}\in\ZZ^{u\times u}$ is a full-rank matrix and $\ba'^{(i,\mathbf{\pi})}=\left(\mathbf{{A}}_{\mathcal{U}}^{(i,\mathbf{\pi})}\right)^T\ba^{(i,\mathbf{\pi})}\in\ZZ^{u\times 1}$. Note that all entries of $\mathbf{{A'}}^{(i,\mathbf{\pi})}$ as well as all entries of $\ba'^{(i,\mathbf{\pi})}$ have magnitude bounded from above by $\tilde{a}_{\text{max}}\triangleq u a^2_{\text{max}}$.
Cramer's rule for solving a system of linear equations (see e.g.,~\cite{harville}) implies that all elements of $\bm{{\ell}}_{\mathcal{U}}^{(i)}$ can be expressed as rational numbers with denominator $q_i\triangleq|\det(\bA'^{(i,\mathbf{\pi})})|$.
Recall the Leibnitz formula (see e.g.,~\cite{harville}) for the determinant of an $n\times n$ matrix $\bG$
\begin{align}
\det(\bG)=\sum_{\sigma\in\mathcal{S}_n}\sign(\sigma)\prod_{i=1}^n G_{i,\sigma_i},
\end{align}
where $\mathcal{S}_n$ is the set of all permutations of $\{1,\ldots,n\}$.
It follows that $\det\left(\bA'^{(i,\mathbf{\pi})}\right)$ must be an integer and in addition $1 \leq|\det\left(\bA'^{(i,\mathbf{\pi})}\right)|\leq u!(\tilde{a}_{\text{max}})^{u}$. Thus, $1\leq q_i \leq u!(\tilde{a}_{\text{max}})^{u}$. Moreover, Cramer's rule also implies that the numerator of each element in $\bm{{\ell}}_{\mathcal{U}}^{(i)}$ is an integer not greater than $u!(\tilde{a}_{\text{max}})^{u}$ in magnitude.
Since $u\leq K$, and since each element of $\bm{{\ell}}^{(i)}$ is either zero or corresponds to an element in $\bm{{\ell}}_{\mathcal{U}}^{(i)}$, each element $\ell_{ij}$, $j\leq i$ of $\bL$ can be written as a rational number $\ell_{ij}=m_{ij}/q_i$
with $1\leq q_i\leq K!(K a^2_{\text{max}})^{K}$ and $|m_{ij}|\leq K!(K a^2_{\text{max}})^{K}$ for $i=1,\ldots,K$.

Now, define the matrix $\mathbf{\tilde{L}}=\diag(q_1,\ldots,q_K)\bL$ and note that $\mathbf{\tilde{L}}\in\ZZ^{K\times K}$ due to the above. Let $\mathbf{\tilde{A}}'^{(p)}=[\mathbf{\tilde{L}}\bA]\bmod p$.
Since multiplying a row in a matrix by a constant leaves its zero entries unchanged, the entries of the matrix
\begin{align}
\mathbf{\tilde{A}}'^{(p)}&=[\mathbf{\tilde{L}}\bA]\bmod p\nonumber\\
&=[\diag(q_1,\ldots,q_K)\bL\bA]\bmod p\nonumber\\
&=[\diag(q_1,\ldots,q_K)\mathbf{\tilde{A}}]\bmod p,\label{modPeq}
\end{align}
are zero whenever the entries of $\mathbf{\tilde{A}}$ are equal to zero. Moreover, since all elements of $\bL$ are bounded in magnitude by $K!(K a^2_{\text{max}})^K$ and all elements of $\bA$ are bounded in magnitude by $a_{\text{max}}$, all elements of $\mathbf{\tilde{A}}=\bL\bA$ are bounded in magnitude by $K K!(K a^2_{\text{max}})^{K}a_{\text{max}}$. Combining with the fact that $1\leq q_i \leq K!(K a^2_{\text{max}})^{K}$, we have $|\tilde{a}'^{(p)}_{ij}|\leq K (K!)^2(K a_{\text{max}})^{2K}a_{\text{max}}$ for all $i=1,\ldots,K$, $j=1,\ldots,K$. Therefore, for a prime number $p>K (K!)^2(K a_{\text{max}})^{2K}a_{\text{max}}$ the modulo reduction in~\eqref{modPeq} does not change any of the non-zero entries of $\diag(q_1,\ldots,q_K)\mathbf{\tilde{A}}$ to zero.

Recall that if $\bA$ can be pseudo-triangularized with a matrix $\bL$ and permutation vector $\mathbf{\pi}$ then $\tilde{a}_{i,\pi(i)}\neq 0$, and hence also $\tilde{a}'^{(p)}_{i,\pi(i)}\neq 0$ for $i=1,\ldots,K$. We have therefore shown that for $p$ large enough there exists a lower-triangular matrix $\mathbf{\tilde{L}}^{(p)}=[\diag(q_1,\ldots,q_K)\bL]\bmod p$ with elements from $\{0,1,\ldots,p-1\}$ such that $\mathbf{\tilde{A}}'^{(p)}=[\mathbf{\tilde{L}}^{(p)}\bA]\bmod p$ is upper-triangular up to column permutation $\mathbf{\pi}$. In order to complete the proof, it is left to transform $\mathbf{\tilde{L}}^{(p)}$ to a lower-triangular matrix with elements from $\{0,1,\ldots,p-1\}$ and \emph{unit diagonal}. Let $(q_i)^{-1}$ be an integer that satisfies $[(q_i)^{-1} q_i]\bmod p=1$. Such an integer always exists since $q_i$ is an integer different than zero, and $p$ is prime. It is easy to verify that the matrix $\bL^{(p)}=[\diag\left((q_1)^{-1},\ldots,(q_K)^{-1}\right)\mathbf{\tilde{L}}^{(p)}]\bmod p$ is a lower-triangular matrix with elements from $\{0,1,\ldots,p-1\}$ and unit diagonal, and $\mathbf{\tilde{A}}^{(p)}=\left[\bL^{(p)}\bA\right]\mod p$ is upper triangular up to column permutation $\mathbf{\pi}$.
\end{proof}

\vspace{2mm}

\begin{lemma}
\label{lem:ModOperetaions}
Let $\bt_1,\ldots,\bt_k$ be lattice points from a chain of nested lattices satisfying the conditions of Theorem~\ref{thm:CoF}. Let \mbox{$\bv=\big[\sum_{k=1}^K a_k\bt_k\big]\Mod$} and \mbox{$\bu=\big[\sum_{k=1}^K b_k\bt_k\big]\Mod$} be integer linear combinations of these points. Then
\begin{align}
\left[\bv+\bu\right]\Mod=\left[\sum_{k=1}^K\big((a_k+b_k)\bmod p\big) \bt_k\right]\Mod.\nonumber
\end{align}
\end{lemma}
\begin{proof}
Due to the distributive property of the modulo operation we have
\begin{align}
\left[\bv+\bu\right]\Mod&=\left[\sum_{k=1}^K(a_k+b_k)\bt_k\right]\Mod.\nonumber\\
&=\left[\sum_{k=1}^K[(a_k+b_k)\bmod p + M_k\cdot p]\bt_k\right]\Mod\nonumber\\
&=\bigg[\sum_{k=1}^K\big((a_k+b_k)\bmod p\big) \bt_k \nonumber\\
&+\sum_{k=1}^K M_k\cdot [p\cdot\bt_k]\Mod\bigg]\Mod
\end{align}
where $\{M_k\}_{k=1}^K$ are some integers. Utilizing the fact that $[p\cdot\bt_k]\Mod=\mathbf{0}$ for all lattice points in the chain, which follows from Theorem~\ref{thm:CoF}\eqref{thm:CoF3}, the lemma is established.
\end{proof}

\vspace{1mm}

We are now ready to prove Theorem~\ref{thm:MAC}.

\vspace{1mm}

\begin{proof}[Proof of Theorem~\ref{thm:MAC}]
Let $\bT=[\bt_1^T \ \cdots\  \bt^T_K]^T$ and $\bV=[\bv_1^T \ \cdots\  \bv_K^T]^T=[\bA\bT]\Mod$. The compute-and-forward transform of the MAC~\eqref{MACchannel} can be written as
\begin{align}
\bS&=\left[\bA \ \left(
                       \begin{array}{c}
                         \bt_1 \\
                         \vdots \\
                         \bt_K \\
                       \end{array}
                     \right)+{\bZ}_{\text{eff}} \right]\bmod\Lambda\nonumber\\
                     &=[\bA\bT+{\bZ}_{\text{eff}}]\Mod\nonumber\\
                     &=[\bV+{\bZ}_{\text{eff}}]\Mod.\nonumber
\end{align}
Assume there exists a pseudo-triangularization of ${\bA}$ with permutation vector $\mathbf{\pi}$, i.e., there exists a lower triangular matrix $\bL$ with unit diagonal such that $\mathbf{\tilde{A}}=\bL\bA$ is upper triangular up to column permutation $\mathbf{\pi}$. Lemma~\ref{lem:ModuloTriangularization} implies that there exists a lower triangular matrix $\bL^{(p)}$ with elements from $\{0,1,\ldots,p-1\}$ and unit diagonal such that \mbox{$\mathbf{\tilde{A}}^{(p)}=\left[\bL^{(p)}\bA\right]\mod p$} is upper triangular up to column permutation $\mathbf{\pi}$. Since $\bL^{(p)}$ has a unit diagonal it can be written as $\bL^{(p)}=\bI+\bR$ where $\bI$ is the identity matrix and $\bR$ has non-zero entries only below the main diagonal.

Assume the receiver has access to the side information $\bv_1,\ldots,\bv_{K-1}$. As the entries of $\bR$ are non-zero only below the main diagonal, the receiver could compute $\bR\cdot\bV$, add it to $\bS$ and reduce modulo $\Lambda$, giving rise to
\begin{align}
\bS^{\text{SI}}&=\left[\bS+\bR\cdot\bV\right]\Mod\nonumber\\
&=\left[\bA\bT+\bR\bA\bT+{\bZ}_{\text{eff}}\right]\Mod\nonumber\\
&=\left[(\bI+\bR)\bA\bT+{\bZ}_{\text{eff}}\right]\Mod\nonumber\\
&=\left[\bL^{(p)}\bA\bT+{\bZ}_{\text{eff}}\right]\Mod\nonumber\\
&=\left[[\bL^{(p)}\bA]\bmod p\cdot \bT+{\bZ}_{\text{eff}}\right]\Mod\label{modEquallity}\\
&=\left[\mathbf{\tilde{A}}^{(p)} \left(
                       \begin{array}{c}
                         \bt_1 \\
                         \vdots \\
                         \bt_K \\
                       \end{array}
                     \right)+{\bZ}_{\text{eff}}\right]\Mod\nonumber
\end{align}
where~\eqref{modEquallity} follows from Lemma~\ref{lem:ModOperetaions}.
Let $\mathbf{\tilde{V}}=[\mathbf{\tilde{A}}^{(p)}\bT]\Mod$ and recall that $\mathbf{\tilde{A}}^{(p)}$ is upper-triangular up to column permutation $\mathbf{\pi}$, thus $\tilde{a}^{(p)}_{j,\pi(m)}=0$ for all $j=\pi(m)+1,\ldots,K$.
It follows that for any $m<K$ the lattice point $\bt_{\pi(m)}$ does not participate in any of the linear combinations $\mathbf{\tilde{v}}_{m+1},\ldots,\mathbf{\tilde{v}}_K$.

Assume the mapping function between users and lattices is chosen as $\map(k)=\pi^{-1}(k)$, i.e., each user $k$ employs the codebook $\mathcal{L}_k=\Lambda_{\pi^{-1}(k)}\cap\Lambda$. In this case, the densest lattice participating in linear combination $\mathbf{\tilde{v}}_{m}$ is $\Lambda_m$.
The decoder uses $\bs_m^{\text{SI}}$ in order to produce an estimate
\begin{align}
\mathbf{\hat{\tilde{v}}}_m=\left[Q_{\Lambda_m}(\bs_m^{\text{SI}})\right]\Mod
\end{align}
for each one of the linear combinations $\mathbf{\tilde{v}}_m$, $m=1,\ldots,K$. It follows from Theorem~\ref{thm:CoF} that there exists a chain of nested lattices $\Lambda\subseteq\Lambda_K\subseteq\cdots\subseteq\Lambda_1$ forming the set of codebooks $\mathcal{L}_1,\ldots,\mathcal{L}_K$ with rates $R_1,\ldots,R_K$ such that all linear combinations $\mathbf{\tilde{v}}_1,\ldots,\mathbf{\tilde{v}}_K$ can be decoded with a vanishing error probability as long as the rates of all users satisfy the constraints of~\eqref{rateAllocation}.

We have shown that if the receiver has access to $\bv_1,\ldots,\bv_{K-1}$ it can decode the set of linear combinations $\mathbf{\tilde{V}}$. We now show a sequential decoding procedure which guarantees that the receiver has the right amount of side information at each step. First, note that
\begin{align}
\bs_m^{\text{SI}}=\left[\bs_m+\sum_{\ell=1}^{m-1}r_{m \ell}\bv_\ell\right]\Mod,
\end{align}
thus the necessary side information for decoding $\mathbf{\tilde{v}}_m$ is only $\bv_1,\ldots,\bv_{m-1}$. In particular, $\bs_1^{\text{SI}}=\bs_1$ and hence $\bv_1$ can be decoded with a vanishing error probability with no side information. After decoding $\bv_1$ the receiver has it as side information, and can therefore compute $\bs_2^{\text{SI}}$ and decode $\mathbf{\tilde{v}}_2$. As $\mathbf{\tilde{v}}_2=[r_{21}\bv_1+\bv_2]\Mod$ and the receiver knows $\bv_1$, it can use it in order to recover $\bv_2$. Now, the receiver has $\bv_1$ and $\bv_2$ as side information and can use it to compute $\bs_3^{\text{SI}}$. The process continues sequentially until all linear combinations $\mathbf{\tilde{v}}_1,\ldots,\mathbf{\tilde{v}}_K$ are decoded.

Conditioned on correct decoding, we obtain $K$ noiseless linear combinations
\begin{align}
\left(
  \begin{array}{c}
    \mathbf{\tilde{v}}_1 \\
    \vdots \\
    \mathbf{\tilde{v}}_K \\
  \end{array}
\right)=\left[\mathbf{\tilde{A}}^{(p)}\left(
                            \begin{array}{c}
                              \bt_1 \\
                              \bt_2 \\
                              \vdots \\
                              \bt_K \\
                            \end{array}
                          \right)
\right]\bmod\Lambda.\label{noiselessEqs}
\end{align}
Since $\mathbf{\tilde{A}}^{(p)}$ is upper-triangular up to column permutation, and in particular full-rank modulo $p$, the original lattice points $\bt_1,\ldots,\bt_K$ each user transmitted can be recovered.
\end{proof}

\section{Proof of Theorems~\ref{thm:EffSumRate} and~\ref{thm:MACeff}}
\label{app:effproofs}

\begin{proof}[Proof of Theorem~\ref{thm:EffSumRate}]
The proof is identical to that of Theorem~\ref{thm:SumRate} with $\bF=\left(\Tsnr^{-1}\bB^{-1}+\bg\bg^T\right)^{-1/2}$.
\end{proof}

\vspace{2mm}

\begin{proof}[Proof of Theorem~\ref{thm:MACeff}]
Let
\begin{align}
{\bS}=\left[\bA\left(
                       \begin{array}{c}
                         \bt_{\text{eff},1} \\
                         \vdots \\
                         \bt_{\text{eff},L} \\
                       \end{array}
                     \right)+\bZ_{\text{eff}} \right]\bmod\Lambda\nonumber
\end{align}
be the compute-and-forward transform of the effective $L$-user MAC, and assume that $\bA$ can be pseudo-triangularized with permutation vector $\mathbf{\pi}$. Repeating the proof of Theorem~\ref{thm:MAC} it is easy to see that, for any set of rates
\begin{align}
R_{\ell}<R_{\text{comp},\mathbf{\pi}^{-1}(\ell)}, \ \ell=1,\ldots,L\nonumber \ ,
\end{align}
there exists a chain of nested lattices $\Lambda\subseteq\Lambda_L\subseteq\cdots\subseteq\Lambda_1$ inducing the codebooks $\mathcal{L}_\ell=\Lambda_{\mathbf{\pi}^{-1}(\ell)}\cap\CV$ with rates $R_{\ell}$, such that if $\bt_{\text{eff},\ell}\in\mathcal{L}_{\ell}$ for all $\ell=1,\ldots,L$, all effective lattice points can be decoded from $\bS$.

If each of the users $i\in\mathcal{K}_{\ell}$ that comprise effective user $\ell$ uses the lattice codebook $\mathcal{L}_{\ell}$ (or any codebook nested in $\mathcal{L}_\ell$), then $\bt_{\text{eff},\ell}\in\mathcal{L}_{\ell}$ and all effective lattice points can be decoded. 
\end{proof}

\section{Proof of Lemma~\ref{lem:equiEffChannel}}
\label{app:CsymProofs}

In order to decode the desired effective lattice points, it suffices to decode $L-1$ linearly independent integer linear combinations of them, in which $\bt_{\text{eff},L}$ does not participate. Let $\bar{\ba}=[\bar{a}_1 \ \cdots \ \bar{a}_{L-1} \ 0]^T$ be some coefficient vector for such a linear combination. The effective rate for computing the linear combination $\bar{\bv}=\big[\sum_{\ell=1}^{L-1}\bar{a}_\ell\bt_{\text{eff},\ell}\big]\Mod$ with the coefficient vector $\bar{\ba}$ over the channel~\eqref{effectiveMAC} is
\begin{align}
R_{\text{comp}}(\bg,\bar{\ba},\bB)=\frac{1}{2}\log\left(\frac{\Tsnr}{\sigma^2_{\text{eff}}(\bg,\bar{\ba},\bB)}\right),
\end{align}
where
\begin{align}
\sigma^2_{\text{eff}}(\bg,\bar{\ba},\bB)&=\min_{\bar{\beta}\in\RR}\Tsnr\sum_{\ell=1}^{L-1}(\bar{\beta}g_\ell-\bar{a}_\ell)^2 b^2_{\text{eff},\ell}\nonumber\\
&+\bar{\beta}^2(1+\Tsnr g^2_L b^2_{\text{eff},L})\nonumber\\
&=\min_{\beta\in\RR}\Tsnr\sum_{\ell=1}^{L-1}(\beta\kappa g_\ell-\bar{a}_\ell)^2 b^2_{\text{eff},\ell}+\beta^2\label{eqAlpha},
\end{align}
where~\eqref{eqAlpha} follows by substituting $\bar{\beta}=\beta\kappa$. The effective noise variance and computation rate for decoding a linear combination with coefficient vector $\bar{\ba}=[\bar{a}_1 \ \cdots \ \bar{a}_{L-1} \ 0]$ over the effective channel~\eqref{effectiveMAC} are therefore the same as those of decoding a linear combination with $\ba=[\bar{a}_1 \ \cdots \ \bar{a}_{L-1}]$ over the effective channel~\eqref{zeroEffMAC}. Thus, for purposes of decoding integer linear combinations of effective lattice points $\bt_{\text{eff},1},\ldots,\bt_{\text{eff},L-1}$ the two channels are equivalent. Since this is all we need in order to decode $\bt_{\text{eff},1},\ldots,\bt_{\text{eff},L-1}$, the lemma follows.

\section{Derivation of the upper bounds on $\sigma^2_{\text{HK}}$ within the different intervals}

\subsection{Moderately weak interference regime}
\label{app:boundsModeratelyWeak}
We upper bound $\sigma^2_{\text{HK}}$ for all values of  $\tilde{\beta}$ within each of the four intervals. Recall that in the moderately weak interference regime $\Tsnr^{-1/3}\leq g^2\leq 1$.
Define $\delta=(2c+8)/\log(\Tsnr)$, where $c>0$ is some constant.

\vspace{1mm}

\underline{Interval $1$ : $0<|\tilde{\beta}|\leq 1/2$}

In this interval the choice $a_1=\sign(\tilde{\beta})$ is optimal due to~\eqref{integerChoices2}. Therefore, for all $|\tilde{\beta}|\leq1/2$ we have $(\tilde{\beta}-a_1)^2\geq1/4$. Thus,
\begin{align}
\sigma^2_{\text{HK}}\geq\frac{\Tsnr}{4}.\nonumber
\end{align}

\vspace{1mm}

\underline{Interval 2 : $1/2<|\tilde{\beta}|\leq \sqrt{|g|}\Tsnr^{1/4-\delta/2}/2$}  

Since $|\tilde{\beta}|>1/2$ we can express it as $\tilde{\beta}=q+\varphi$ with $q\in\mathbb{Z}\setminus 0$ and $\varphi\in[-1/2,1/2)$. We can further lower bound $\sigma^2_{\text{HK}}$ as
\begin{align}
\sigma^2_{\text{HK}}&>\min_{\varphi,q,a_1,a_3}\bigg(\left((\varphi+q-a_1)^2+(qg-a_3+\varphi g)^2\right)\Tsnr \bigg) \nonumber\\
&=\min_{\varphi,q,a_3}\bigg(\left(\varphi^2+(qg-a_3+\varphi g)^2\right)\Tsnr \bigg).\label{phibound1}
\end{align}
Ignoring the constraint $\varphi\in[-1/2,1/2)$, the minimizing value of $\varphi$ is found to be
\begin{align}
\varphi^*=-\frac{g}{1+g^2}(qg-a_3).\nonumber
\end{align}
Substituting $\varphi^*$ into~\eqref{phibound1} gives
\begin{align}
\sigma^2_{\text{HK}}&>\min_{q,a_3}\bigg(\frac{1}{1+g^2}(qg-a_3)^2\Tsnr \bigg)\nonumber\\
&\geq\frac{1}{2}\min_{q,a_3}\bigg((qg-a_3)^2\Tsnr \bigg).\label{phibound2}
\end{align}
For $b=1,2,\ldots,\lceil1/6\log(\Tsnr)\rceil$ we define the sets
\begin{align}
\mathcal{G}_b=\left\{g : g\in\left[2^{-b},2^{-b+1}\right) \right\},
\end{align}
and the quantities
\begin{align}
q_{\text{max},b}&\triangleq \sqrt{2^{-b+1}}\Tsnr^{1/4-\delta/2},\nonumber\\
\Phi_{b}&\triangleq \frac{1}{\sqrt{2^{-b+1}}}\Tsnr^{-1/4-\delta/2}.\nonumber
\end{align}
Let $\mathcal{S}_b$ be the set of all values of $g\in\mathcal{G}_b$ such that the inequality
\begin{align}
|qg-a_3|<\Phi_b \label{DiophIneq3}
\end{align}
has at least one solution with $0<|q|\leq q_{\text{max},b}$ and $a_3\in\mathbb{Z}$. Note that since $q=\tilde{\beta}-\varphi$ and we assume in this interval that $1/2<|\tilde{\beta}|\leq\sqrt{|g|}\Tsnr^{1/4-\delta/2}/2$, we have $$|q|<\sqrt{|g|}\Tsnr^{1/4-\delta/2}.$$ Thus, for all $g\in\mathcal{G}_b$ and $\tilde{\beta}$ in the considered interval, we have
$$|q|<q_{\text{max},b}.$$
Let $\bar{\mathcal{S}}_b=\mathcal{G}_b\setminus\mathcal{S}_b$. Using~\eqref{phibound2}, we have that for all $g\in\bar{\mathcal{S}}_b$ and $\tilde{\beta}$ in the considered interval
\begin{align}
\sigma^2_{\text{HK}}&\geq\frac{1}{2}\Phi_b^2\Tsnr\nonumber\\
&\geq \frac{1}{2}\frac{\Tsnr^{1/2-\delta}}{2^{-b+1}}\nonumber\\
&\geq \frac{1}{4}\frac{\Tsnr^{1/2-\delta}}{\sqrt{g^2}}.\label{SigmaGammaBound2}
\end{align}
The condition~\eqref{DiophIneq3}, which defines the set $\mathcal{S}_b$, can be written equivalently as
\begin{align}
&|q \cdot 2^b g-2^b a_3|<2^b\Phi_b.\label{DiophIneq4}
\end{align}
Define $\tilde{g}=2^b\cdot g$, and note that for all $g\in\mathcal{G}_b$ we have $\tilde{g}\in[1,2)$. With this notation,~\eqref{DiophIneq4} becomes
\begin{align}
\left|\tilde{g}-\frac{2^b a_3}{q}\right|<2^b\frac{\Phi_b}{q}.\label{DiophIneq5}
\end{align}
Define
\begin{align}
\mathcal{T}_b(q)&=\Bigg[\Bigg\{\frac{0}{q},\frac{1\cdot 2^b}{q},\frac{2\cdot 2^b}{q},\hdots,\frac{\left\lfloor\frac{2q-1}{2^b} \right\rfloor\cdot 2^b}{q}\Bigg\}\nonumber\\
&~~~~~~~~~~~~~~~~~~~~~~+2^b\frac{\Phi_b}{q}{\mathcal{I}}\Bigg]\bmod[0,2),\label{SetDef1}
\end{align}
where $\mathcal{I}=[-1,1)$ and the sum in~\eqref{SetDef1} is a Minkowski sum. It is easy to verify that
\begin{align}
\mathcal{S}_b&\subseteq 2^{-b}\bigcup_{q=1}^{\lfloor q_{\text{max}}\rfloor}\mathcal{T}_b(q).
\end{align}
Since $\left\lfloor\frac{2q-1}{2^b} \right\rfloor=0$ for all $0<q<2^{b-1}$, for all values of $q$ in this range we have
\begin{align}
\mathcal{T}_b(q)&=\left[2^b\frac{\Phi_b}{q}{\mathcal{I}}\right]\bmod[0,2)\nonumber\\
&\subseteq\left[2^b\Phi_b{\mathcal{I}}\right]\bmod[0,2)\nonumber\\
&=\mathcal{T}_b(1).
\end{align}
Therefore,
\begin{align}
\mathcal{S}_b&\subseteq 2^{-b}\left(\left(\bigcup_{q=1}^{2^{b-1}-1}\mathcal{T}_b(q)\right)\cup\left(\bigcup_{q=2^{b-1}}^{\lfloor q_{\text{max}}\rfloor}\mathcal{T}_b(q)\right)\right)\nonumber\\
&=2^{-b}\left(\mathcal{T}_b(1)\cup\left(\bigcup_{q=2^{b-1}}^{\lfloor q_{\text{max}}\rfloor}\mathcal{T}_b(q)\right)\right).
\end{align}
The Lebesgue measure of $\mathcal{S}_b$ is bounded by
\begin{align}
\mu(\mathcal{S}_b)&=\Vol\left(\mathcal{S}_b\right)\nonumber\\
&\leq2^{-b}\left(\Vol\left(\mathcal{T}_b(1)\right)+\sum_{q=2^{b-1}}^{\lfloor q_{\text{max},b} \rfloor}\Vol\left(\mathcal{T}_b(q)\right)\right)\nonumber
\end{align}
and hence,
\begin{align}
\mu(\mathcal{S}_b)&\leq 2^{-b}\left(2\cdot2^b\Phi_b+\sum_{q=2^{b-1}}^{\lfloor q_{\text{max},b} \rfloor}\left\lceil\frac{2q}{2^b} \right\rceil\cdot2\cdot 2^b\frac{\Phi_b}{q}\right)\nonumber\\
&\leq 2\Phi_b+2\Phi_b\sum_{q=2^{b-1}}^{\lfloor q_{\text{max},b} \rfloor}2\frac{2q}{2^b} \frac{1}{q}\nonumber\\
&\leq 2\Phi_b+8\cdot2^{-b}\Phi_b q_{\text{max},b}\nonumber\\
&= 2\Phi_b+8\cdot2^{-b}\Tsnr^{-\delta}\nonumber\\
&= \sqrt{2}\cdot2^{b/2}\Tsnr^{-1/4-\delta/2}+8\cdot2^{-b}\Tsnr^{-\delta}.\label{measure1}
\end{align}
We can now upper bound the measure of the outage set
\begin{align}
\mathcal{S}=\bigcup_{b=1}^{\lceil1/6\log(\Tsnr)\rceil}\mathcal{S}_b,\nonumber
\end{align}
of all values of $\Tsnr^{-1/6}\leq g<1$ for which~\eqref{SigmaGammaBound2} does not necessarily hold, as
\begin{align}
\mu(\mathcal{S})&=\sum_{b=1}^{\lceil1/6\log(\Tsnr)\rceil}\mu(\mathcal{S}_b)\nonumber\\
&<\sqrt{2}\Tsnr^{-1/4-\delta/2}\sum_{b=1}^{\lfloor1/6\log(\Tsnr)\rfloor+1}(\sqrt{2})^{b}\nonumber\\
&~~~+~8\Tsnr^{-\delta}\sum_{b=1}^{\lfloor1/6\log(\Tsnr)\rfloor+1}2^{-b}.\nonumber
\end{align}
Using the identity
\begin{align}
\sum_{b=1}^B \rho^b=\frac{\rho}{\rho-1}(\rho^B-1),\nonumber
\end{align}
which is valid for all $\rho\neq 1$, and the fact that $\sum_{b=1}^{\infty}2^{-b}<1$, we have
\begin{align}
\mu(\mathcal{S})&<\sqrt{2}\Tsnr^{-1/4-\delta/2}\frac{\sqrt{2}}{\sqrt{2}-1}\sqrt{2}\Tsnr^{1/12}+8\Tsnr^{-\delta}\nonumber\\
&<7\Tsnr^{-\delta/2}+8\Tsnr^{-\delta}\nonumber\\
&<16\Tsnr^{-\delta/2}.\label{measure2}
\end{align}
Substituting $\delta=(2c+8)/\log(\Tsnr)$ into~\eqref{SigmaGammaBound2} and~\eqref{measure2}, we see that in the interval $1/2<|\tilde{\beta}|\leq \sqrt{|g|}\Tsnr^{1/4-\delta/2}/2$ for all values of $\Tsnr^{-1/3}<|g|\leq 1$ except for an outage set with measure not greater than $2^{-c}$ we have
\begin{align}
\sigma^2_{\text{HK}}>\frac{2^{-2c}}{4\cdot2^{8}}\frac{\Tsnr^{1/2}}{\sqrt{g^2}}.\nonumber
\end{align}

\vspace{1mm}

\underline{Interval $3$ : $\sqrt{|g|}\Tsnr^{1/4-\delta/2}/2<|\tilde{\beta}|\leq \Tsnr^{1/4}/\sqrt{8|g|}$}

Since $\Tsnr^{-1/3}\leq g^2<1$ and we assumed $\Tsnr>4$, we have
\begin{align}
g^2\Tsnr-1>\frac{g^2\Tsnr}{2}.\label{g2SNRbound}
\end{align}
Note that~\eqref{g2SNRbound} continues to hold for all $g^2>\Tsnr^{-1/2}$. This will be useful in the weak interference regime.
For all values of $|\tilde{\beta}|$ in this interval
\begin{align}
\left|\frac{\tilde{\beta}}{\sqrt{g^2\Tsnr-1}}\right|&\leq \frac{\Tsnr^{1/4}/\sqrt{8|g|}}{\sqrt{g^2\Tsnr-1}}\nonumber\\
&<\frac{\Tsnr^{1/4}}{\sqrt{8|g|\cdot\frac{g^2\Tsnr}{2}}}\nonumber\\
&\leq\frac{1}{2}|g|^{-3/2}\Tsnr^{-1/4}\nonumber\\
&\leq\frac{1}{2},\nonumber
\end{align}
and hence, using~\eqref{integerChoices}, the optimal value of $a_2$ is
\begin{align}
a_2=\left\lfloor\tilde{\beta}/\sqrt{g^2\Tsnr-1}\right\rceil=0.\nonumber
\end{align}
Therefore, using the fact that $\delta=(2c+8)/\log(\Tsnr)$, we can upper bound~\eqref{sigmaHK2} as
\begin{align}
\sigma^2_{\text{HK}}\geq \frac{\tilde{\beta}^2\Tsnr}{g^2\Tsnr-1}\geq \frac{2^{-2c}}{4\cdot 2^8}\frac{\Tsnr^{1/2}}{\sqrt{g^2}}.\nonumber
\end{align}

\underline{Interval $4$ : $\Tsnr^{1/4}/\sqrt{8|g|}<|\tilde{\beta}|$}

In this interval,
\begin{align}
\sigma^2_{\text{HK}}\geq K\tilde{\beta}^2\geq \frac{1}{8}\frac{\Tsnr^{1/2}}{\sqrt{g^2}}.\nonumber
\end{align}

\subsection{Weak Interference Regime}
\label{app:boundsWeak}

We upper bound $\sigma^2_{\text{HK}}$ for all values of  $\tilde{\beta}$ within each of the four intervals. Recall that in this regime $\Tsnr^{-1/2}\leq g^2<\Tsnr^{-1/3}$.

\vspace{1mm}

\underline{Interval $1$ : $0<|\tilde{\beta}|\leq 1/2$}

As $a_1\neq 0$, in this interval $(\tilde{\beta}-a_1)^2>1/4$. Thus,
\begin{align}
\sigma^2_{\text{HK}}\geq\frac{\Tsnr}{4}.\nonumber
\end{align}

\vspace{1mm}

\underline{Interval $2$ : $1/2<|\tilde{\beta}|\leq 1/(2|g|)$}

In this interval $a_3=\lfloor \tilde{\beta} g \rceil=0$. Thus,
\begin{align}
\sigma^2_{\text{HK}}\geq (\tilde{\beta} g)^2\Tsnr\geq \frac{g^2\Tsnr}{4}.\nonumber
\end{align}

\vspace{1mm}

\underline{Interval $3$ : $1/(2|g|)<|\tilde{\beta}|\leq \sqrt{g^2\Tsnr/8}$}

Under our assumption that $\Tsnr>4$, for all values of $|\tilde{\beta}|$ in this interval we have
\begin{align}
\left|\frac{\tilde{\beta}}{\sqrt{g^2\Tsnr-1}}\right|&\leq \frac{\sqrt{g^2\Tsnr/8}}{\sqrt{g^2\Tsnr-1}}\nonumber\\
&< \frac{\sqrt{g^2\Tsnr/8}}{\sqrt{g^2\Tsnr/2}}\nonumber\\
&\leq \frac{1}{2},\nonumber
\end{align}
where the second inequality follows from~\eqref{g2SNRbound}.
Thus, the optimal choice for $a_2$ is
\begin{align}
a_2=\left\lfloor\tilde{\beta}/\sqrt{g^2\Tsnr-1}\right\rceil=0.
\end{align}
Therefore,~\eqref{sigmaHK2} can be lower bounded by
\begin{align}
\sigma^2_{\text{HK}}\geq \frac{\tilde{\beta}^2}{g^2\Tsnr}\Tsnr\geq \frac{1}{4 g^4}.\nonumber
\end{align}

\vspace{1mm}

\underline{Interval $4$ : $\sqrt{g^2\Tsnr/8}<|\tilde{\beta}|$}

In this interval,
\begin{align}
\sigma^2_{\text{HK}}\geq K\tilde{\beta}^2>\frac{g^2\Tsnr}{4}.\nonumber
\end{align}

\section{Proof of Theorem~\ref{thm:DoF}}
\label{app:DoFproof}
For the proof we will need a key result from the field of metric Diophantine approximation which is due to Kleinbock and Margulis. The following theorem is a special case of~\cite[Theorem A]{km98}.

\vspace{2mm}

\begin{theorem}
\label{thm:margulis}
Let $\mathcal{U}$ be a domain in $\RR^m$ and let $f_1,f_2,\cdots,f_K$ be real analytic functions in $\mathbf{\tilde{h}}\in \mathcal{U}$, which together with $1$ are linearly independent over $\RR$, and define the manifold
$$\mathcal{M}=\Big\{\big[f_1(\mathbf{\tilde{h}}) \ \cdots \ f_K(\mathbf{\tilde{h}})\big] \ : \ \mathbf{\tilde{h}}\in \mathcal{U}\Big\}.$$
For almost every $\bh\in\mathcal{M}$ and any $\delta>0$, the inequality
\begin{align}
\max_{\ell=1,\ldots,K}\left|q h_\ell-a_\ell\right|\leq |q|^{-\frac{1}{K}-\delta}
\label{DiophantineIneq}
\end{align}
has at most finitely many solutions $(q,\ba)\in\mathbb{Z}\times\mathbb{Z}^K$.
\end{theorem}

\vspace{2mm}

For the proof of Theorem~\ref{thm:DoF} we will need a corollary of Theorem~\ref{thm:margulis}.

\vspace{2mm}

\begin{corollary}
\label{cor:DiophSubset}
Let $f_1,f_2,\ldots,f_K$ be functions from $\RR^m$ to $\RR$ satisfying the following conditions:
\begin{enumerate}
\item $f_i$ for $i=1,\ldots,K$ is analytic in $\RR^m$,
\item $1,f_1,\ldots,f_K$ are linearly independent over $\RR$.
\end{enumerate}
Let $\mathcal{D}=\big\{\mathbf{\tilde{h}}\in\RR^m \ : \ f_1(\mathbf{\tilde{h}})=0\big\}$
and ${\mathcal{D}}(\epsilon)=\mathcal{D}+\mathcal{B}(\mathbf{0},\epsilon)$,
where the sum is a Minkowski sum and $\mathcal{B}(\mathbf{0},\epsilon)$ is an $m$-dimensional closed ball with some radius $\epsilon>0$. Define the set $\mathcal{U}(\epsilon)=\RR^m\setminus{\mathcal{D}}(\epsilon)$, the set of functions $\tilde{f}_k(\mathbf{\tilde{h}})=f_k(\mathbf{\tilde{h}})/f_1(\mathbf{\tilde{h}})$ from $\mathcal{U}(\epsilon)$ to $\RR$ for $k=2,\ldots,K$, and the manifold
\begin{align}
\bar{\mathcal{M}}(\epsilon)=\Big\{\big[\tilde{f}_2(\mathbf{\tilde{h}}) \ \cdots \ \tilde{f}_K(\mathbf{\tilde{h}})\big] \ : \ \mathbf{\tilde{h}}\in \mathcal{U}(\epsilon)\Big\}.\label{ManifoldKminus1}
\end{align}
For all $\epsilon>0$, almost every $\bar{\bh}\in\bar{\mathcal{M}}(\epsilon)$, and any $\delta>0$ the inequality
\begin{align}
\max_{\ell=1,\ldots,K-1}\left|q \bar{h}_\ell-a_\ell\right|\leq |q|^{-\frac{1}{K-1}-\delta}
\label{DiophantineIneq1}
\end{align}
has at most finitely many solutions $(q,\ba)\in\mathbb{Z}\times\mathbb{Z}^{K-1}$.
\end{corollary}

\vspace{2mm}

\begin{proof}[Proof of Corollary~\ref{cor:DiophSubset}]
We would like to apply Theorem~\ref{thm:margulis} for the set of functions $\tilde{f}_2,\ldots,\tilde{f}_K$ from $\mathcal{U}(\epsilon)$ to $\RR$. To that end we have to show that for all  $\epsilon>0$ the functions $\tilde{f}_2,\ldots,\tilde{f}_K$ are analytic in $\mathcal{U}(\epsilon)$ and together with $1$ are linearly independent over $\RR$.

The reciprocal of an analytic function that is nowhere zero is analytic. Thus, for any $\epsilon>0$, the function $1/f_1(\mathbf{\tilde{h}})$ is analytic in $\mathcal{U}(\epsilon)$. Furthermore, the product of two analytic functions is analytic. Therefore, for any $\epsilon>0$, the functions ${\tilde{f}_k=f_k(\mathbf{\tilde{h}})\cdot (1/f_1(\mathbf{\tilde{h}}))}$ are analytic in $\mathcal{U}(\epsilon)$ for $k=2,\ldots,K$.

We show that the functions $1,\tilde{f}_2,\ldots,\tilde{f}_K$ from $\mathcal{U}(\epsilon)$ to $\RR$ are linearly independent for all $\epsilon>0$ by contradiction.
Assume they are linearly dependent. Thus, there a exists a measurable set $\mathcal{S}\in\mathcal{U}(\epsilon)$ and a set of coefficients $\{t_1(\epsilon),\ldots,t_K(\epsilon)\}\in\RR$ not all zero such that $\forall \mathbf{\tilde{h}}\in\mathcal{S}$
\begin{align}
t_1(\epsilon)\cdot 1+t_2(\epsilon)\cdot\frac{f_2(\mathbf{\tilde{h}})}{f_1(\mathbf{\tilde{h}})}+\cdots+t_K(\epsilon)\cdot\frac{f_K(\mathbf{\tilde{h}})}{f_1(\mathbf{\tilde{h}})}=0.\nonumber
\end{align}
This implies that $\forall \mathbf{\tilde{h}}\in\mathcal{S}$
\begin{align}
0\cdot 1+t_1(\epsilon)\cdot f_1(\mathbf{\tilde{h}})+t_2(\epsilon)\cdot f_2(\mathbf{\tilde{h}})+\cdots+&t_K(\epsilon)\cdot f_K(\mathbf{\tilde{h}})=0, \nonumber
\label{linearCombFunctions}
\end{align}
in contradiction to the assumption that the functions $1,f_1,\ldots,f_K$ from $\RR^m$ to $\RR$ are linearly independent over $\RR$.

We can therefore apply Theorem~\ref{thm:margulis} to the set of functions $\tilde{f}_2,\ldots,\tilde{f}_K$ from $\mathcal{U}(\epsilon)$ to $\RR$ for all $\epsilon>0$, and the corollary follows.
\end{proof}

\vspace{2mm}

We are now ready to prove Theorem~\ref{thm:DoF}.
Define the sets $\mathcal{D}$, ${\mathcal{D}}(\epsilon)$ and $\mathcal{U}(\epsilon)$ as in Corollary~\ref{cor:DiophSubset}, and the manifold
$$\tilde{\mathcal{M}}(\epsilon)=\Big\{\big[f_1(\mathbf{\tilde{h}}) \ \cdots \ f_K(\mathbf{\tilde{h}})\big] \ : \ \mathbf{\tilde{h}}\in\mathcal{U}(\epsilon)\Big\}.$$
We begin by showing that, for any $\epsilon>0$, the DoF offered by the first computation rate is upper bounded by
\begin{align}
d_{\text{comp},1}\leq \frac{1}{K}
\end{align}
for almost every $\bh\in\tilde{\mathcal{M}}(\epsilon)$. Then we take $\epsilon$ to zero in order to show that this holds for almost every
$$\bh\in\mathcal{M}=\Big\{\big[f_1(\mathbf{\tilde{h}}) \ \cdots \ f_K(\mathbf{\tilde{h}})\big] \ : \ \mathbf{\tilde{h}}\in\RR^m\Big\}.$$

Consider the manifold $\tilde{\mathcal{M}}(\epsilon)$ for some $\epsilon>0$. Note that $h_1=f_1(\mathbf{\tilde{h}})\neq 0$ for any $\bh\in\tilde{\mathcal{M}}(\epsilon)$, and we can therefore define $\bar{\bh}=\bh/h_1$. We have $\bar{h}_1=1$ and $[\bar{h}_2 \ \cdots \ \bar{h}_K]\in\bar{\mathcal{M}}(\epsilon)$, where $\bar{\mathcal{M}}(\epsilon)$ is the manifold from~\eqref{ManifoldKminus1} in Corollary~\ref{cor:DiophSubset}.

The channel~\eqref{MACchannel} is equivalent to the channel
\begin{align}
\bar{\by}=\frac{1}{h_1}\by=\bx_1+\sum_{k\neq 1}\bar{h}_k \bx_k+\frac{1}{h_1}\bz.
\end{align}
Let $\ba$ be a vector of integer coefficients, and $\beta$ be the scaling factor used by the receiver in order to decode the linear combination $\bv=[\sum_{k=1}^K a_k\bt_k]\Mod$, see Section~\ref{s:cf}. The effective noise encountered in decoding the linear combination $\bv$ with coefficient vector $\ba$ is
\begin{align}
\bz_{\text{eff}}(\bh,\ba,\beta)&=(\beta-a_1)\bx_1+\sum_{k\neq 1}(\beta \bar{h}_k-a_k)\bx_k +\frac{\beta}{h_1}\bz,\nonumber
\end{align}
and its effective variance is given by
\begin{align}
\sigma^2_{\text{eff}}(\bh,\ba,\beta)&=(\beta-a_1)^2\Tsnr\nonumber\\
&~~+~\sum_{k\neq 1}(\beta \bar{h}_k-a_k)^2\Tsnr+
\frac{\beta^2}{|h_1|^2}.\label{DiophSigmaBound}
\end{align}
Recall that \begin{align}
R_{\text{comp},1}&=\max_{\ba,\beta}\frac{1}{2}\log\left(\frac{\Tsnr}{\sigma^2_{\text{eff}}(\bh,\ba,\beta)}\right)\nonumber\\
&=\frac{1}{2}\log\left(\Tsnr\right)-\frac{1}{2}\log\left(\min_{\ba,\beta}\sigma^2_{\text{eff}}(\bh,\ba,\beta)\right).
\label{rateExp}
\end{align}
Thus, in order to obtain an upper bound on $R_{\text{comp},1}$ we need to lower bound $\sigma^2_{\text{eff}}(\bh,\ba,\beta)$ for all values of $\beta\in\RR$ and $\ba\in\ZZ^{K}\setminus\mathbf{0}$. Let $$\bar{h}^*=\max_{k=1,\ldots,K}\bar{h}_k,$$ and $$k^*=\argmax_{k=1,\ldots,K}\bar{h}_k.$$ Note that if $|\beta|<1/(2\bar{h}^*)$ the minimizing corresponding choice of integers $a_1,\ldots,a_K$ in~\eqref{DiophSigmaBound} is $a_{k^*}=\sign(\beta)$, and $a_k=0$ for all $k\neq k^*$. This in turn, implies that for $|\beta|<1/(2\bar{h}^*)$ we have
\begin{align}
\sigma^2_{\text{eff}}(\bh,\ba,\beta)>(\beta\bar{h}^*-\sign(\beta))^2\Tsnr>\frac{\Tsnr}{4},
\end{align}
which means $d_{\text{comp},1}=0$. Thus, in order to obtain a positive DoF, $|\beta|$ must be greater than $1/(2\bar{h}^*)$.

If $1/(2\bar{h}^*)\leq|\beta|\leq 1/2$, then the minimizing value of $a_1$ in~\eqref{DiophSigmaBound} is $a_1=0$. This implies that for all values of $1/(2\bar{h}^*)\leq|\beta|\leq 1/2$ we have
\begin{align}
\sigma^2_{\text{eff}}(\bh,\ba,\beta)>\beta^2\Tsnr>\frac{\Tsnr}{4\left(\bar{h}^*\right)^2},
\end{align}
which also means $d_{\text{comp},1}=0$. Thus, in order to obtain a positive DoF, $|\beta|$ must be greater than $1/2$.

Hence, in order to lower bound~\eqref{DiophSigmaBound} in the limit of very high $\Tsnr$, it suffices to limit the optimization space of $\beta$ to $|\beta|>1/2$. For such values, $\beta$ can be written in the form $\beta=q+{\varphi}$ where ${\varphi}\in[-1/2,1/2)$, and $q\in\mathbb{Z}\setminus 0$ is the nearest integer to $\beta$.

For any $|{\varphi}|<1/2$, $q\in\mathbb{Z}\setminus 0$ and $\ba\in\mathbb{Z}^K\setminus\mathbf{0}$ we have
\begin{align}
&\sigma^2_{\text{eff}}(\bh,\ba,q,{\varphi})\nonumber\\
&\geq (\varphi+q-a_1)^2\Tsnr\nonumber\\
&~~+~\max_{k\neq 1}\left(q \bar{h}_k-a_k+{\varphi}\bar{h}_k \right)^2\Tsnr + \frac{(q/2)^2}{|h_1|^2}\nonumber\\
&\geq {\varphi}^2\Tsnr+\max_{k\neq 1}\left(q \bar{h}_k-a_k+{\varphi}\bar{h}_k \right)^2\Tsnr + \frac{(q/2)^2}{|h_1|^2}\nonumber\\
&=\max_{k\neq 1}\bigg(\left({\varphi}^2+(q \bar{h}_k-a_k+{\varphi}\bar{h}_k)^2\right)\Tsnr + \frac{q^2}{|2h_1|^2}\bigg)\label{sigmaBound}
\end{align}
We further bound~\eqref{sigmaBound} by substituting the minimizing value of $\varphi$ for each $k\neq 1$. It follows by simple differentiation that for each $k\neq 1$ the minimum occurs at
\begin{align}
{\varphi}^*(k)&=\frac{-\bar{h}_k}{1+\bar{h}^2_k}(q \bar{h}_k-a_k).\nonumber
\end{align}
Substituting ${\varphi}^*(k)$ into~\eqref{sigmaBound} yields
\begin{align}
\sigma^2_{\text{eff}}&(\bh,\ba,q,{\varphi})\geq\max_{k\neq 1}\bigg(\frac{1}{1+\bar{h}^2_k}(q \bar{h}_k-a_k)^2\Tsnr + \frac{q^2}{|2h_1|^2}\bigg)\nonumber\\
&>\frac{1}{1+\max_{k\neq 1}\bar{h}^2_k}\cdot\max_{k\neq 1}|q \bar{h}_k-a_k|^2\Tsnr + \frac{q^2}{|2h_1|^2}\nonumber\\
&=c_0(\bh)\cdot\left(\max_{k\neq 1}|q \bar{h}_k-a_k|^2\Tsnr + q^2\right),
\label{sigmaBound3}
\end{align}
where $c_0(\bh)>0$ is some constant independent of the $\Tsnr$.

Consider the limit of $\Tsnr\rightarrow\infty$, and assume $|q|$ is upper bounded by some finite integer $q_0>0$. Then, for almost every $\bh\in\tilde{\mathcal{M}}(\epsilon)$, there exists a constant $c_1({\bh},q_0)>0$, independent of the $\Tsnr$, for which
\begin{align}
\max_{k\neq 1}|q \bar{h}_k-a_k|>c_1({\bh},q_0)\label{irrationalityEq}
\end{align}
for all $0<|q|\leq q_0$ and $\ba\in\mathbb{Z}^{K-1}$. Note that $\bh$ does not satisfy~\eqref{irrationalityEq} only if all elements of $\bar{\bh}$ are rational. Substituting~\eqref{irrationalityEq} into~\eqref{sigmaBound3} gives $\sigma^2_{\text{eff}}(\bh,\ba,q,{\varphi})>c_2({\bh},q_0)\Tsnr$ which means that the DoF is zero. Therefore, in order to get a positive DoF, $q$ must tend to infinity when the SNR tends to infinity.

Any positive integer $|q|$ can be expressed as $|q|=\Tsnr^{\gamma}$ for some $\gamma>0$. From Corollary~\ref{cor:DiophSubset} we know that for any $\epsilon,\delta>0$, almost every $\bar{\bh}\in\bar{\mathcal{M}}(\epsilon)$, and $q$ large enough, we have
\begin{align}
\max_{k\neq 1}|q \bar{h}_k-a_k|>|q|^{-\frac{1}{K-1}-\delta}=\Tsnr^{-\frac{\gamma}{K-1}-\gamma\delta}.\label{diophTerm}
\end{align}
Thus, for $|q|$ large enough and almost every $\bh\in\tilde{\mathcal{M}}(\epsilon)$, we have
\begin{align}
\sigma^2_{\text{eff}}(\bh,\ba,q,{\varphi})&\geq c^2_0(\bh)\cdot\max\bigg(\Tsnr^{1-\frac{2\gamma}{K-1}-2\gamma\delta}, \Tsnr^{2\gamma}\bigg).\label{sigmaBound4}
\end{align}
Minimizing~\eqref{sigmaBound4} with respect to $\gamma$ gives $$\gamma=\frac{K-1}{2(K+\delta K-\delta)}.$$ Hence, for all $q\in\mathbb{Z}$, ${\varphi}\in[-1/2,1/2)$, $\ba\in\mathbb{Z}^K\setminus\mathbf{0}$ and almost every $\bh\in\tilde{\mathcal{M}}(\epsilon)$
\begin{align}
\sigma^2_{\text{eff}}(\bh,\ba,q,{\varphi})>c_3(\bh)\Tsnr^{\frac{K-1}{K+\delta(K-1)}},
\end{align}
where $c_3(\bh)>0$ is also a constant independent of the $\Tsnr$.
Substituting into~\eqref{rateExp} gives
\begin{align}
R_{\text{comp},1}<\frac{1+\delta(K-1)}{K+\delta(K-1)}\cdot\frac{1}{2}\log(\Tsnr)-\frac{1}{2}\log(c_3(\bh))\label{rateExp2}
\end{align}
for any $\delta> 0$. Taking $\delta\rightarrow 0$, it follows that the DoF the highest computation rate offers is upper bounded by
\begin{align}
\lim_{\Tsnr\rightarrow\infty}\frac{R_{\text{comp},1}}{\frac{1}{2}\log{(1+\Tsnr)}}\leq\frac{1}{K},\label{DoFBound}
\end{align}
for almost every $\bh\in\tilde{\mathcal{M}}(\epsilon)$.
Since this holds for all $\epsilon>0$, we can now take $\epsilon$ to zero (note that the bound does not depend on $\epsilon$).
The set $\mathcal{D}$ has measure zero since $f_1$ is analytic on $\RR^m$ and is not identically zero (otherwise, the set of functions $1,f_1,\ldots,f_K$ is not linearly independent). Note that the measure of $\mathcal{D}(\epsilon)$ goes to zero as $\epsilon\rightarrow 0$, and furthermore $\mathcal{D}=\cap_{\epsilon>0}\mathcal{D}(\epsilon)$. Therefore, the claim holds for almost every $\bh\in\mathcal{M}$.

\section{Proof of Theorem~\ref{thm:DoFeff}}
\label{app:DoFeffProof}
Consider the reference $L$-user MAC
\begin{align}
\by_{\text{ref}}=\sum_{\ell=1}^L g_\ell\bx_\ell+\bz,\label{LmacRef}
\end{align}
where $\bz$ is AWGN with zero mean and unit variance and all users are subject to the power constraint $\|\bx_\ell\|^2\leq n\Tsnr$.
Applying Corollary~\ref{thm:symmetricDoF} to this channel implies that, for almost every $\bg\in\mathcal{M}$, the DoF that each optimal computation rate offers is $1/L$.
Let $R^{\text{ref}}_{\text{comp}}(\bg,\ba)$ be the computation rate corresponding to the coefficient vector $\ba$ over the reference MAC~\eqref{LmacRef}. We now show the computation rate of the same coefficient vector $R_{\text{comp}}(\bg,\ba,\bB)$ over the effective MAC~\eqref{effectiveMAC} is within a constant number of bits from $R^{\text{ref}}_{\text{comp}}(\bg,\ba)$.

For the reference channel~\eqref{LmacRef} the effective noise variance for a given $\ba$ and $\beta$ is
\begin{align}
\sigma^2_{\text{ref}}(\bg,\ba,\beta)=\Tsnr \|\beta \bg -\ba\|^2+\beta^2,\nonumber
\end{align}
while for the effective $L$-user MAC~\eqref{effectiveMAC} the effective variance for the same $\ba$ and $\beta$ is
\begin{align}
\sigma^2_{\text{eff}}(\bg,\ba,\beta,\bB)=\Tsnr\sum_{\ell=1}^L(\beta g_\ell-a_\ell)^2b^2_{\text{eff},\ell}+\beta^2.\nonumber
\end{align}
Letting $b^*=\max_{\ell=1,\ldots,L}b^2_{\text{eff},\ell}$ and noting that $b^*\geq1$ gives
\begin{align}
\sigma^2_{\text{ref}}(\bg,\ba,\beta) \leq \sigma^2_{\text{eff}}(\bg,\ba,\beta,\bB) \leq b^* \sigma^2_{\text{ref}}(\bg,\ba,\beta).\nonumber
\end{align}
Since the above inequalities are valid for any value of $\beta$, in particular they hold true for the optimal value of $\beta$ and it follows that
\begin{align}
R^{\text{ref}}_{\text{comp}}(\bg,\ba)-\frac{1}{2}\log(b^*)\leq R_{\text{comp}}(\bg,\ba,\bB)\leq R^{\text{ref}}_{\text{comp}}(\bg,\ba).\nonumber
\end{align}
As $b^*$ is independent of the $\Tsnr$, it follows that the DoF offered by each computation rate for the reference and effective MACs~\eqref{LmacRef} and~\eqref{effectiveMAC} are equal, In particular, this is the case for the optimal computation rates, thus the theorem follows.

\end{appendices}

\bibliographystyle{IEEEtran}
\bibliography{symICbib}
\end{document}